\documentclass{siamltex}
\usepackage{amssymb,mathrsfs,amsmath,paralist}
\usepackage[raiselinks=false,colorlinks=true,citecolor=blue,urlcolor=blue,linkcolor=blue,bookmarksopen=true]{hyperref}
\usepackage[usenames,dvipsnames]{xcolor}
\usepackage{bm}
\usepackage[normalem]{ulem}
\usepackage[all,dvips,color,arrow,curve,poly]{xy}
\UseCrayolaColors
\makeatletter
 \newcommand\setreflabel[1]{\protected@edef\@currentlabel{#1}}\newcounter{claim}
 
 \newcommand\claimlabelfmt[1]{(#1)}
   {\removelastskip\ifx\newenvironment#1\newenvironment%
    \refstepcounter{claim}\def\@claim{\arabic{claim}}\else\def\@claim{#1}\fi%
    \setreflabel{\expandafter\claimlabelfmt{\@claim}}%
    \global\edef\@lastclaim{clm@\the\inputlineno}%
    \label{\@lastclaim}\smallskip\par%
    \begin{compactitem}[~{\@currentlabel}]\item\it}
   {\end{compactitem}\smallskip\par}
   {\par\edef\@thisclaim{\ifx\newenvironment#1\newenvironment%
    \@lastclaim\else#1\fi}\noindent\ignorespaces}
   {This proves~\expandafter\ref{\@thisclaim}.\smallskip\par\relax}
   {\let\@oldclaim\@thisclaim}
   {\let\@thisclaim\@oldclaim}
 \makeatother
\newcommand{\Eq}{Eq}

\DeclareMathOperator{\Bool}{Boole}

\DeclareMathOperator{\NNE}{\mathit{NNE}}
\DeclareMathOperator{\NEN}{\mathit{NEN}}
\DeclareMathOperator{\ENN}{\mathit{ENN}}
\DeclareMathOperator{\NNN}{\mathit{NNN}}
\DeclareMathOperator{\EEE}{\mathit{EEE}}
\DeclareMathOperator{\EEN}{\mathit{EEN}}
\DeclareMathOperator{\ENE}{\mathit{ENE}}
\DeclareMathOperator{\NEE}{\mathit{NEE}}
\DeclareMathOperator{\EE}{\mathit{EE}}
\DeclareMathOperator{\NN}{\mathit{NN}}
\DeclareMathOperator{\EN}{\mathit{EN}}
\DeclareMathOperator{\QN}{\mathit{QN}}
\DeclareMathOperator{\NPe}{\mathit{NP}}
\DeclareMathOperator{\NE}{\mathit{NE}}

\DeclareMathOperator{\Eeq}{\mathit{E{=}}}
\DeclareMathOperator{\eqE}{\mathit{{=}E}}
\DeclareMathOperator{\Neq}{\mathit{N{=}}}
\DeclareMathOperator{\eqN}{\mathit{{=}N}}

\DeclareMathOperator{\ENeq}{\mathit{EN}{=}}
\DeclareMathOperator{\eqEE}{\mathit{{=}EE}}
\DeclareMathOperator{\eqEN}{\mathit{{=}EN}}
\DeclareMathOperator{\eqNE}{\mathit{{=}NE}}
\DeclareMathOperator{\eqNN}{\mathit{{=}NN}}
\DeclareMathOperator{\eqeqE}{\mathit{{=}{=}E}}
\DeclareMathOperator{\eqeqN}{\mathit{{=}{=}N}}
\DeclareMathOperator{\eqeqeq}{\mathit{{=}{=}{=}}}

\DeclareMathOperator{\EEQ}{\mathit{E}{=}}

\DeclareMathOperator{\NEQ}{\mathit{N}{=}}
\DeclareMathOperator{\NEQEQ}{{\neq}{=}}
\DeclareMathOperator{\NEQNEQ}{\neq\neq}

\newcommand{\nin}{\notin}
\DeclareMathOperator{\Aut}{Aut}

\newcommand{\ignore}[1]{}

\newcommand{\CSP}{\text{\rm CSP}}

\newcommand{\PP}{\ensuremath{\mathrm{P}}}
\newcommand{\NP}{\ensuremath{\mathrm{NP}}}

\DeclareMathOperator{\inj}{inj}

\DeclareMathOperator{\Pol}{Pol}
\DeclareMathOperator{\End}{End}

\DeclareMathOperator{\Csp}{CSP}

\DeclareMathOperator{\tp}{tp}
\DeclareMathOperator{\maxi}{max}
\DeclareMathOperator{\mini}{min}
\newcommand{\To}{\rightarrow}

\newtheorem{thm}[theorem]{Theorem}

\newtheorem{lem}[theorem]{Lemma}

\newtheorem{prop}[theorem]{Proposition}

\newtheorem{cor}[theorem]{Corollary}
\newtheorem{defn}[theorem]{Definition}

\title{Constraint satisfaction problems for reducts of homogeneous graphs\footnotemark[1]}
\author{Manuel Bodirsky  \footnotemark[2]\ \footnotemark[6] \and Barnaby Martin  \footnotemark[3] \and Michael Pinsker \footnotemark[4]\  \footnotemark[7] \and
 Andr\'{a}s Pongr\'{a}cz \footnotemark[5]\ \footnotemark[8]
}

\begin{document}
\maketitle

\renewcommand{\thefootnote}{\fnsymbol{footnote}}

\footnotetext[1]{An extended abstract of this paper has appeared at the 43rd International Colloquium on Automata, Languages and Programming (ICALP) Track B, 2016.}
\footnotetext[2]{Institut f\"{u}r Algebra, TU Dresden, 01062 Dresden, Germany. (\texttt{Manuel.Bodirsky@tu-dresden.de}).}
\footnotetext[3]{Department of Computer Science, Durham University, South Road, Durham, UK.  (\texttt{barnabymartin@gmail.com}).}
\footnotetext[4]{Institut f\"{u}r Diskrete Mathematik und Geometrie, FG Algebra, TU Wien, Austria, and Department of Algebra, Charles University, Czech Republic. (\texttt{marula@gmx.at}).}
\footnotetext[5]{Department of Algebra and Number Theory,  University of Debrecen,
    4032 Debrecen, Egyetem square 1,
    Hungary.
(\texttt{pongracz.andras@science.unideb.hu}).}
\footnotetext[6]{The first and fourth author have received funding from the European Research Council under the European Community's Seventh Framework Programme (FP7/2007-2013 Grant Agreement no. 257039). The first author also received funding from the German Science Foundation (DFG, project number 622397) and from the European Research Council (Grant Agreement no. 681988, CSP-Infinity).}
    \footnotetext[7]{The third author has received funding from the  Austrian Science Fund (FWF) through  project No P27600, and from the Czech Science Foundation (grant No 18-20123S).}
\footnotetext[8]{The fourth author was supported by the Hungarian Scientific Research Fund (OTKA) grant no.~K109185, by the National Research, Development and Innovation Fund of Hungary, financed under the FK~124814 and PD~125160 funding schemes, the J\'anos Bolyai Research Scholarship of the Hungarian Academy of Sciences, the New National Excellence Program Bolyai+ of the Ministry of Human Capacities and by the European Social Fund (EFOP-3.6.2-16-2017-00015)}.

\begin{abstract}
For $n\geq 3$, let $(H_n, E)$ denote the $n$-th Henson graph, i.e., the unique countable homogeneous graph with exactly those finite graphs as induced subgraphs that do not embed the complete graph on $n$ vertices. 
We show that for all structures $\Gamma$ with domain
$H_n$ whose relations are first-order definable in $(H_n,E)$ 
the constraint satisfaction problem for $\Gamma$ 
 is either in $\PP$ or is $\NP$-complete. 
 
 We moreover show a similar complexity dichotomy for all structures whose relations are first-order definable in a homogeneous graph whose reflexive closure is an equivalence relation.
 
 Together with earlier results, in particular for the random graph,  this completes the complexity classification of constraint satisfaction problems of structures first-order definable in countably infinite homogeneous graphs: all such problems are either in $\PP$ or $\NP$-complete.
\end{abstract}

\section{Introduction}

\subsection{Constraint satisfaction problems} A \emph{constraint satisfaction problem} (CSP) is a computational problem in which the input consists of a finite set of variables and a finite set of \emph{constraints}, and where the question is whether there exists a mapping from the variables to some fixed domain such that all the constraints are satisfied. We can thus see the possible constraints as relations on that fixed domain, and in an instance of the CSP, we are asked to assign domain values to the variables such that certain specified tuples of variables become elements of certain specified relations. 

When the domain is finite, and arbitrary constraints are permitted, then the CSP is NP-complete. 
However, when only constraints from a restricted set of relations on the domain are allowed in the input, there might be a polynomial-time algorithm for the CSP.  
The set of relations that is allowed to formulate the constraints in the input is often called the \emph{constraint language}. The question which constraint
languages give rise to polynomial-time solvable CSPs
has been the topic of intensive research over the past years. It was conjectured by Feder and Vardi~\cite{FederVardi} that CSPs for constraint languages over finite domains have a complexity dichotomy: they are either in P or NP-complete. Over the years, the conjecture was proved for substantial classes (for example when the domain has at most three elements~\cite{Schaefer,Bulatov} or when the constraint language contains a single binary relation without sources and sinks ~\cite{HellNesetril,BartoKozikNiven}). Various methods, combinatorial (graph-theoretic), logical,  and universal-algebraic were  brought to bear on this classification project, with many remarkable consequences. A conjectured delineation for the dichotomy was given in the algebraic language in~\cite{JBK}, and finally the conjecture, and in particular this delineation, has recently been proven to be accurate~\cite{BulatovFVConjecture,ZhukFVConjecture}.

When the domain is infinite, the complexity of the CSP can be outside NP, and even undecidable~\cite{BodirskyNesetrilJLC}.  
But for natural classes of such CSPs there is often the potential for structured classifications, and this has proved to be the case for structures first-order definable over the order 
$({\mathbb Q},<)$ of the rationals~\cite{tcsps-journal} or over the integers with successor~\cite{dCSPs2}.
Another classification of this type
has been obtained for CSPs where the
constraint language is first-order definable over the random (Rado) graph~\cite{BodPin-Schaefer},
making use of structural Ramsey theory.
This paper was titled `Schaefer's theorem for graphs' and it can be seen as lifting the famous classification of Schaefer~\cite{Schaefer} from Boolean logic to logic over finite graphs, since the random graph is universal for the class of finite graphs.

\subsection{Homogeneous graphs and their reducts}
The notion of \emph{homogeneity} 
from model theory plays 
an important role
when applying techniques from
finite-domain constraint satisfaction 
to constraint satisfaction over infinite domains. A relational structure is 
\emph{homogeneous} if every isomorphism
between finite induced substructures can be extended 
to an automorphism of the entire structure. 
Homogeneous
structures are uniquely 
(up to isomorphism) given by
the class of finite structures 
that embed into them. 
The structure $(\mathbb Q,<)$
and the random graph are among the most prominent 
examples of homogeneous structures.
The class of structures that are first-order definable over a homogeneous structure with finite relational signature is a very large generalization of the class of all finite structures,
and CSPs for those structures have
been studied independently in many different areas of theoretical computer science, e.g. in temporal and spatial reasoning, phylogenetic analysis, computational linguistics, scheduling, graph homomorphisms, and many more; see~\cite{Bodirsky-HDR} for references. 

While homogeneous relational structures are abundant, there are remarkably few countably infinite homogeneous (undirected, irreflexive) \emph{graphs}; they have been classified by Lachlan and Woodrow~\cite{LachlanWoodrow}. 
Besides the random graph mentioned earlier, an example of such a graph is  
 the countable homogeneous \emph{universal triangle-free} graph, one of the
fundamental structures that appears in most textbooks in model theory. 
This graph is the up to isomorphism unique countable triangle-free graph $(H_3,E)$ with 
the property that for every finite independent set $X \subseteq H_3$ and
for every finite set $Y \subseteq H_3$ there exists a vertex
$x \in H_3 \setminus (X \cup Y)$ such that $x$ is adjacent to 
every vertex in $X$ and to no vertex in $Y$.

Further examples of homogeneous 
graphs are the graphs $(H_4,E)$, $(H_5,E)$, and so forth, which together with $(H_3,E)$ are called the \emph{Henson graphs}, and their complements. 
Here, $(H_n,E)$ for $n > 3$ is the generalization of the graph $(H_3,E)$ above from triangles to cliques of size $n$. Finally, the list of Lachlan and Woodrow contains only one more family of infinite graphs, namely the graphs $(C_n^s,E)$ whose reflexive closure $\Eq$ is an equivalence relation with $n$ classes of equal size $s$, where $1\leq n,s\leq \omega$ and either $n$ or $s$ equals $\omega$, as well as their complements. We remark that $(C_n^s,\Eq)$ is itself homogeneous and first-order interdefinable with $(C_n^s,E)$, and so we shall sometimes refer to the \emph{homogeneous equivalence relations}.

All countable homogeneous graphs, and even all structures which are first-order definable over homogeneous graphs, are \emph{$\omega$-categorical}, 
that is, all countable models of their first-order theory are isomorphic. 
Moreover,  
all countably infinite homogeneous graphs $\Gamma$ are
\emph{finitely bounded} in the sense
that the \emph{age} of $\Gamma$, i.e., 
the class of finite structures that embed into $\Gamma$, can be described by finitely many forbidden substructures. 
Finitely bounded homogeneous structures also share with finite structures the property of  having a finite description: up to isomorphism, they are uniquely given by the finite list of forbidden structures that describes their age.  
Recent work indicates the importance of finite boundedness for complexity classification~\cite{BPT-decidability-of-definability, BP-reductsRamsey,Bodirsky-Mottet,TwoDichotomyConjectures}, and it has been conjectured that all structures with a first-order definition in a finitely bounded homogeneous structure enjoy a complexity dichotomy, i.e., their CSP is either in P or NP-complete (cf.~\cite{BPP-projective-homomorphisms, wonderland, TwoDichotomyConjectures}). The structures first-order definable in homogeneous graphs 
therefore provide the most natural class on which to test further the methods developed in~\cite{BodPin-Schaefer} specifically for the random graph.



In this article we obtain a complete classification
of the computational complexity of CSPs where all constraints have a first-order definition in one of the Henson graphs. We moreover obtain such a classification for CSPs where all constraints have a first-order definition in a countably infinite homogeneous graph whose reflexive closure is an equivalence relation, expanding earlier results for the special cases of one single equivalence class (so-called \emph{equality constraints}~\cite{ecsps}) and infinitely many 
infinite classes~\cite{equiv-csps}. Together with the above-mentioned result on the random graph, this completes the classification of  CSPs for constraints with a first-order definition in any countably infinite homogeneous graph, by Lachlan and Woodrow's classification. Our result is in accordance with the delineations between tractability and hardness predicted in general for structures with a first-order definition in a finitely bounded homogeneous structure~\cite{BPP-projective-homomorphisms, wonderland, TwoDichotomyConjectures}.


Following an established convention (e.g.,~\cite{RandomReducts, BP-reductsRamsey}, and many more) we call a relational structure $\Gamma$ a \emph{reduct} of a structure $\Delta$ if it has the same domain as $\Delta$ and all relations of $\Gamma$ are first-order definable without parameters in $\Delta$. That is, for us a reduct of $\Delta$ is as the classical definition of a reduct with the difference that we first allow a first-order expansion of $\Delta$. With this terminology, the present article provides a complexity classification of the CSPs for all reducts of countably infinite homogeneous graphs. In other words, for every such reduct we determine the complexity of
deciding its \emph{primitive positive theory}, which consists of all sentences which are existentially quantified conjunctions of atomic formulas and which hold in the reduct. We remark that all reducts of such graphs can be defined by quantifier-free first-order formulas, by homogeneity and $\omega$-categoricity.

For reducts of $(H_n, E)$, the 
CSPs express computational problems where the task is to decide whether
there exists a finite graph without any clique of size $n$ that meets certain constraints.
An example of a reduct whose CSP can be solved in polynomial time is $(H_n, E, \{(x,y,u,v) \;|\; E(x,y) \Rightarrow E(u,v)\})$, where $n\geq 3$ is arbitrary. As it turns out, for every CSP of a reduct of a Henson graph which is solvable in polynomial time, the corresponding reduct over the random graph, i.e., the reduct whose relations are defined by the same quantifier-free formulas, is also polynomial-time solvable. On the other hand, the CSP of the reduct $(H_n, 
\{(x,y,u,v) \;|\; E(x,y) \vee E(u,v)\})$ is NP-complete for all $n\geq 3$, but the corresponding reduct over
the random graph can be decided in polynomial time. 

Similarly, for reducts of the graph $(C_n^s,E)$ whose reflexive closure is an equivalence relation with $n$ classes of size $s$, where $1\leq n,s\leq \omega$, the computational problem is to decide whether there exists an equivalence relation with $n$ classes of size $s$ that meets certain constraints. For example, consider the structure $(C_\omega^2;\Eq,A)$
where \begin{align*}
A := \big \{(x_1,y_1,x_2,y_2,x_3,y_3) \mid & \text{ if } \Eq(x_1,y_1), \Eq(x_2,y_2) \text{ and } \Eq(x_3,y_3) \text{ then there is} \\
& \text{ an odd number of } i \in \{1,2,3\} \text{ such that }
x_i \neq y_i \big \}.
\end{align*}
This structure is a reduct of $(C_\omega^2;E)$ and it follows from our results in Section~\ref{thm:C-low-omega-high-2-P} that its CSP can be solved in polynomial time.

\subsection{Results} Our first result is the complexity classification of the CSPs of all reducts of Henson graphs, showing in particular that a uniform approach to infinitely many `base structures' in the same language (namely, the $n$-th Henson graph for each $n\geq 3$) is, in principle, possible.

\begin{theorem}\label{thm:main}
Let $n\geq 3$, and let $\Gamma$ be a finite signature reduct of the $n$-th Henson graph $(H_n, E)$. Then $\CSP(\Gamma)$ is either in $\PP$ or $\NP$-complete.
\end{theorem}

We then obtain a similar complexity dichotomy for reducts of homogeneous equivalence relations, expanding earlier results for special cases~\cite{equiv-csps, ecsps}.

\begin{theorem}\label{thm:equiv}
Let  $(C_n^s,E)$ be a graph whose reflexive closure $\Eq$ is an equivalence relation with $n$ classes of size $s$, where $1\leq n, s \leq \omega$ and one of $s$ or $n$ is $\omega$. Then for any finite signature reduct $\Gamma$ of $(C_n^s,E)$, the problem $\CSP(\Gamma)$ is either in $\PP$ or $\NP$-complete.
\end{theorem}

Together with the classification of countable homogeneous graphs, and the fact that the complexity of the CSPs of the reducts of the random graph have been classified~\cite{BodPin-Schaefer}, this completes the CSP classification of reducts of all countably infinite homogeneous graphs, confirming further instances of the open conjecture that CSPs of reducts of finitely bounded homogeneous structures are either in P or NP-complete~\cite{BPP-projective-homomorphisms,wonderland,TwoDichotomyConjectures}.

\begin{cor}\label{cor:homo}
Let $\Gamma$ be a finite signature reduct of a countably infinite homogeneous graph. Then $\CSP(\Gamma)$ is either in $\PP$ or $\NP$-complete.
\end{cor}

We are going to provide more detailed versions of Theorems~\ref{thm:main} and~\ref{thm:equiv}, which describe in particular the  delineation between the tractable and the NP-complete cases algebraically, in Sections~\ref{sect:summary_Henson} and~\ref{sect:summary_equivalence}. We would like to emphasize that our proof does not assume or use the dichotomy for CSPs of finite structures, as opposed to some other dichotomy results for CSPs of infinite  structures such as~\cite{dCSPs2}.


\subsection{The strategy} The method we employ follows broadly the method invented in~\cite{BodPin-Schaefer} 
for the corresponding classification problem where the `base structure' is the random graph. The key component of this method
is the usage of Ramsey theory (in our case, a result of Ne\v{s}et\v{r}il and R\"odl~\cite{NesetrilRoedlPartite}) and the concept of \emph{canonical functions} introduced in~\cite{RandomMinOps}. There are, however, some interesting differences and novelties that appear in the present proof, as we now shortly outline.

\subsubsection{Henson graphs}  
When studying the proofs in~\cite{BodPin-Schaefer}, one might get the impression that the complexity of the method grows with the model-theoretic complexity of the base structure, and that for the random
graph we have really reached the limits of bearableness for applying 
the Ramsey method. 

However, quite surprisingly, when we step from
the random graph to the graphs $(H_n,E)$, which are in a sense more complicated structures from a model-theoretic point of view\footnote{For example, the random graph has a \emph{simple} theory~\cite{Tent-Ziegler}, whereas the Henson graphs are among the most basic examples of structures whose theory is \emph{not} simple.},
the classification and its proof become easier again. 
It is one of the contributions of the present article to
explain the reasons behind this effect. Essentially, certain \emph{behaviours} of canonical functions (cf.~Section~\ref{sect:prelims}) existing on the random graph cannot be realised in $(H_n, E)$. For example the behaviour `$\maxi$' (cf.~Section~\ref{sect:prelims}) plays no role for the present classification, but accounts over the random graph for the tractability of, inter alia, the 4-ary relation
defined by the formula $E(x,y) \vee E(u,v)$.

Remarkably, we are able to reuse results about canonical functions over the random graph, since the calculus for composing behaviours of canonical functions is the same for any other structure with a smaller type space, and in particular the Henson graphs. Via this meta-argument we can, on numerous occasions, make statements about canonical functions over the Henson graphs which were proven earlier for the random graph, ignoring completely the actual underlying structure; even more comfortably, we can \emph{a posteriori} rule out some possibilities in those statements because of the $K_n$-freeness of the Henson graphs.
Instances of this phenomenon appear in the analysis of canonical functions in Section~\ref{prop:getbinary}.


On the other hand, along with these simplifications,  
there are also new additional difficulties that appear 
when investigating reducts of $(H_n,E)$ and 
that were not present in the classification
of reducts of the random graph, which basically stem from the lower degree of symmetry of $(H_n,E)$ compared to the random graph. For example, in expansions of Henson graphs by finitely many constants,
not all orbits induce copies of Henson graphs; the fact that the analogous statement does hold for the random graph was used extensively in~\cite{BodPin-Schaefer}, for example in the rather technical proof of Proposition~7.18 of that paper.

\subsubsection{Equivalence relations} Similarly to the situation for the equivalence relation with infinitely many infinite classes studied in~\cite{equiv-csps}, there are two interesting sources of NP-hardness for the reducts $\Gamma$ of other homogeneous equivalence relations: namely, if the equivalence relation is invariant under the polymorphisms of $\Gamma$, then the structure obtained from $\Gamma$ by factoring by the equivalence relation might have an NP-hard CSP, implying NP-hardness for the CSP of $\Gamma$ itself; or, roughly, for a fixed equivalence class the restriction of $\Gamma$ to that class might have an NP-hard CSP, again implying NP-hardness of the CSP of $\Gamma$ (assuming that $\Gamma$ is a \emph{model-complete core}, see Sections~\ref{sect:polymorphisms} and~\ref{sect:polymorphisms_equivalence}). But whereas for the equivalence relation with infinitely many infinite classes both the factor structure and the restriction to a class are again infinite structures, for the other homogeneous equivalence relations one of the two is a finite structure. This obliges us to combine results about CSPs of finite structures with those of infinite structures. As it turns out, the two-element case is, not surprisingly, different from the other finite cases and, quite surprisingly, significantly more involved than the other cases. One particularity of this case is that tractability is, for some reducts, implied by a ternary \emph{non-injective} canonical function which we obtain by our Ramsey-analysis. Among all the classification results for $\omega$-categorical structures obtained so far,
this ternary function is the first example of a non-injective canonical function leading to a maximal tractable class. The occurrence of this phenomenon is of technical interest in the quest for a proof of the CSP dichotomy conjecture for reducts of finitely bounded homogeneous structures  via a reduction  to the finite CSP dichotomy.

\subsection{Overview} We organize the remainder of this article as follows. Basic notions and definitions, as well as the fundamental facts of the method we are going to use, are provided in Section~\ref{sect:prelims}. 

Sections~\ref{sect:polymorphisms} to~\ref{sect:summary_Henson} deal with the Henson graphs: Section~\ref{sect:polymorphisms} is complexity-free and investigates the structure of reducts of Henson graphs via polymorphisms and Ramsey theory. In Section~\ref{sect:CSP}, we provide hardness and tractability proofs for different classes of reducts. Section~\ref{sect:summary_Henson} contains the proof of Theorem~\ref{thm:main}, and  
 we discuss the complexity classification in more detail, formulating in particular a tractability criterion for CSPs of reducts of Henson graphs.

We then turn to homogeneous equivalence relations in Sections~\ref{sect:polymorphisms_equivalence} to~\ref{sect:summary_equivalence}. Similarly to the Henson graphs, the first section (Section~\ref{sect:polymorphisms_equivalence})  is complexity-free and investigates the structure of reducts of homogeneous equivalence relations via polymorphisms and Ramsey theory. Section~\ref{sect:CSP_equivalence} contains the algorithms proving tractability where it applies. Finally, Section~\ref{sect:summary_equivalence} provides the proof of Theorem~\ref{thm:equiv}, and describes in detail the delineation between the tractable and the NP-complete cases.

We finish this work with further research directions in Section~\ref{sect:final}.

\section{Preliminaries}\label{sect:prelims}

\subsection{General notational conventions} We use one single symbol, namely $E$, for the edge relation of all homogeneous graphs; since we never consider several such graphs at the same time, this should not cause confusion. Moreover, we use $E$ for the symbol representing the relation $E$, for example in logical formulas. In general, we shall not distinguish between relation symbols and the relations which they denote. The binary relation $N(x,y)$ is defined by the formula $\neg E(x,y)\wedge x\neq y$. 

When $E$ is the edge relation of a homogeneous graph whose reflexive closure is an equivalence relation, then we denote this equivalence relation by $\Eq$; so $\Eq(x,y)$ is defined by the formula $E(x,y) \vee x=y$.

When $t$ is an $n$-tuple, we refer to its entries by $t_1,\ldots,t_n$. When $f \colon A \rightarrow B$ is a function and $C \subseteq A$, we write $f[C]:=\{f(a) \; | \; a \in C\}$.

\subsection{Henson graphs} For $n \geq 2$, denote the clique on $n$ vertices by $K_n$. For $n\geq 3$, the graph $(H_n, E)$ is the up to isomorphism unique countable graph which is
\begin{itemize}
\item \emph{homogeneous}: any isomorphism between two finite induced subgraphs of $(H_n, E)$ can be extended to an automorphism of $(H_n, E)$, and 
\item \emph{universal for the class of $K_n$-free graphs}: 
$(H_n, E)$ contains all finite (in fact, all countable) $K_n$-free graphs as induced subgraphs.
\end{itemize}
The graph $(H_n, E)$ has the \emph{extension property}: for all disjoint finite $U, U'\subseteq H_n$ such that $U$ is not inducing any isomorphic copy of $K_{n-1}$ in $(H_n,E)$ there exists $v\in H_n$ such that $v$ is adjacent in $(H_n, E)$ to all members of $U$ and to none in $U'$. Up to isomorphism, there exists a unique countably infinite $K_n$-free graph with this extension property, and hence the property can be used as an alternative definition of $(H_n, E)$.

\subsection{Homogeneous equivalence relations} For $1\leq n,s \leq \omega$ the graph $(C_n^s,E)$ is the up to isomorphism unique countable graph whose reflexive closure is an equivalence relation $\Eq$ with $n$ classes $C_i$, where ${0\leq i<n}$, all of which have size $s$. Clearly, $(C_n^s,E)$ is homogeneous and universal in a similar sense as above.


\subsection{Constraint satisfaction problems} For a relational signature $\tau$, a first-order $\tau$-formula is called \emph{primitive positive} (or \emph{pp} for short) if
it is of the form $$\exists x_1,\dots,x_n \, (\psi_1 \wedge \dots \wedge \psi_m)$$ where the $\psi_i$ are \emph{atomic}, i.e., of the form $y_1=y_2$ or $R(y_1,\dots,y_k)$ for a $k$-ary relation symbol 
$R \in \tau$ and not necessarily distinct variables $y_i$. 

Let $\Gamma$ be a structure with a finite relational signature $\tau$.
The \emph{constraint satisfaction problem for $\Gamma$}, denoted by $\Csp(\Gamma)$, is the computational problem of deciding for a given primitive positive (pp-) $\tau$-sentence
$\phi$ whether $\phi$ is true in $\Gamma$.
The following lemma has been first stated in~\cite{JeavonsClosure} for finite domain structures $\Gamma$ only, but the proof there also works for arbitrary infinite structures. 

\begin{lemma}\label{lem:pp-reduce}
Let $\Gamma = (D, R_1,\dots,R_\ell)$ be a relational structure,
and let $R$ be a
relation that has a primitive positive definition in $\Gamma$, i.e., a definition via a pp formula.
Then $\Csp(\Gamma)$ and
$\Csp(D, R, R_1, \dots, R_\ell)$ are polynomial-time equivalent.
\end{lemma}

When a relation $R$ has a primitive positive definition in a structure $\Gamma$, then we also say that $\Gamma$ \emph{pp-defines} $R$. Lemma~\ref{lem:pp-reduce} enables the so-called \emph{universal-algebraic approach} to constraint satisfaction, as exposed in the following.

\subsection{The universal-algebraic approach}\label{subsect:ua}
We say that a $k$-ary function 
(also called \emph{operation})
$f \colon D^k \rightarrow D$ \emph{preserves} an $m$-ary relation
$R \subseteq D^m$ if for all $t_1,\dots,t_k \in R$ the tuple
$f(t_1,\dots,t_k)$, calculated componentwise, is also contained in $R$. 
If an operation $f$ does not
preserve a relation $R$, we say that $f$ \emph{violates} $R$. We say that a set of operations preserves a relation when all of its elements do.

If $f$ preserves all relations of a structure $\Gamma$, we say that $f$ is a \emph{polymorphism} of $\Gamma$,  and that $f$ \emph{preserves} $\Gamma$. We write $\Pol(\Gamma)$ for the set of all polymorphisms of $\Gamma$. 
The unary polymorphisms of $\Gamma$
are just the \emph{endomorphisms} of $\Gamma$, and denoted by $\End(\Gamma)$.

The set of all polymorphisms $\mathrm{Pol}(\Gamma)$ of a relational structure $\Gamma$ forms an algebraic object called a \emph{function clone} (see~\cite{Szendrei}, \cite{GoldsternPinsker}), which is
a set of finitary operations defined on a fixed domain that is closed
under composition and that contains all projections.
Moreover, $\mathrm{Pol}(\Gamma)$ is closed in the \emph{topology of pointwise convergence}, i.e., an $n$-ary function $f$ is contained in $\Pol(\Gamma)$ if and only if for all finite subsets $A$ of $\Gamma^n$ there exists an $n$-ary $g\in\Pol(\Gamma)$ which agrees with $f$ on $A$. We will write $\overline{F}$ for the closure of a set $F$ of functions on a fixed domain in this topology; so $\overline{\Pol(\Gamma)}=\Pol(\Gamma)$. This closure is sometimes referred to as \emph{local closure}, and closed sets as locally closed, but we will use the terminology \emph{topologically closed} throughout this work. For an arbitrary  set $F$ of functions on a fixed domain, when $\Gamma$ is the structure whose relations are precisely those which are preserved by all functions in $F$, then $\Pol(\Gamma)$ is the smallest topologically closed function clone containing $F$ (cf.~\cite{Szendrei}).

When $\Gamma$ is a countable and $\omega$-categorical structure, then we can characterize primitive positive definable relations via $\Pol(\Gamma)$, as follows. 

\begin{theorem}[from~\cite{BodirskyNesetrilJLC}]
\label{conf:thm:inv-pol}
Let $\Gamma$ be a countable $\omega$-categorical structure.
Then the relations preserved by $\Pol(\Gamma)$ 
are precisely those having a primitive positive definition in $\Gamma$.
\end{theorem}

Theorem~\ref{conf:thm:inv-pol} and Lemma~\ref{lem:pp-reduce} imply that if two countable $\omega$-categorical structures $\Gamma, \Delta$ with finite relational signatures have the same clone of polymorphisms, then their CSPs are polynomial-time
equivalent. Moreover, if $\mathrm{Pol}(\Gamma)$ is contained in $\mathrm{Pol}(\Delta)$, then $\CSP(\Gamma)$ is, up to polynomial time, at least as hard as $\CSP(\Delta)$.


Note that the \emph{automorphisms} of a structure $\Gamma$
are just the bijective unary polymorphisms of $\Gamma$ whose inverse function is also a polymorphism of $\Gamma$; the set of all automorphisms of $\Gamma$
is denoted by $\Aut(\Gamma)$. For every reduct $\Gamma$ of a structure $\Delta$ we have that $\Pol(\Gamma)\supseteq \Aut(\Gamma)\supseteq \Aut(\Delta)$. In particular, this is the case for reducts of the homogeneous graphs $(H_n,E)$ and $(C_n^s,E)$. Conversely, it follows from the $\omega$-categoricity of homogeneous graphs $(D,E)$ (in our case, $D=H_n$ or $D=C_n^s$)  that every topologically closed function clone containing $\Aut(D,E)$ is the polymorphism clone of a reduct of $(D,E)$. 

When $(D,E)$ is a homogeneous graph, and $F$ is a set of functions and $g$ is a function on the domain $D$, then we say that $F$ \emph{generates} $g$ if $g$ is contained in the smallest topologically closed function clone which contains $F\cup \Aut(D, E)$. This is the same as saying that for every finite $S\subseteq D$, there exists a term function over $F\cup \Aut(D, E)$ which agrees with $g$ on $S$. By the discussion preceding Theorem~\ref{conf:thm:inv-pol}, this is equivalent to $g$ preserving all relations which are preserved by $F\cup \Aut(D, E)$.

We finish this section with a general lemma that we will refer to on numerous occasions; it allows us to restrict the arity of functions violating a relation. For a structure $\Gamma$ and a tuple $t \in \Gamma^k$, the \emph{orbit of $t$} in $\Gamma$ is the set
$\{ \alpha(t) \; | \; \alpha \in \Aut(\Gamma) \}$. We also call this the orbit of $t$ with respect to $\Aut(\Gamma)$.

\begin{lemma}[from~\cite{tcsps-journal}]\label{lem:arity-reduction}
    Let $\Gamma$ be a relational structure. Suppose that $R\subseteq \Gamma^k$ intersects at most $m$ orbits of $k$-tuples in $\Gamma$. If  $\Pol(\Gamma)$ contains a function violating $R$, then $\Pol(\Gamma)$ also contains an $m$-ary operation violating $R$.
\end{lemma}




\subsection{Canonical functions}
It will turn out that the polymorphisms relevant for the CSP classification show regular behaviour with respect to the underlying homogeneous graph, in a sense that we are now going to define.

\begin{definition}
    Let $\Delta$ be a structure. The \emph{type} $\tp(a)$ of an $n$-tuple $a=(a_1,\ldots,a_n)$ of elements in $\Delta$ is the set of first-order formulas with
     free variables $x_1,\dots,x_n$ that hold for $a$ in $\Delta$. For structures $\Delta_1,\ldots,\Delta_k$ and $k$-tuples $a^1,\ldots,a^n\in\Delta_1\times\cdots\times\Delta_k$, the type of $(a^1,\ldots,a^n)$ in $\Delta_1\times\cdots\times\Delta_k$, denoted by $\tp(a^1,\ldots,a^n)$, is the $k$-tuple containing the types of $(a^1_i,\ldots,a^n_i)$ in $\Delta_i$ for each $1\leq i\leq k$.
\end{definition}

We bring to the reader's attention the well-known fact that in $\omega$-categorical structures, in particular in $(H_n, E)$ and $(C_n^k,E)$, two $n$-tuples have the same type if and only if their orbits coincide.

\begin{definition}\label{defn:arbitrarilyLarge}
Let $\Delta_1,\ldots,\Delta_k$ and $\Lambda$ be structures. A \emph{behaviour} $B$ between $\Delta_1,\ldots,\Delta_k$ and $\Lambda$ is a partial function from the types over $\Delta_1,\ldots,\Delta_k$ to the types over $\Lambda$. Pairs $(s,t)$ with $B(s)=t$ are also called \emph{type conditions}. We say that a function $f \colon\Delta_1\times\cdots\times\Delta_k\To\Lambda$ \emph{satisfies the behaviour $B$} if whenever $B(s)=t$ and $(a^1,\ldots,a^n)$ has type $s$ in $\Delta_1\times \cdots \times \Delta_k$, then the $n$-tuple $(f(a^1_1,\ldots,a^1_k),\ldots,f(a^n_1,\ldots,a^n_k))$ has type $t$ in $\Lambda$. A function $f \colon\Delta_1\times\cdots\times\Delta_k\To\Lambda$ is \emph{canonical} if it satisfies a behaviour which is a total function from the types over $\Delta_1\times\cdots\times \Delta_k$ to the types over $\Lambda$.
\end{definition}

We remark that since our structures are homogeneous and have only binary relations, the type of an $n$-tuple $a$ is determined by its binary subtypes, i.e., the types of the pairs $(a_i, a_j)$, where $1\leq i,j\leq n$. In other words, the type of $a$ is determined by which of its components  are equal, and between which of its components there is an edge.
Therefore, a function $f \colon (H_n,E)^k \To (H_n,E)$ or $f \colon (C_n^s,E)^k \To (C_n^s,E)$ is canonical iff it satisfies the condition of the definition for types of 2-tuples. 

To provide immediate examples for these notions, we now define some behaviours that will appear in our proof as well as in the precise CSP classification. For $m$-ary relations
$R_1,\dots,R_k$ over a set $D$,
we will in the following write $R_1\cdots R_k$ for the $m$-ary relation on $D^k$ defined as follows: $R_1\cdots R_k(x^1,\dots,x^m)$ holds for $k$-tuples $x^1,\dots,x^m \in D^k$ if and only if $R_i(x^1_i,\dots,x^m_i)$ holds for all $1\leq i \leq k$. For example, when $p,q\in D^3$ are triples of elements in a homogeneous graph $(D,E)$, then $\ENeq(p,q)$ holds if and only if $E(p_1,q_1)$, $N(p_2,q_2)$, and $p_3=q_3$  hold in $(D,E)$.
We start with behaviours of binary injective functions $f$ on homogeneous graphs.

\begin{definition}\label{defn:behaviours_binary}
    Let $(D,E)$ be a homogeneous graph. We say that a binary injective operation $f \colon D^2\To D$ is
    \begin{itemize}
        \item \emph{balanced in the first argument} if for all $u,v\in D^2$ we have that $\EEQ(u,v)$ implies $E(f(u),f(v))$ and $\NEQ(u,v)$ implies $N(f(u),f(v))$;         
        \item \emph{balanced in the second argument} if $(x,y) \mapsto f(y,x)$ is balanced in the first argument;
        \item \emph{balanced} if $f$ is balanced in both arguments;
        \item \emph{$E$-dominated ($N$-dominated) in the first argument} if for all $u,v \in D^2$ with $\NEQEQ(u,v)$
        we have that $E(f(u),f(v))$ ($N(f(u),f(v))$);
        \item \emph{$E$-dominated ($N$-dominated) in the second argument} if
        $(x,y) \mapsto f(y,x)$ is $E$-dominated ($N$-dominated) in the first argument;
        \item \emph{$E$-dominated ($N$-dominated)} if it is $E$-dominated ($N$-dominated) in both arguments;
        \item \emph{of behaviour $\mini$} if for all $u,v\in D^2$ with $\NEQNEQ(u,v)$ we have
        $E(f(u),f(v))$ if and only if $\EE(u,v)$;
        \item \emph{of behaviour $\maxi$} if for all $u,v\in D^2$ with $\NEQNEQ(u,v)$ we have
        $N(f(u),f(v))$ if and only if $\NN(u,v)$;
        \item \emph{of behaviour $p_1$} if for all $u,v \in D^2$ with $\NEQNEQ(u,v)$ we have
        $E(f(u),f(v))$ if and only if $E(u_1,v_1)$;
        \item \emph{of behaviour $p_2$} if $(x,y) \mapsto f(y,x)$ is of behaviour $p_1$;
        \item \emph{of behaviour projection} if it is of behaviour $p_1$ or $p_2$;
        \item \emph{of behaviour xnor if for all $u,v\in D^2$ with $\NEQNEQ(u,v)$ we have
        $E(f(u),f(v))$ if and only if $\EE(u,v)$ or $\NN(u,v)$.}
    \end{itemize}
\end{definition}
Each of these properties describes the set of all functions of a certain behaviour. We explain this for the first item defining functions which are balanced in the first argument, which can be expressed  by the behaviour consisting of the following two type conditions. Let $(u,v)$ be any pair of elements $u,v\in D^2$ such that $\EEQ(u,v)$, and let $s$ be the type of the pair $(u,v)$ in $(D,E)\times (D,E)$. Let $x,y\in D$ satisfy $E(x,y)$, and let $t$ be the type of $(x,y)$ in $(D,E)$. Then the first type condition is $(s,t)$. Now 
   let $s'$ be the type in $(D,E)\times (D,E)$ of any pair $(u,v)$, where $u,v\in D^2$ satisfy $\NEQ(u,v)$,
 and let $t'$ be the type in $(D,E)$ of any $x,y \in D$ with $N(x,y)$.  The second type condition is $(s',t')$.

To justify the less obvious names of some of the above behaviours, we would like to point out that a binary injection of behaviour $\mini$ is reminiscent of the
Boolean minimum function on $\{0,1\}$, 
where $E$ takes the role of $1$
and $N$ the role of $0$: for $u,v\in H_n^2$ with $\NEQNEQ(u,v)$, we have
$E(f(u),f(v))$ if $u,v$ are connected by an edge in both coordinates, and $N(f(u),f(v))$ otherwise. The names `$\maxi$' and `projection' can be explained similarly.\smallskip

\begin{definition}~\label{defn:behaviours_ternary}
    Let $(D,E)$ be a homogeneous graph. We say that a ternary injective operation $f \colon D^3\To D$ is of behaviour
     \begin{itemize}
       \item \emph{majority} if for all $u,v\in D^3$ with ${\neq}{\neq}{\neq}(u,v)$ we have that $E(f(u),f(v))$ if and only if $\EEE(u,v)$, $\EEN(u,v)$, $\ENE(u,v)$, or $\NEE(u,v)$;

        \item \emph{minority} if for all $u,v\in D^3$ with ${\neq}{\neq}{\neq}(u,v)$ we have $E(f(u),f(v))$ if and only if $\EEE(u,v)$, $\NNE(u,v)$, $\NEN(u,v)$, or $\ENN(u,v)$.
     \end{itemize}
\end{definition}
In this article, contrary to $\mini$ and minority, neither $\maxi$ nor majority will play a role but we introduce them for the sake of completeness since they occur in \cite{BodPin-Schaefer}.

When we want to explain a type condition over a homogeneous graph $(D,E)$, we are going to express it in the form $f(R_1,\ldots,R_k)=S$ for binary relations $R_1,\ldots,R_k$ and a binary relation $S$; the meaning is that whenever $p,q\in D^k$, then $R_1\cdots R_k(p,q)$  implies $S(f(p),f(q))$. The relations we use in this notation range among $\{E,N,\Eq,\neq,=\}$. Examples of type conditions expressed this way include $f(E,N)=N$ (meaning that $\EN(p,q)$ implies $N(f(p),f(q))$, for all $p,q\in D^2$), and $f(E,=)=E$. In the latter, note that the second $=$ has different semantic content from the first.
Similarly, the majority behaviour in Definition~\ref{defn:behaviours_ternary} can be expressed by writing $f(E,E,E)=f(E,E,N)=f(E,N,E)=f(N,E,E)=E$ and $f(N,N,N)=f(E,N,N)=f(N,E,N)=f(N,N,E)=N$. As another example, note that $E$-dominated in the first argument can be expressed as $f(\neq,=)=E$, or equivalently, as the conjunction of $f(E,=)=E$ and $f(N,=)=E$. 
Our notation
is justified by the fact that the type conditions satisfied by a function induce a partial function from types to types, and that in the case of homogeneous graphs, all that matters is the three types of pairs, given by the relations $E$, $N$, and $=$; the relation $\neq$ is the union of $E$ and $N$, and used as a shortcut.

\begin{definition}
Let $(D,E)$ be a homogeneous graph. We say a ternary canonical injection $f \colon D^3\To D$ is \emph{hyperplanely of behaviour projection} if the functions $(u,v) \mapsto f(c,u,v)$, $(u,v) \mapsto f(u,c,v)$, and $(u,v) \mapsto f(u,v,c)$ are of behaviour projection for all $c\in D$. Similarly other hyperplane behaviours, such as hyperplanely $E$-dominated, are defined.
\end{definition}

Note that hyperplane behaviours are defined by conditions for the type functions $f(=,\cdot,\cdot)$, $f(\cdot,=,\cdot)$, and $f(\cdot,\cdot,=)$. For example, hyperplanely $E$-dominated precisely means that $$f(=,=,\neq)=f(=,\neq,=)=f(\neq,=,=)=E \, .$$


\subsection{Achieving canonicity in Ramsey structures}\label{sect:Ramsey}
The next proposition, which is an instance of more general statements from~\cite{BP-reductsRamsey, BPT-decidability-of-definability}, provides us with  the main combinatorial tool for analyzing functions on Henson graphs. Equip $H_n$ with a total order $\prec$ in such a way that $(H_n, E,\prec)$ is homogeneous; up to isomorphism, there is only one such structure $(H_n, E,\prec)$, called the \emph{random ordered $K_n$-free graph}. The order $(H_n, \prec)$ is then isomorphic to the order $(\mathbb{Q}, <)$ of the rationals. By~\cite{NesetrilRoedlPartite}, $(H_n, E,\prec)$ is a \emph{Ramsey structure}, which implies the following proposition --  for more details, see the survey~\cite{BP-reductsRamsey}. 

\begin{prop}\label{prop:canfct}
Let $f\colon H_n^k \rightarrow H_n$, let $c_1,\ldots,c_r\in H_n$, 
and let $(H_n,E,\prec,c_1,\ldots,c_r)$ be the expansion of $(H_n,E,\prec)$ by the constants $c_1,\ldots,c_r$. Then
$$
\overline{\{\alpha\circ f\circ (\beta_1, \ldots, \beta_r)\;|\; \alpha\in \Aut(H_n, E, \prec),\; \beta_1, \ldots, \beta_r\in \Aut(H_n, E, \prec, c_1, \ldots, c_r)\}}
$$ 
contains a function $g$ such that
\begin{itemize}
\item $g$ is canonical as a function from $(H_n, E, \prec, c_1, \ldots, c_r)$ to $(H_n, E, \prec)$;
\item $g$ agrees with $f$ on $\{c_1,\ldots, c_r\}^k$.
\end{itemize}
In particular, $f$ generates a function $g$ with these properties.
\end{prop}

Similarly, Ramsey theory allows us to produce canonical functions on $(C_n^s,E)$, expanded with  a certain linear order. 
Equip $C_n^s$ with a total order $\prec$ so that the equivalence classes of $(C_n^s,\Eq)$ are \emph{convex} with respect to $\prec$, i.e., whenever $\Eq(u,v)$ holds and $u\prec w\prec v$, then $\Eq(u,w)$. 
Moreover, in the case where the size of the classes $s=\omega$, we require the order $\prec$ to be isomorphic to the order of the rational numbers on each equivalence class, and in case where the number of classes $n=\omega$, we require the order to be isomorphic to the order of the rational numbers between the classes (note that we already required convexity, so that $\prec$ naturally induces a linear order between the classes). 

If the number of classes $n$ is finite and their size $s=\omega$ infinite, let $P_1,\dots,P_n$ 
denote unary predicates such that $P_i$ contains precisely the elements in the $i$-th equivalence class of $\Eq$ with respect to the order on the classes induced by $\prec$. 
The structure 
$(C_n^\omega,E,\prec,P_1,\dots,P_n)$ is homogeneous and a Ramsey structure, since its automorphism group is,  as a topological group, isomorphic to $\Aut(\mathbb{Q};<)^n$, and since being a Ramsey structure is a property of the automorphism group (as a topological group)~\cite{Topo-Dynamics}. Thus, by~\cite{BP-reductsRamsey, BPT-decidability-of-definability}, we have the following analogous statement to Proposition~\ref{prop:canfct} for this structure. In the statement, we may drop the mention of the auxiliary relations  $P_1,\dots,P_n$, since these are first-order definable in $(C_n^s,E,\prec)$ and since the types over first-order interdefinable structures coincide; in other words, the relations were only needed temporarily in order to achieve homogeneity, required in~\cite{BP-reductsRamsey, BPT-decidability-of-definability}, but 
 not for the Ramsey property.
\begin{prop}\label{prop:canfct-C-high-s-low-n}
Let $n\geq 1$ be finite. Let $f\colon {(C_n^\omega)}^k \rightarrow C_n^\omega$, and let $c_1,\ldots,c_r\in C_n^\omega$. Then
$$
\overline{\{\alpha\circ f\circ (\beta_1, \ldots, \beta_r)\;|\; \alpha\in \Aut(C_n^\omega, E, \prec),\; \beta_1, \ldots, \beta_r\in \Aut(C_n^\omega, E, \prec, c_1, \ldots, c_r)\}}
$$ 
contains a function $g$ such that
\begin{itemize}
\item $g$ is canonical as a function from $(C_n^\omega, E, \prec, c_1, \ldots, c_r)$ to $(C_n^\omega, E, \prec)$;
\item $g$ agrees with $f$ on $\{c_1,\ldots, c_r\}^k$.
\end{itemize}
In particular, $f$ generates a function $g$ with these properties.
\end{prop}

If the class size $s$ is finite and their number $n=\omega$, we add $s$ unary predicates $Q_1,\dots,Q_s$ where $Q_i$ contains precisely the
$i$-th element for each equivalence class with respect to the order $\prec$.  
Then $(C_\omega^s,E,\prec,Q_1,\dots,Q_s)$ is homogeneous 
and Ramsey, since its automorphism group is isomorphic as a topological group to $\Aut(\mathbb{Q};<)$, so that we obtain an analogue of Propositions~\ref{prop:canfct} and~\ref{prop:canfct-C-high-s-low-n} also in this case. Again, we may drop the relations $Q_1,\dots,Q_n$, which are first-order definable in $(C_\omega^n,E,\prec)$, in the statement.
\begin{prop}\label{prop:canfct-C-high-n-low-s}
Let $s\geq 1$ be finite. Let $f\colon {(C_\omega^s)}^k \rightarrow C_\omega^s$, and let $c_1,\ldots,c_r\in C_\omega^s$. Then
$$
\overline{\{\alpha\circ f\circ (\beta_1, \ldots, \beta_r)\;|\; \alpha\in \Aut(C_\omega^s, E, \prec),\; \beta_1, \ldots, \beta_r\in \Aut(C_\omega^s, E, \prec, c_1, \ldots, c_r)\}}
$$ 
contains a function $g$ such that
\begin{itemize}
\item $g$ is canonical as a function from $(C_\omega^s, E, \prec, c_1, \ldots, c_r)$ to $(C_\omega^s, E, \prec)$;
\item $g$ agrees with $f$ on $\{c_1,\ldots, c_r\}^k$.
\end{itemize}
In particular, $f$ generates a function $g$ with these properties.
\end{prop}


\section{Polymorphisms over Henson graphs}\label{sect:polymorphisms}

We investigate polymorphisms of reducts of $(H_n,E)$. We start with unary polymorphisms in Section~\ref{sect:unary}, obtaining that we can assume that the relations $E$ and $N$ are pp-definable in our reducts, since otherwise their $\CSP$ can be modeled by a reduct of equality and hence has already been classified in~\cite{ecsps}.

We then turn to binary polymorphisms in Section~\ref{sect:binary}, obtaining Lemma~\ref{lem:essbin} telling us that, excluding in addition just one degenerate case where all polymorphisms are essentially unary functions, we may further assume the existence of a binary injective polymorphism. 

Building on the results of those sections, we show in Section~\ref{subsect:H} via an analysis of ternary polymorphisms that for any reduct which pp-defines the relations $E$ and $N$, either the polymorphisms preserve a certain relation $H$ (and hence, $H$ is pp-definable in the reduct by Theorem~\ref{conf:thm:inv-pol}), or there is a polymorphism of behaviour $\mini$ (Proposition~\ref{prop:higherArity}).

\subsection{The unary case: model-complete cores}\label{sect:unary}

A countable $\omega$-categorical structure $\Delta$ is called a \emph{model-complete core} if $\Aut(\Delta)$ is dense in $\End(\Delta)$, or equivalently, every endomorphism of $\Delta$ is an elementary self-embedding, i.e., preserves all first-order formulas over $\Delta$. Every countable $\omega$-categorical structure $\Gamma$ is \emph{homomorphically equivalent} to an up to isomorphism unique $\omega$-categorical model-complete core $\Delta$, that is, there exists homomorphisms from $\Gamma$ into $\Delta$ and vice-versa~\cite{Cores-journal}. Since the CSPs of homomorphically equivalent structures are equal, it has proven fruitful in classification projects to always work with model-complete cores. The following proposition essentially calculates the model-complete cores of the reducts of Henson graphs.

\begin{proposition}\label{prop:redendo}
Let $\Gamma$ be a reduct of $(H_n, E)$. Then either $\End(\Gamma)$ contains a function whose image induces an independent set  or $\End(\Gamma)=\overline{\Aut(\Gamma)}=  \overline{\Aut(H_n, E)}$.
\end{proposition}
\begin{proof}
Assume that $\End(\Gamma)\neq \overline{\Aut(H_n, E)}$. 
Then, since $\Gamma$ is $\omega$-categorical and by Theorem~\ref{conf:thm:inv-pol} and Lemma~\ref{lem:arity-reduction}, there exists an $f\in \End(\Gamma)$ which violates $E$ or $N$. 
If $f$ violated $N$ but not $E$, then there would be a copy of $K_n$ in the range of $f$, a contradiction.

Thus, we may assume that $f$ violates $E$, i.e., there exists $(u,v)\in E$ such that 
$(f(u), f(v))\in N$ or $f(u)= f(v)$. If for some such $(u,v)$ we have $f(u)=f(v)$, then one can generate by topological closure from $f$ a function whose image is an independent set. Since this is the first time we appeal to an argument with a flavour of topological closure, let us give it in longhand. First fix $u,v\in H_n$  such that $E(u,v)$ so that $f(u)=f(v)$. Given a subset $A$ of vertices containing $m\geq 1$ edges, we argue there is a $g$ generated by $f$ so that $g[A]$ contains fewer vertices than $A$. Indeed, take any $a,b\in A$ with $E(a,b)$, and an automorphism $\alpha\in\Aut(H_n,E)$ mapping $(a,b)$ to $(u,v)$, and use $g=f(\alpha(x),\alpha(y))$. Note that $g$ maps the edge $(a,b)$ to a single vertex, so that $g[A]$ is indeed smaller than $A$. By iterating this method, we can see that for every finite subset $A$ of $H_n$, there is a function $g$ generated by $f$ so that $g[A]$ is an independent set. The conclusion that then $f$ also generates a function which sends the entire domain $H_n$ onto an independent set is achieved via a typical compactness argument which appears in one form or another in most works on polymorphism clones of $\omega$-categorical structures; it uses topological closure together with $\omega$-categoricity. The modern and perhaps most elegant way to present it is to consider an equivalence relation $\sim$ on the set $F$ of all functions generated by $f$, defined by $g\sim g'$ if and only if $\overline{\{\alpha\circ g\;|\; \alpha\in \Aut(H_n,E)\}}=\overline{\{\alpha\circ g'\;|\; \alpha\in \Aut(H_n,E)\}}$. Then the factor space $F/_\sim$ is compact since $(H_n,E)$ is $\omega$-categorical. This has first been observed, in slightly different form, in~\cite{Topo-Birk}; we refer to~\cite{BodPin-CanonicalFunctions} for a proof of the variant we are using here. Let $(A_i)_{i\in\omega}$ be an increasing sequence of finite sets so that $\bigcup_{i \in \omega} A_i = H_n$. Fix a function $g_i$ generated by $f$ which sends $A_i$ onto an independent set. By compactness, a subsequence of $([g_i]_\sim\;|\; i\in\omega)$ converges in  $F/_\sim$ to a class $[g]_\sim$. This means that there are $\alpha_i\in\Aut(H_n,E)$, for $i\in\omega$, such that a subsequence of $(\alpha_i\circ g_i\;|\; i\in\omega)$ converges to $g$. But then $g$ maps $H_n$ onto an independent set.

%

%
Thus, we may assume that there exists $(u,v)\in E$ such that $(f(u), f(v))\in N$, and that no edges are collapsed to $=$ by $f$. If there existed $(u',v')\in N$ such that $f(u')= f(v')$, then picking an automorphism $\alpha\in\Aut(H_n,E)$ such that $(\alpha(f(u)),\alpha(f(v)))=(u',v')$, we would have that $f\circ\alpha\circ f$ collapses an edge to $=$. Having considered this situation above, we may hence assume that $f$ is injective.

By Proposition~\ref{prop:canfct}, the operation $f$ generates an injective canonical function $g\colon (H_n, E,\prec, u, v)\rightarrow (H_n, E,\prec)$ such that $f(u)=g(u)$ and $f(v)=g(v)$; in fact, since $f$ is unary, we can disregard the order $\prec$ and assume that $g$ is canonical as a function from $(H_n, E, u, v)$ to $(H_n, E)$~\cite[Proposition 3.7]{Pon11}.

Let 
\begin{align*}
U_{u v} & := \{x\in H_n \mid E(u, x)\wedge E(v, x)\}\;, \\ 
U_{u \overline{v}} & := \{x\in H_n \mid E(u, x)\wedge N(v, x)\}\;, \\
U_{\overline{u} v} & := \{x\in H_n \mid N(u, x)\wedge E(v, x)\}\;, \\
\text{ and }\;\; U_{\overline{u} \overline{v}} 
& := \{x\in H_n \mid N(u, x)\wedge N(v, x)\}. 
\end{align*}
As all four of these sets contain an independent set of size $n$, we cannot have $g(N)=E$ on any of them, as this would introduce a copy of $K_n$. 
Since no non-edges are collapsed to $=$ by our assumption above, we infer that $N$ is preserved by $g$ on all four sets.

If $g$ violates $E$ on $U_{\overline{u} \overline{v}}$, then, since $U_{\overline{u} \overline{v}}$ induces an isomorphic copy of $(H_n,E)$ therein, $g$ generates a function whose image is an independent set. Thus, we may assume that $g$ preserves $E$ on $U_{\overline{u} \overline{v}}$. 

Then $g$ preserves $N$ between $U_{\overline{u} \overline{v}}$ and any other orbit $X$ of $\Aut(H_n, E, u, v)$, as otherwise it would send non-edges to edges between these orbits, and 
 the image of the $n$-element induced subgraph of $(H_n, E)$ induced by any point in $X$ together with a copy of $K_{n-1}$ in $U_{\overline{u} \overline{v}}$ would be isomorphic to $K_n$. 

Assume that $g$ violates $E$ between $U_{\overline{u} \overline{v}}$ and another orbit $X$ of $\Aut(H_n, E, u, v)$. 
Let $A\subseteq H_n$ be finite with an edge $(x, y)$ in $A$. 
Then there exists an $\alpha\in \Aut(H_n, E)$ such that $\alpha(x)\in X$ and $\alpha[A\setminus \{x\}]\subseteq U_{\overline{u} \overline{v}}$. 
The function $(g\circ \alpha)\upharpoonright_{A}$ preserves $N$, and it maps $(x, y)$ to a non-edge. 
By an iterative application of this step we can systematically delete all edges of $A$. 
Hence, by topological closure, $g$ generates a function whose image is an independent set. 
Thus, we may assume that $g$ preserves $E$ between $U_{\overline{u} \overline{v}}$ and any other orbit of $\Aut(H_n, E, u, v)$. 

Let $X$ and $Y$ be infinite orbits of $\Aut(H_n, E, u, v)$, and assume that $g$ violates $N$ between $X$ and $Y$. 
There exist vertices $x\in X$ and $y\in Y$, and a copy of $K_{n-2}$ in $U_{\overline{u} \overline{v}}$ such that $(x,y)$ is the only non-edge in the graph induced by these $n$ vertices. 
Then, by the above, the image of this $n$-element set under $g$ induces a copy of $K_n$, a contradiction. 
Hence, we may assume that $g$ preserves $N$ on $H_n\setminus \{u, v\}$.

If $g$ violates $E$ on $H_n\setminus \{u, v\}$, then we can systematically delete the edges of any finite subgraph of $(H_n, E)$ whilst preserving the non-edges, and conclude that $g$ generates a function whose image is an independent set. 
Thus, we may assume that $g$ preserves $E$ on $H_n\setminus \{u, v\}$.

Assume that $g$ violates $E$ between $u$ and $U_{u \overline{v}}$. 
Given any finite $A\subseteq H_n$ with a vertex $x\in A$, there exists a $\beta\in \Aut(H_n, E)$ such that $\beta(x)=u$ and $\beta[A\setminus \{x\}]\subseteq U_{u \overline{v}}\cup U_{\overline{u} \overline{v}}$. 
Since, as observed earlier, $g$ preserves $N$ between $U_{\overline{u} \overline{v}}$ and any other orbit of $\Aut(H_n, E, u, v)$, including the orbits  $U_{u \overline{v}}$ and $\{u\}$, we conclude that $(g\circ \beta)\upharpoonright_A$ preserves $N$, and it maps edges from $x$ to non-edges. Thus, we can systematically delete the edges of $A$, and consequently, $g$ generates a function whose image is an independent set. Hence, we may assume that $g$ preserves $E$ between $u$ and $U_{u \overline{v}}$.

There exists a vertex $x\in U_{\overline{u} v}$ and a copy of $K_{n-2}$ in $U_{u \overline{v}}$ such that $(x,u)$ is the only non-edge in the graph induced by these $n-1$ vertices together with $u$. 
Thus, if $g$ violates $N$ between $\{u\}$ and $U_{\overline{u} v}$, then the image of this $n$-element set under $g$ induces a copy of $K_n$, a contradiction. Hence, $g$ preserves $N$ between $u$ and $U_{\overline{u} v}$.

By symmetry, we may assume that $g$ preserves $N$ between $v$ and $U_{u \overline{v}}$. 
Thus, $g$ preserves $N$. 
As $g$ deletes the edge between $u$ and $v$, we can systematically delete the edges of any finite subgraph of $(H_n, E)$. 
Hence, $g$ generates a function whose image is an independent set. 
\end{proof}

In the first case of Proposition~\ref{prop:redendo}, the model-complete core of the reduct is in fact a reduct of equality. Since the CSPs of reducts of equality have been classified~\cite{ecsps}, we do not have to consider any further reducts with an endomorphism whose image induces an independent set.

\begin{lemma}\label{lem:emptyendo}
Let $\Gamma$ be a reduct of $(H_n, E)$, and assume that $\End(\Gamma)$ contains a function whose image is an independent set. Then $\Gamma$ is homomorphically equivalent to a reduct of $(H_n, =)$.
\end{lemma}
\begin{proof}
Trivial.
\end{proof}

In the second case of Proposition~\ref{prop:redendo}, it turns out that all polymorphisms preserve the relations $E$, $N$, and $\neq$, by the following lemma and Theorem~\ref{conf:thm:inv-pol}.

\begin{lemma}\label{lem:neq-pp}
Let $\Gamma$ be a reduct of $(H_n, E)$. Then the following are equivalent:
\begin{itemize}
\item[(1)] $\End(\Gamma)= \overline{\Aut(H_n, E)}$.
\item[(2)] $E$ and $N$ have primitive positive definitions in $\Gamma$. 
\item[(3)] $E$, $N$, and $\neq$ have primitive positive definitions in $\Gamma$. 
\end{itemize}
\end{lemma}
\begin{proof}
Since $E$ and $N$ are orbits
of pairs with respect to $\Aut(H_n,E)$, the implication from (1) to (2) is an immediate consequence of Theorem~\ref{conf:thm:inv-pol} and Lemma~\ref{lem:arity-reduction}. For the implication from (2) to (3), it is enough to observe that the primitive positive formula
$\exists z (E(x,z) \wedge N(y,z))$ 
defines $x \neq y$.  Finally, the implication from (3) to (1) follows from the homogeneity of $(H_n,E)$.
\end{proof}

Before moving on to binary polymorphisms, we observe the following corollary of Proposition~\ref{prop:redendo}, 
first mentioned in~\cite{RandomReducts}.

\begin{cor}
For every $n\geq 3$, the permutation group $\Aut(H_n,E)$ is a maximal closed subgroup of the full symmetric group on $H_n$, i.e., every closed subgroup of the full symmetric group containing $\Aut(H_n,E)$ either equals $\Aut(H_n,E)$ or the full symmetric group.
\end{cor}
\begin{proof}
Let $G\supseteq \Aut(H_n,E)$ be a closed subgroup of the full symmetric group on $H_n$. 
Its closure $\overline{G}$ in the set of all unary functions on $H_n$ is a closed transformation monoid, i.e., a topologically closed monoid of unary functions, and hence the monoid of endomorphisms of a reduct of $(H_n,E)$ (cf.~for example~\cite{RandomMinOps}). By Proposition~\ref{prop:redendo}, $\overline{G}$  
 either contains a function $e$ whose image induces an independent set, or it equals $\overline{\Aut(H_n,E)}$. In the latter case, $G=\Aut(H_n,E)$, and in the first case
 we prove that $G$ equals the full symmetric group. Since $G$ is closed in the full symmetric group, it suffices to prove that for every $k \geq 1$ and
 all $s,t \in H_n^k$ there exists an element of $G$ which sends $s$ to $t$. Since $e \in \overline{G}$, there exists a $\beta \in G$ such that 
$e$ and $\beta$ agree on the tuples $s$ and $t$, and consequently $\beta$ sends the two tuples into independent sets. By the homogeneity of $(H_n,E)$, we have that $\beta(s)$ and $\beta(t)$ lie in the same orbit of $G$, and hence so do $s$ and $t$.
\end{proof}

We remark that the automorphism group of the random graph has five closed supergroups~\cite{RandomReducts}, which leads to more cases in the corresponding CSP classification in~\cite{BodPin-Schaefer}.

\subsection{Higher arities: generating injective polymorphisms}\label{sect:binary}

We investigate at least binary functions preserving $E$ and $N$ (and hence, by Theorem~\ref{conf:thm:inv-pol}, also $\neq$, since this relation is pp-definable from $E$ and $N$
by Lemma~\ref{lem:neq-pp}); our goal in this section is to show that they generate injections. Every unary function gives rise to a binary function by adding a dummy variable; the following definition rules out such ``improper" higher-arity functions.

\begin{defn}
A finitary operation $f(x_1,\ldots,x_n)$ on a set is \emph{essential} if it depends on more than one of its variables $x_i$.
\end{defn}

\begin{lemma}\label{lem:essbin}
Let $f\colon H_n^2\rightarrow H_n$ be a binary essential function that preserves $E$ and $N$. Then $f$ generates a binary injection.
\end{lemma}
\begin{proof}
Let $\Delta$ be the structure with domain $H_n$ and whose relations are those preserved by $\{f\}\cup\Aut(H_n,E)$; 
in particular, $E$, $N$, and $\neq$ are relations of $\Delta$. It is sufficient to show that $\Pol(\Delta)$ contains a binary injection (see Section~\ref{subsect:ua}).

We follow the strategy of the proof of \cite[Theorem 38]{RandomMinOps}. 
By \cite[Lemma 42]{RandomMinOps} it is enough to show that for all primitive positive formulas $\phi$ over $\Delta$ we have that whenever $\phi\wedge x\neq y$ and $\phi\wedge s\neq t$ are satisfiable in $\Delta$, then the formula $\phi\wedge x\neq y\wedge s\neq t$ is also satisfiable in $\Delta$. 
Still following the proof of \cite[Proposition 38]{RandomMinOps} it is enough to show the following claim.

\textbf{Claim.} Given two $4$-tuples $a = (x, y, z, z)$ and $b = (p, p, q, r)$ in $H_n^4$ such that $x\neq y$ and $q\neq r$, there exist $4$-tuples $a'$ and $b'$ such that $\tp(a)=\tp(a')$ and $\tp(b)=\tp(b')$ in $(H_n,E)$ and such that $f(a', b')$ is a $4$-tuple whose first two coordinates are different and whose last two coordinates are different. 

\textit{Proof of Claim.} We may assume that $x\neq z$ and $p\neq q$. 
We may also assume that $f$ itself is not a binary injection. 

In the following, we say that a point $(u,v)\in H_n^2$ is \emph{good} if $f(u,v)\neq f(u,w)$ for all $v\neq w$. 
Assume without loss of generality that there exist $u_1\neq u_2, v\in H_n$ such that $f(u_1, v)= f(u_2, v)$.
In particular, as $f$ preserves $\neq$, the points $(u_1, v)$ and $(u_2, v)$ are good. 
First fix any values $z',q'$ such that $(z',q')$ is good. 
We may assume that for any $x', y', p'\in H_n$ with $\tp(x', y', z')=\tp(x, y, z)$ and $\tp(p', q')= \tp(p, q)$ we have $f(x', p')= f(y', p')$, otherwise the tuples $a'= (x', y', z', z')$ and $b'= (p', p', q', r')$ are appropriate with any $r'\in H_n$ with $\tp(p', q', r')= \tp(p, q, r)$. 
Hence, as $f$ preserves $\neq$, all the points $(x',p')$ with $\tp(x', z') = \tp(x, z)$ and $\tp(p', q') = \tp(p, q)$ are good. 
So we obtained that whenever the point $(s,t)$ is good, and $s_0,t_0\in H_n$ are such that $\tp(s, s_0) = \tp(x, z)$ and $\tp(t, t_0) = \tp(p, q)$, then $(s_0, t_0)$ is also good, or otherwise we are done. 
We show that whatever the types $Q_1= \tp(x, z)$ and $Q_2= \tp(p, q)$ are, we can reach any point $(s_4, t_4)$ in $H_n^2$ from a given good point $(s_0, t_0)$ by at most four such steps. 
To see this, note that $Q_1$ and $Q_2$ are different from $=$ by assumption. 
Now let $s_1, s_2, s_3, t_1, t_2, t_3$ be such that 
\begin{itemize}
\item $s_0, s_1, s_2, s_3, s_4$ are pairwise different except that $s_0=s_4$ is possible, and
\item $t_0, t_1, t_2, t_3, t_4$ are pairwise different except that $t_0=t_4$ is possible, and
\item $(s_0, s_1), (s_1, s_2), (s_2, s_3), (s_3, s_4)\in Q_1$ and all other pairs $(s_i, s_j)$ are in $N$ except that $s_0=s_4$ is possible, and
\item $(t_0, t_1), (t_1, t_2), (t_2, t_3), (t_3, t_4)\in Q_2$ and all other pairs $(t_i, t_j)$ are in $N$ except that $t_0=t_4$ is possible.
\end{itemize}
These rules are not in contradiction with the extension property of $(H_n,E)$, thus such vertices exist, and we can propagate the good property from $(s_0, t_0)$ to $(s_4, t_4)$. 
Hence, every point is good, 
or we are done. 
If $f(u_1, v)= f(u_2, v)$ for all $u_1, u_2, v\in H_n$ with $\tp(u_1, u_2)= \tp(x, y)$, then $f$ would be essentially unary, since $(H_n, E)$ and its complement have diameter $2$. 
As $f$ is a binary essential function, we can choose $x', y', p'\in H_n$ such that $\tp(x', y')=\tp(x, y)$ and $f(x', p')\neq f(y', p')$. 
By choosing points $z', q', r'\in H_n$ such that $\tp(x', y', z') = \tp(x, y, z)$ and $\tp(p', q', r') = \tp(p, q, r)$ the tuples $a'= (x', y', z', z')$ and $b'= (p', p', q', r')$ are appropriate.
\end{proof}

The following lemma allows us to drop the restriction to binary essential functions.

\begin{lemma}\label{lem:essgen}
Let $k\geq 2$. Every essential function $f\colon H_n^k\rightarrow H_n$ that preserves $E$ and $N$ generates a binary injection.
\end{lemma}
\begin{proof}
By~\cite[Lemma 40]{RandomMinOps}, every essential operation generates a binary essential operation over the random graph; the very same proof works for the Henson graphs. Therefore, we may assume that $f$ itself is binary. 
The assertion now follows from Lemma~\ref{lem:essbin}.
\end{proof}

\subsection{The relation $H$}\label{subsect:H}

Let us investigate the case in which $\Gamma$, a reduct of $(H_n,E)$,  pp-defines $E$ and $N$ (and hence, $\neq$). The following relation characterizes the NP-complete cases in this situation.

\begin{definition}\label{defn:H}
We define a 6-ary relation 
$H(x_1,y_1,x_2,y_2,x_3,y_3)$ on $H_n$ by
\begin{align}
& \bigwedge_{i,j \in \{1,2,3\}, i \neq j, u \in \{x_i,y_i\}, v \in \{x_j,y_j\}} N(u,v) \nonumber \\
\wedge & \; \big((E(x_1,y_1) \wedge N(x_2,y_2) \wedge N(x_3,y_3))
\nonumber \\ 
& \vee \; (N(x_1,y_1) \wedge E(x_2,y_2) \wedge N(x_3,y_3)) \nonumber\\  
& \vee \; (N(x_1,y_1) \wedge N(x_2,y_2) \wedge E(x_3,y_3)) \big)\; . \nonumber
\end{align}
\end{definition}

Our goal for this section is to prove the following proposition, which states that if $\Gamma$ is a reduct of $(H_n, E)$ with $E$ and $N$ primitive positive definable in $\Gamma$, then either $H$ has a primitive positive definition in $\Gamma$, in which case $\Csp(\Gamma)$ is NP-complete, or $\Pol(\Gamma)$ has a certain canonical polymorphism which will imply tractability of the CSP. NP-completeness and tractability for those cases will be shown in Section~\ref{sect:CSP}.

\begin{proposition}\label{prop:higherArity}
Let $\Gamma$ be a reduct of $(H_n, E)$ with $E$ and $N$ primitive positive definable in $\Gamma$. Then at least one of the following holds:
    \begin{enumerate}
\item[(a)] There is a primitive positive definition of $H$ in $\Gamma$.
\item[(b)] $\Pol(\Gamma)$ contains a canonical binary injection of behaviour $\mini$.
\end{enumerate}
\end{proposition}

\subsubsection{Arity reduction: down to binary} With the ultimate goal of producing a binary canonical polymorphism of behaviour $\mini$, we now show that under the assumption that $\Gamma$ has a polymorphism preserving $E$ and $N$ yet violating $H$, it also has a binary polymorphism which is not of behaviour projection. We begin by ruling out some ternary behaviours which do play a role on the random graph.
\begin{lemma}\label{lem:nominority}
On $(H_n, E)$, there are no ternary functions of behaviour majority or satisfying the type conditions $f(N,N,E)=f(E,N,N)=E$.
\end{lemma}
\begin{proof}
These could introduce a $K_n$ in the $K_n$-free graph $(H_n,E)$, in the following fashions.

Suppose $f$ has behaviour majority, and choose $x_1,\ldots,x_{n-1}\in H_n$ inducing a copy of $K_{n-1}$, as well as a distinct $x_0\in H_n$ adjacent to $x_1$ and no other $x_i$. Then $\{f(x_0,x_1,x_2),f(x_1,x_2,x_0),f(x_2,x_0,x_1)\}$ induces $K_3$, and is adjacent to any element in $\{f(x_i,x_i,x_i)\;|\; 2< i\leq {n-1}\}$ since $E$ is preserved, so that altogether we obtain a copy of $K_n$.

Suppose now $f$ satisfies the type conditions $f(N,N,E)=f(E,N,N)=E$, and choose elements $x_1,\ldots,x_{n-1}\in H_n^3$ such that $\NNE(x_i,x_j)$ holds for distinct $1\leq i,j\leq {n-1}$. Pick furthermore $x_0\in H_{n}^3$ with $\ENN(x_0,x_i)$ for all $1\leq i \leq n-1$. Then $\{f(x_0),\ldots,f(x_{n-1})\}$ induces a $K_n$.
\end{proof}

\begin{proposition}\label{prop:getbinary}
	Let $f\colon H_n^k\To H_n$ be an operation that preserves $E$ and $N$ and violates $H$. Then $f$ generates a binary injection which is not of behaviour projection.
\end{proposition}

\begin{proof}
Since $H$ consists of three orbits of $6$-tuples in $(H_n,E)$, we may assume that $f$ is ternary, by Lemma~\ref{lem:arity-reduction}. Moreover, since $f$ preserves $E$ and $N$, it can only violate $H$ if it is essential. Thus, by Lemma~\ref{lem:essgen}, $f$ generates a binary injection $g$. If $g$ is not of behaviour projection, then we have proved the proposition. Otherwise, assume without loss of generality that it is of behaviour $p_1$. Consider the ternary function  $g(g(g(f(x,y,z),x),y),z)$. This function is injective, since $g$ is. Moreover, it violates $H$: if $x^1,x^2,x^3 \in H$ are so that $t:=f(x^1,x^2,x^3) \notin H$, then $t$ has pairwise distinct entries since $f$ preserves $\neq$. Hence, because $g$ is of behaviour $p_1$, $t':=g(t,x^1)$ has the same type as $t$, and so do $t'':=g(t',x^2)$ and $t''':=g(t'',x^3)$, proving the claim. By substituting $f$ by this function, we can therefore in the following assume that $f$ is itself injective.

We now prove the proposition by showing that a function of the form $(x,y)\mapsto f(x,y,\alpha(x))$, or $(x,y)\mapsto f(x,\alpha(x),y)$, or $(x,y)\mapsto f(y,x,\alpha(x))$, where $\alpha\in\Aut(H_n, E)$, is not of type projection.

Fix $x^1,x^2,x^3 \in H$ such that $f(x^1,x^2,x^3) \notin H$.
In the following, we will write $x_i := (x^1_i,x^2_i,x^3_i)$ for $1\leq i\leq 6$. So $(f(x_1),\dots,f(x_6)) \notin H$. If there exists $\alpha\in\Aut(H_n, E)$ such that $\alpha(x^i) = x^j$ for $1\leq i \neq j \leq 3$,
then our claim follows: for example, if $i=1$ and $j=3$, then the function $(x,y)\mapsto f(x,y,\alpha(x))$ violates $H$, and hence cannot be of behaviour projection.

We assume henceforth that there is no such automorphism $\alpha$. 
In this situation, by permuting arguments of $f$ if necessary, we can  
assume without loss of generality that 
\begin{align*}
	\ENN(x_1,x_2),\, \NEN(x_3,x_4),\,\text{and } \NNE(x_5,x_6).
\end{align*}
We set $$S := \{ y \in H_n^3 \; | \; \NNN(x_i,y) \text{ for all } 1\leq i \leq 6 \} \; .$$
Consider the binary relations $Q_1Q_2Q_3$ on $H_n^3$, where $Q_i\in\{E,N\}$ for $1\leq i\leq 3$.
We show that either our claim above proving the proposition holds, or for each such relation $Q_1Q_2Q_3$, whether $E(f(u),f(v))$ or $N(f(u),f(v))$ holds for $u,v \in S$ with $Q_1Q_2Q_3(u,v)$ does not depend on $u,v$; that is, whenever $u,v,u',v'\in S$ satisfy $Q_1Q_2Q_3(u,v)$ and $Q_1Q_2Q_3(u',v')$, then $E(f(u),f(v))$ if and only if $E(f(u'),f(v'))$. Note that this is another way of saying that $f$ satisfies some type conditions on $S$. We go through all possibilities of $Q_1Q_2Q_3$. 
\begin{enumerate}
\item[(1)] $Q_1Q_2Q_3=\ENN$. Let $\alpha \in \Aut(H_n, E)$ be such that $(x^2_1,x^2_2,u_2,v_2)$ is mapped
to $(x^3_1,x^3_2,u_3,v_3)$; such an automorphism exists since 
$$
\NNN(x_1, u), \NNN(x_1, v),
\NNN(x_2, u), \NNN(x_2, v)
$$ 
hold, and since $(x^2_1,x^2_2)$ has the same type as
$(x^3_1,x^3_2)$, and $(u_2,v_2)$ has the same type as $(u_3,v_3)$.
We are done if the operation $g$ defined by $g(x,y):=f(x,y,\alpha(y))$ is not of type projection. Otherwise, $E(g(u_1,u_2),g(v_1,v_2))$ iff $E(g(x_1^1,x_1^2),g(x_2^1,x_2^2))$. Combining this with the equations $(f(u),f(v))=(g(u_1,u_2),g(v_1,v_2))$ and 
$(g(x_1^1,x_1^2),g(x_2^1,x_2^2))=(f(x_1),f(x_2))$, we get that $E(f(u),f(v))$ iff $E(f(x_1),f(x_2))$, and so our claim holds for this case.
\item[(2)] $Q_1Q_2Q_3=\NEN$ or $Q_1Q_2Q_3=\NNE$. These cases are analogous to the previous case.
\item[(3)] $Q_1Q_2Q_3=\NEE$. Let $\alpha$ be defined as in the first case. 
Reasoning as above, if the operation defined by $f(x,y,\alpha(y))$ is of type projection, then one gets that $E(f(u),f(v))$ iff $N(f(x_1),f(x_2))$.
\item[(4)] $Q_1Q_2Q_3=\ENE$ or $Q_1Q_2Q_3=\EEN$. These cases are analogous to the previous case.
\item[(5)] $Q_1Q_2Q_3= \EEE$ or $Q_1Q_2Q_3=\NNN$. Trivial since $f$ preserves $E$ and $N$.
\end{enumerate}

\ignore{
We now make another case distinction, based on the fact that $(f(x_1),\dots,f(x_6)) \notin H$. 
\begin{enumerate}
\item[(1)] Suppose that $E(f(x_1),f(x_2)), E(f(x_3),f(x_4)),E(f(x_5),f(x_6))$. 
Then by the above $f$ is of behaviour minority on $S$, a contradiction since $S$ induces a copy of $(H_n,E)^3$ and because of Lemma~\ref{lem:nominority}.
\item[(2)] Suppose that $N(f(x_1),f(x_2)),N(f(x_3),f(x_4)),N(f(x_5),f(x_6))$. 
Then $f$ has behaviour  majority on $S$, again contradicting Lemma~\ref{lem:nominority}.
\item Suppose that $E(f(x_1),f(x_2)),E(f(x_3),f(x_4)),N(f(x_5),f(x_6))$. 
Let $e$ be an endomorphism of $(H_n, E,N)$ such that for all $w \in H_n$, all $1\leq j\leq 3$, and all $1\leq i \leq 6$ 
we have that $N(x_i^j,e(w))$. 
Then $(u_1,u_2,e(f(u_1,u_2,u_3))) \in S$ for all 
$(u_1,u_2,u_3) \in S$. 
Hence, by the above, 
the ternary operation defined by $f(x,y,e(f(x,y,z)))$ is of type behaviour on $S$, a contradiction.
\item Suppose that $E(f(x_1),f(x_2))$, $N(f(x_3),f(x_4))$, $E(f(x_5),f(x_6))$, or 
$N(f(x_1),$ $f(x_2))$, $E(f(x_3),f(x_4))$, $E(f(x_5),f(x_6))$. 
These cases are analogous to the previous case.
\end{enumerate}
Each of the cases leads to a contradiction, hence proving the proposition.
}
Now we show that $f$ actually cannot satisfy the type conditions above on $S$. First note that setting  $h(x,y,z):=f(e^1(x),e^2(y),e^3(z))$ for self-embeddings $e_1,e_2,e_3$ of $(H_n,E)$ such that $(e_1,e_2,e_3)(u)\in S$ for all $u\in H_n^3$, we obtain a function $h$ which satisfies the same type conditions everywhere; such embeddings exist since by its definition, the projection of $S$ onto any coordinate has an induced copy of $(H_n,E)$. Then, if $\{(f(x_1),f(x_2)),(f(x_3),f(x_4)),(f(x_5),f(x_6))\}$ has $E$ twice or more, by~(1) and~(2) we get that $h$ satisfies two type conditions from the minority behaviour, say $h(N,N,E)=E$ and $h(E,N,N)=E$, contradicting Lemma~\ref{lem:nominority}. If $\{(f(x_1),f(x_2)),(f(x_3),f(x_4)),(f(x_5),f(x_6))\}$  has $E$ no times, then by~(3) and~(4) $h$ is of behaviour majority, again contradicting Lemma~\ref{lem:nominority}. Thus, the set must have precisely one $E$, contradicting $f(x^1,x^2,x^3)\nin H$.
\end{proof}

\subsubsection{Producing min}\label{sect:producingmin} By Proposition~\ref{prop:getbinary}, it remains to show the following to obtain a proof of Proposition~\ref{prop:higherArity}.

\begin{proposition}\label{prop:nonProjGeneratesMin}
    Let $f \colon H_n^2 \rightarrow H_n$ be a binary injection preserving $E$ and $N$ that is not of behaviour projection. Then $f$ generates a binary canonical injection of behaviour $\mini$.
\end{proposition}

In the remainder of this section we will prove this proposition by a Ramsey theoretic analysis of $f$, which requires the following definitions and facts from~\cite{RandomMinOps} concerning behaviours with respect to the homogeneous expansion of the graphs $(H_n,E)$ by the total order $\prec$ from Section~\ref{sect:Ramsey}. At this point, it might be appropriate to remark that canonicity of functions on $H_n$, and even the notion of behaviour, does depend on which underlying structure we have in mind, in particular, whether or not we consider the order $\prec$ (which we almost managed to ignore so far).
Let us define the following behaviours for functions from $(H_n, E,\prec)^2$ to $(H_n, E)$; we write $\succ$ for the relation $\{(a,b) \; | \; b \prec a\}$. 

\begin{definition}
    Let $f \colon H_n^2 \rightarrow H_n$ be injective. If for all $u,v\in H_n^2$ with $u_1\prec v_1$ and $u_2\prec v_2$ 
    \begin{itemize}
        \item $E(f(u),f(v))$ if and only if $\EE(u,v)$, then \emph{$f$ behaves like $\mini$ on input $(\prec,\prec)$}.
        \item $E(f(u),f(v))$ if and only if $E(u_1,v_1)$, then \emph{$f$ behaves like $p_1$ on input $(\prec,\prec)$}.
        \item $E(f(u),f(v))$ if and only if $E(u_2,v_2)$, then \emph{$f$ behaves like $p_2$ on input $(\prec,\prec)$}.
    \end{itemize}
    Analogously, we define behaviours on input $(\prec,\succ)$ using pairs $u,v\in H_n^2$ with $u_1 \prec v_1$ and $u_2\succ v_2$.
\end{definition}

\begin{proposition}\label{prop:binaryBehaviourOnInputBLaBla}
Let $f \colon H_n^2 \rightarrow H_n$ be an injection which is canonical as a function from $(H_n, E,\prec)^2$ to $(H_n, E,\prec)$ and suppose $f$ preserves $E$ and $N$. Then it behaves like $\mini$, $p_1$ or $p_2$ on input $(\prec ,\prec )$ (and similarly on input $(\prec ,\succ)$).
\end{proposition}
\begin{proof}
    By definition of the term canonical; one only needs to enumerate all possible types of pairs $(u,v)$, where $u,v \in H_n^2$, and recall that $(H_n, E)$ does not contain any clique of size $n$, which makes some behaviours impossible to be realized by $f$.
\end{proof}

\begin{definition}
    If an injection $f \colon H_n^2 \rightarrow H_n$ behaves like $X$ on input $(\prec ,\prec )$ and like $Y$ on input $(\prec ,\succ )$, where $X,Y\in\{\mini, p_1,p_2\}$, then we say that $f$ is of \emph{behaviour $X / Y$}.
\end{definition}




In the following lemmas, we
show that every injective canonical binary function which
behaves differently on input $(\prec ,\prec )$ 
and on input $(\prec,\succ)$ generates a function 
which behaves the same way on both inputs, allowing us to ignore the order again.

\begin{lemma}\label{lem:mixtyp:minp}
    Suppose that $f \colon H_n^2 \rightarrow H_n$ is injective and canonical from $(H_n, E,\prec)^2$ to $(H_n, E,\prec)$, and suppose that it is of type $\mini / p_i$ or of type $p_i / \mini$, where $i\in\{1,2\}$. Then $f$ generates a binary injection of type $\mini$.
\end{lemma}
\begin{proof}
Since the calculus for behaviours on the Henson graphs is the same as that on the random graph, the same proof as in \cite{BodPin-Schaefer} works.
\end{proof}

\begin{lemma}\label{lem:p1p2impossible}
No binary injection $f \colon H_n^2 \rightarrow H_n$ can have behaviour $p_1/p_2$.
\end{lemma}
\begin{proof}
Such a behaviour would introduce a $K_n$ in a $K_n$-free graph.
\end{proof}

Having ruled out some behaviours without constants, we finally introduce constants to the language to prove Proposition~\ref{prop:nonProjGeneratesMin}. 
\begin{proof}[of Proposition~\ref{prop:nonProjGeneratesMin}]
   Fix a finite set $C:=\{c_1,\ldots, c_m\}\subseteq H_n$ on which the
fact that $f$ is not of behaviour projection is witnessed. Invoking Proposition~\ref{prop:canfct}, we
may henceforth assume that $f$ is canonical as a function from $(H_n,
E,\prec, c_1,\ldots,c_m)^2$ to $(H_n, E,\prec)$. We are going to show that $f$ generates a binary injection $g$ of behaviour $\mini$. Then another application of Proposition~\ref{prop:canfct} to $g$ yields a canonical function $g'$; this function is still of behaviour $\mini$  because any function of the form $\alpha(g(\beta(x),\gamma(y))$ is of type $\mini$, for automorphisms $\alpha,\beta,\gamma$ of $(H_n,E)$, and $g'$ is generated from operations of this type by  topological closure.

To obtain $g$, consider in the structure ${(H_n, E,\prec, c_1,\ldots,c_m)}$ the orbit
$$
O:=\{a\in H_n\; |\; N(a,c_i) \text{ and } a\prec c_i \text{ for all }
1\leq i\leq m\}.
$$
Then $O$ induces a structure isomorphic to $(H_n, E,\prec)$, as it
satisfies the extension property for totally ordered $K_n$-free
graphs: the same extensions can be realized in $O$ as in $(H_n,
E,\prec)$. Therefore, by
Lemma~\ref{prop:binaryBehaviourOnInputBLaBla}, $f$ has one of the
three mentioned behaviours on input $(\prec,\prec)$ and on input $(\prec,\succ)$. By Lemmas~\ref{lem:mixtyp:minp}
and~\ref{lem:p1p2impossible}, we may assume that
$f$ behaves like a projection on $O$, for any other combination of behaviours implies that it generates a binary injection of behaviour $\mini$. 

Assume without loss of generality that $f$ behaves like $p_1$ on $O$. Let $u\in O^2$ and $v\in (H_n\setminus \{c_1,\ldots, c_m\})^2$ satisfy
$\neq\neq(u,v)$; we claim that $f$ behaves like $p_1$ or like $\mini$
on $\{u,v\}$. Otherwise we must have $\NE(u,v)$ and $E(f(u),f(v))$. Pick  $q_1,\ldots, q_{n-1}\in O^2$ forming a
clique in the first coordinate, an independent set in the second
coordinate, and such that the type of $(q_i,v)$ equals the type of
$(u,v)$ in $(H_n,E,\prec,c_1,\ldots,c_n)$. Then by canonicity, the image of $\{q_1,\ldots,q_{n-1},v\}$
under $f$ forms a clique of size $n$, a contradiction.




Suppose next that there exist $u\in O^2$ and $v\in (H_n\setminus C)^2$ with $\neq\neq(u,v)$ such that $f$ does not behave like $p_1$ (and hence, by the above, behaves like $\mini$) on $\{u,v\}$. This means that $\EN(u,v)$ and $N(f(u),f(v))$. We use topological closure to show that $f$ generates a binary injection which behaves like $\mini$.
To this end, set $$S:=\{p\in H_n^2\;|\; \tp(p,v)=\tp(u,v) \text{ in } (H_n,E,\prec,c_1,\ldots,c_n)\}\subseteq O^2\; .$$
Now let $q_0\in H_n^2$ be arbitrary. Pick a self-embedding $e$ of $(H_n,E)$ whose range is contained in $O$. Then the function $r\colon H_n^2\To H_n^2$ defined by $(x,y)\mapsto (f(e(x),e(y)),f(e(y),e(x)))$ has the property that $\EN(p,q)$ implies $\EN(r(p),r(q))$ and $\NE(p,q)$ implies $\NE(r(p),r(q))$, for all $p,q\in H_n^2$, since $f$ behaves like $p_1$ on $O$. Moreover, since $f$ is injective, we have that $p\neq q$ implies $\NEQNEQ(r(p),r(q))$. By the latter property, there exist self-embeddings $e_1,e_2$ of $(H_n,E)$ such that for the function $r'\colon H_n^2\To H_n^2$ defined by $r':=(e_1,e_2)\circ r$ we have that $r'(q_0)=v$, that $r'(p)\in O^2$ for all $p\in H_n^2\setminus \{q_0\}$,  and that $r'(p)\in S$ for all $p\in H_n^2$ with $\EN(p,q_0)$.  Then the function $h\colon H_n^2\To H_n^2$ defined by $h(x,y):=(f(r'(x,y)),y)$ has the property that $\NN(h(p),h(q_0))$ holds for all $p\in H_n^2$ with $\EN(p,q_0)$, since $f$ behaves like $\mini$ between $S$ and $v$. Moreover, $\NE(h(p),h(q_0))$ holds for all $p\in H_n^2$ with $\NE(p,q_0)$, since $f$ behaves like $p_1$ or like $\mini$ between $O^2$ and $v$.  Finally, for any $p,p'\in H_n^2$ distinct from $q_0$ we have that $\EN(p,p')$ implies $\EN(h(p),h(p'))$ and $\NE(p,p')$ implies $\NE(h(p),h(p'))$, since $f$ behaves like $p_1$ on $O$. Similarly, one can construct a function $h'$ on $H_n^2$ which preserves $\EN$ and $\NE$ between any $p,p'\in H_n^2$ distinct from $q_0$,  and such that $\NE(p,q_0)$ implies $\NN(h'(p),h'(q_0))$. 
Iterating such functions for different choices of $q_0$, we obtain functions $r_A\colon H_n^2\To H_n^2$ for every finite subset $A\subseteq H_n^2$ such that $\EN(p,p')$ or $\NE(p,p')$ implies $\NN(r_A(p),r_A(p'))$  for all $p,p'\in A$. By topological closure (cf.~Proposition~\ref{prop:redendo}), one then gets a function $r\colon H_n^2\To H_n^2$ which has this property everywhere, and then $f(r)$ is the desired binary injection of behaviour $\mini$.

So we assume henceforth that $f$ behaves like $p_1$ on $\{u,v\}$ for
all $u\in O^2$ and all $v\in (H_n\setminus C)^2$ with $\NEQNEQ(u,v)$. 
We
then claim that $f$ must behave like $p_1$ or like
$\mini$ on $\{u,v\}$ for all $u, v\in (H_n\setminus C)^2$ with $\NEQNEQ(u,v)$. Otherwise, we must have $\NE(u,v)$ and $E(f(u),f(v))$.  Pick  $q_1,\ldots, q_{n-2}\in O^2$
forming a clique in the first coordinate, an independent set in the
second coordinate, and adjacent to $u$ and $v$ in the first coordinate. Applying $f$ we get a clique of size $n$, a
contradiction.


If there exist $u,v\in (H_n\setminus C)^2$ with $\EN(u,v)$ and $N(f(u),f(v))$, then by precomposing $f$ with a self-embedding $e$ of $(H_n,E)$ whose range equals $H_n\setminus C$, we may moreover assume that $f$ behaves like $p_1$ or like $\mini$ on $\{u',v'\}$, for all $u',v'\in H_n^2$. A standard iterating argument, similar to the one above (or the one given in detail in the proof of Proposition~\ref{prop:redendo}), then shows that $f$ generates a binary injection $g$ of type
$\mini$.

We thus henceforth assume that $f$ behaves like $p_1$ on $\{u,v\}$ for
all $u,v\in (H_n\setminus C)^2$. We next claim that $f$ must behave like $p_1$ or like
$\mini$ on $\{u,v\}$ for all $u\in H_n^2$ and all $v\in (H_n\setminus C)^2$ with $\NEQNEQ(u,v)$. Otherwise, we must have $\NE(u,v)$ and $E(f(u),f(v))$.  Pick  $q_1,\ldots, q_{n-2}\in H_n^2$
forming a clique in the first coordinate, an independent set in the
second coordinate, adjacent to $u$ in both coordinates, and adjacent to $v$ in precisely the first coordinate. Applying $f$ we get a clique of size $n$, a
contradiction.

A similar argument as two paragraphs above now shows that if there exist $u,v\in H_n^2$ with $\EN(u,v)$ and $N(f(u),f(v))$, then $f$ generates a binary injection $g$ of type $\mini$.

It thus remains to consider the case where $f$ behaves like $p_1$ on $\{u,v\}$ whenever $u\in H_n^2$ and $v\in (H_n\setminus C)^2$, and where there exist $u,v\in H_n^2$ such that $\NE(u,v)$ and $E(f(u),f(v))$. Pick $q_1,\ldots, q_{n-2}\in (H_n\setminus C)^2$ which are adjacent to $u$ and $v$ in the first coordinate and not the second, which form a clique in the first coordinate, and which form an independent set in the second coordinate. The image of the set $\{u,v,q_1,\ldots,q_{n-2}\}$ under $f$ then is a clique of size $n$, so that this case cannot occur.

\end{proof}

\section{CSPs over Henson graphs}\label{sect:CSP}
\subsection{Hardness of $H$}\label{subsect:hardnessOfH}

We now show that any reduct of $(H_n,E)$ which has $H$ among its relations, and hence by Lemma~\ref{lem:pp-reduce} every reduct which pp-defines $H$, has an NP-hard CSP. We first show hardness directly by reduction from positive 1-in-3-SAT; then, we provide another proof via \emph{h1 clone homomorphisms} which gives further insight into the mathematical structure of such reducts, and draws connections to the general dichotomy conjecture for reducts of finitely bounded homogeneous structures.

\subsubsection{Reduction from positive 1-in-3-SAT} We start by showing hardness directly, which however does not tell us anything about the structure of the polymorphism clones of reducts which pp-define $H$.

\begin{proposition}\label{prop:Hhardnew}
$\Csp(H_n, H)$ is NP-hard.
\label{prop:new-henson-hardness}
\end{proposition}
\begin{proof}
We reduce positive 1-in-3-SAT to $\Csp(H_n, H)$. Each variable $v$ in an instance $\phi$ of the former becomes two variables $v,v'$ in the corresponding instance $\psi$ of the latter. Each clause $(u,v,w)$ from $\phi$ becomes a tuple $H(u,u',v,v',w,w')$ in $\psi$. It is easy to see that $\phi$ is a yes-instance of 1-in-3-SAT if and only if $\psi$ is a yes-instance of  $\Csp(H_n, H)$, and the result follows.
\end{proof}

\subsubsection{Clone homomorphisms}

We will now show another way to prove NP-hardness of $\CSP(H_n,H)$ via a structural property of $\Pol(H_n,H)$, using general results from~\cite{wonderland} (a strengthening of the structural hardness proof in~\cite{Topo-Birk}). This will allow us to show that the dichotomy for the Henson graphs is in line with the dichotomy conjecture, for CSPs of reducts of finitely bounded homogeneous structures, from~\cite{wonderland} (and the earlier dichotomy conjecture for the same class, due to Bodirsky and Pinsker (cf.~\cite{BPP-projective-homomorphisms}), which has recently been proved equivalent~\cite{TwoDichotomyConjectures}.

\begin{defn}\label{defn:clonehomo}
Let $\Gamma$ be a structure. A \emph{projective clone homomorphism} of $\Gamma$ (or $\Pol(\Gamma)$) is a mapping from $\Pol(\Gamma)$ onto its projections which
\begin{itemize}
\item preserves arities;
\item fixes each projection;
\item preserves composition.
\end{itemize}
 A \emph{projective strong h1 clone homomorphism} of $\Gamma$ is a mapping as above, where the third condition is weakened to preservation of composition of any function in $\Pol(\Gamma)$ with projections only.
\end{defn}

Recall that $\Pol(\Gamma)$ is equipped with the topology of pointwise convergence, for any structure $\Gamma$.

\begin{thm}[from \cite{wonderland}]\label{thm:wonderland}
Let $\Gamma$ be a countable $\omega$-categorical structure in a finite relational language which has a uniformly continuous projective strong h1 clone homomorphism. Then $\CSP(\Gamma)$ is NP-hard.
\end{thm}

\begin{proposition}\label{prop:h-hard}
The structure $(H_n, H)$ has a uniformly continuous projective strong h1 clone homomorphism. Consequently,   $\Csp(H_n, H)$ is NP-hard.
\end{proposition}
\begin{proof}
Note that $H$ consists of three orbits of $6$-tuples with respect to $\Aut(H_n,E)$. Let $a^1,a^2,a^3\in H$ be representatives of those three orbits. By reshuffling the $a^i$ we may assume that $\ENN(a_1,a_2)$, $\NEN(a_3,a_4)$, $\NNE(a_5,a_6)$ (where $a_i$ denotes the $i$-th row of the matrix $(a^1,a^2,a^3)$, for $1\leq i\leq 6$).

We claim that whenever $f\in\Pol(H_n,H)$ is ternary, and $b^1,b^2,b^3\in H$ are so that $\tp(b^1,b^2,b^3)=\tp(a^1,a^2,a^3)$, then $\tp(f(b^1,b^2,b^3))=\tp(f(a^1,a^2,a^3))$ in $(H_n,E)$. To see this, let $c^1,c^2,c^3\in H$ be so that $\tp(c^1,c^2,c^3)=\tp(b^1,b^2,b^3)$, and such that no entry of any $c^i$ is adjacent to any component of any $b^j$ or $a^j$. Suppose that $f(b^1,b^2,b^3)$ and $f(a^1,a^2,a^3)$ do not have the same type, then one of them, say $f(a^1,a^2,a^3)$, does not have the same type as $f(c^1,c^2,c^3)$. Without loss of generality, this is witnessed on the first two components of the 6-tuples $f(c^1,c^2,c^3)$ and $f(a^1,a^2,a^3)$. For $1\leq i\leq 3$, consider the $6$-tuple $d^i:=(c_1^i,c_2^i,a_3^i,\ldots,a_6^i)$, i.e., in $a^i$ we replace the first two components by the components from $c^i$. Then $d^i\in H$, but $f(d^1,d^2,d^3)\nin H$, a contradiction.

Let $f\in\Pol(H_n,H)$. Then precisely one out of $(f(a_1),f(a_2))$, $(f(a_3),f(a_4))$, and $(f(a_5),f(a_6))$ is contained in $E$. If this is the case for the first pair, then it follows from the claim above that $f$ satisfies the three type conditions $f(E,N,N)=E$ and $f(N,E,N)=f(N,N,E)=N$; in the other two cases we obtain similar type conditions.

Let $\xi$ be the mapping which sends every ternary $f\in\Pol(H_n,H)$ to the ternary projection which is consistent with the type conditions satisfied by $f$ (in the case considered above, the projection onto the first coordinate). Then $\xi$ clearly preserves arities and projections. Moreover, let $f\in\Pol(H_n,H)$ and say without loss of generality that $f(E,N,N)=E$, so that $\xi(f)$ is the first projection. Then whenever $g_1,g_2,g_3\in\Pol(H_n,H)$ are ternary projections, we have $\xi(f(g_1,g_2,g_3))=g_1$; this is easy to verify by checking the behaviour of $f(g_1,g_2,g_3)$ on a suitable triple of the form $(E,N,N),(N,E,N)$, or $(N,N,E)$. Hence, $\xi$ satisfies the definition of a strong projective h1 clone homomorphism for ternary functions.  It has been observed (see e.g.~\cite{Topo}) that $\xi$ then uniquely extends to a strong projective h1 clone homomorphism of the entire clone $\Pol(H_n,H)$. 
Since the value of every $f$ under $\xi$ can be seen on any test matrix $(a^1,a^2,a^3)$ as above, we have that $\xi$ is uniformly continuous, and so is its extension (the latter follows from the proof in~\cite{Topo}).
\end{proof}

\subsection{Tractability of min}\label{subsect:tractabilityOfMin}

We now show that if a reduct $\Gamma$ of $(H_n,E)$ with finite relational signature has 
a polymorphism which is of behaviour $\mini$, 
then $\Csp(\Gamma)$ is in P.
We are going to apply Theorem~\ref{thm:maximal} below for the structure $\Delta := (H_n,E)$. 
In the theorem, $\hat \Delta$ denotes the
expansion of $\Delta$ by the inequality relation $\neq$ and 
by the complement $\hat R$
of each relation $R$ in $\Delta$.

\begin{theorem}[Proposition~14 in~\cite{Maximal}]\label{thm:maximal}
Let $\Delta$ be an $\omega$-categorical structure, 
and let $\Gamma$ be a reduct of $\Delta$. If $\Gamma$ has a polymorphism $e$ which is
an embedding of $\Delta^2$ into $\Delta$,
and if $\Csp(\hat \Delta)$ is in P, then $\Csp(\Gamma)$
is in P as well. 
\end{theorem}

\begin{prop}\label{prop:mintractable}
Let $\Gamma$ be a reduct of $(H_n,E)$ which has a polymorphism of behaviour $\mini$. Then $\CSP(\Gamma)$ is in $P$.
\end{prop}
\begin{proof}
To apply Theorem~\ref{thm:maximal} to $\Delta=(H_n,E)$, we first show that the CSP for
 $\hat \Delta = (H_n,E,\hat E,\neq)$ can be solved
 in polynomial time. 
Given an input primitive positive formula, we identify all variables $x,y$ such that $x=y$ is a constraint in the input. Then the formula is satisfiable if and only if for all
 variables $x_1,\dots,x_n$ we have
 \begin{itemize}
 \item $E(x_i,x_j)$ is not in the input for some distinct $i,j \in \{1,\dots,n\}$ (in particular, the statement for $x_1= \dots = x_n$ implies that the input does not contain constraints of the form $E(x,x)$), 
\item there are no constraints of the form $x_1 \neq x_1$,  and
\item there are no constraints of the form $E(x_1,x_2)$ and $\hat E(x_1,x_2)$. 
\end{itemize}
Since $n$ is fixed, it is clear that these conditions can be checked in polynomial time. 

Now let $f\in\Pol(\Gamma)$ be a canonical binary injection of  behaviour $\mini$. Each of the type conditions $f(N,{=})=E$ and $f({=},N)=E$ is impossible, because they would introduce a $K_n$. Further, $f(E,{=})=N$ or $f({=},E)=N$, for the same reason. But then $g(x,y):=f(f(x,y),f(y,x))$ is of behaviour $\mini$ and $N$-dominated, and therefore an embedding from $(H_n,E)^2$ into $(H_n,E)$. Hence, $\CSP(\Gamma)$ is in P by Theorem~\ref{thm:maximal}.
\end{proof}

\section{Summary for the Henson graphs}\label{sect:summary_Henson}
\subsection{Proof of the complexity dichotomy}
We are ready to assemble our results to prove the dichotomy for the CSPs of reducts of Henson graphs. 

\begin{proof}[of Theorem~\ref{thm:main}]
Let $\Gamma$ be a reduct of $(H_n,E)$. If $\End(\Gamma)$ contains a
function whose image is an independent set,
then $\Csp(\Gamma)$ equals the CSP
for a reduct of $(H_n,=)$ by Lemma~\ref{lem:emptyendo}, and such CSPs are either in P or NP-complete~\cite{ecsps}. 
Otherwise, $\End(\Gamma) = \overline{\Aut(H_n,E)}$ by Proposition~\ref{prop:redendo}. 
Lemma~\ref{lem:neq-pp} 
shows that $E$, $N$,
and $\neq$ are pp-definable
in $\Gamma$. 

If also the relation $H$ is pp-definable in $\Gamma$, then
$\Csp(\Gamma)$ is NP-hard by Proposition~\ref{prop:h-hard} (or Proposition~\ref{prop:Hhardnew}); it is in NP since $\Gamma$ is a reduct of $(H_n,E)$, which is a finitely bounded homogeneous structure. 

So let us assume that $H$ is not pp-definable in $\Gamma$; then Proposition~\ref{prop:higherArity} shows that $\Pol(\Gamma)$
contains a canonical binary injection $f$ of  behaviour $\mini$. Hence, $\CSP(\Gamma)$ is in P by Proposition~\ref{prop:mintractable}.
\end{proof}

\subsection{Discussion} We can restate Theorem~\ref{thm:main} in a more detailed fashion as follows.
\begin{thm}\label{thm:main2}
Let $\Gamma$ be a reduct of a Henson graph $(H_n,E)$. Then one of the following holds.
\begin{itemize}
\item[(1)] $\Gamma$ has an endomorphism inducing an independent set, and is homomorphically equivalent to a reduct of $(H_n,=)$.
\item[(2)] $\Pol(\Gamma)$ has a uniformly continuous projective clone homomorphism.
\item[(3)] $\Pol(\Gamma)$ contains a binary canonical injection which is of behaviour $\mini$ and $N$-dominated.
\end{itemize}
Items~(2) and~(3) cannot simultaneously hold, and when $\Gamma$ has a finite relational signature, then $(2)$ implies NP-completeness and (3) implies tractability of its CSP.
\end{thm}

The first statement of Theorem~\ref{thm:main2} follows directly from the proof of Theorem~\ref{thm:main}, with the additional observation that the strong h1 clone homomorphism  defined in Proposition~\ref{prop:h-hard} is in fact a clone homomorphism. When~(3) holds for a reduct, then~(2) cannot hold, because~(3) implies the existence of $f(x,y)\in\Pol(\Gamma)$ and $\alpha\in\overline{\Aut(\Gamma)}$ satisfying the equation $f(x,y)=\alpha f(y,x)$, an  equation impossible to satisfy by projections. In fact, by further analyzing case~(1), using what is known about reducts of equality, one can easily show that it also implies either~(2) or~(3), so that we have the following.

\begin{cor}
For every reduct  $\Gamma$ of a Henson graph $(H_n,E)$, precisely one of the following holds.
\begin{itemize}
\item $\Pol(\Gamma)$ has a uniformly continuous projective clone homomorphism.
\item $\Pol(\Gamma)$ contains $f(x,y)\in\Pol(\Gamma)$ and $\alpha\in\overline{\Aut(\Gamma)}$ such that $f(x,y)=\alpha f(y,x)$.
\end{itemize}
When $\Gamma$ has a finite relational signature, then the first case implies NP-completeness and the second case implies tractability of its CSP.
\end{cor}

\section{Polymorphisms over homogeneous equivalence relations}\label{sect:polymorphisms_equivalence}

We now investigate polymorphisms of reducts of the graphs $(C_n^s,E)$, for $2\leq n,s\leq\omega$, with precisely one of $n,s$ equal to $\omega$. Recall from Section~\ref{sect:prelims} that we write $\Eq$ for the reflexive closure of $E$, that $\Eq$ is an equivalence relation with $n$ classes of size $s$, and that we denote its equivalence classes by $C_i$ for $0\leq i<n$.

Similarly to the case of the Henson graphs, we start with unary polymorphisms in Section~\ref{sect:eq:unary}, reducing the problem to model-complete cores. 

We then turn to higher-arity polymorphisms; here, the organization somewhat differs from the case of the Henson graphs. The role of the NP-hard relation $H$ from the Henson graphs is now taken by the two sources of NP-hardness mentioned in the introduction: the first source being that factoring by the equivalence relation $\Eq$ yields a structure with an NP-hard problem, and the second source being that restriction to some equivalence class yields a structure with an NP-hard problem. In Section~\ref{sect:eq:other}, we show that in fact, one of the two sources always applies for model-complete cores when $2<n<\omega$ or  $2<s<\omega$. Consequently, only the higher-arity polymorphisms of the reducts of $(C_2^\omega,E)$ and $(C_\omega^2,E)$ require deeper investigation using Ramsey theory; this will be dealt with in Sections~\ref{sect:eq:2infinity} and~\ref{sect:eq:infinity2}, respectively. 

\subsection{The unary case: model-complete cores}\label{sect:eq:unary} \

\begin{prop}\label{prop:EndEq1}
Let $\Gamma$ be a reduct of $(C_n^s, E)$, where $1\leq n,s\leq \omega$, and at least one of $n,s$ equals $\omega$. 
Then either $\End(\Gamma)=\overline{\Aut(\Gamma)}=\overline{\Aut(C_n^s,E)}$, or $\End(\Gamma)$ contains an endomorphism onto a clique or an independent set.
\end{prop}

\begin{proof}
Assume that $\End(\Gamma)\neq \overline{\Aut(C_n^s,E)}$, so there is an endomorphism $f$ of $\Gamma$ violating either $E$ or $N$. \smallskip

\textbf{Case 0.} If $n=1$ or $s=1$ then the statement is trivial.\smallskip

\textbf{Case 1.} If $n=s=\omega$, so $\Eq$ has infinitely many infinite classes, we can refer to \cite{equiv-csps}.\smallskip

\textbf{Case 2.} Assume that $1<n<\omega$ and $s=\omega$.

Suppose that $f$ violates $\Eq$ and preserves $N$; then clearly, iterating applications of automorphisms of $(C_n^\omega,E)$ and $f$, we could send any finite subset of $C_n^\omega$ to an independent set in $(C_n^\omega,E)$, contradicting that the number of equivalence classes is the fixed finite number $n$. 

If $f$ preserves both $\Eq$ and $N$, then there exist $a,b$ with $E(a,b)$ and $f(a)=f(b)$. Via a standard iterative argument using topological closure, one then sees that $f$ generates a function whose range is an independent set.

Therefore, it remains to consider the case where $f$ violates $N$. Fix $u,v\in C_n^\omega$ with $N(u,v)$ and $\Eq(f(u),f(v))$.
Without loss of generality we may assume $u\in C_0$ and $v\in C_1$. By Proposition~\ref{prop:canfct-C-high-s-low-n}, we may assume that $f$ is canonical as a function from $(C_n^\omega, E, \prec, u,v)$ to $(C_n^\omega, E, \prec)$. Clearly, $f$ must preserve $\Eq$ on each class $C_i$ with $i>1$, as otherwise canonicity would imply the existence of an infinite independent set in $(C_n^\omega,E)$. For the same reason, $f$ preserves $\Eq$ on each of the four sets
\begin{align*}
C_0^- & :=\{a\in C_0\;|\;a\prec u\}\;, \\ 
C_0^+ & :=\{a\in C_0\;|\;u\prec a\}\;, \\
C_1^- & :=\{a\in C_1\;|\;a\prec v\}\;, \\ 
\text{ and }\; C_1^+ & :=\{a\in C_1\;|\;v\prec a\}.
\end{align*}
If $N$ is not preserved between two sets among $S:=\{C_0^-,C_0^+,C_1^-,C_1^+,C_2,C_3,\ldots\}$, then we pick these two sets along with $n-2$ further sets from $S$ belonging to distinct equivalence classes. The union of this collection induces a copy of $(C_n^\omega,E)$ on which $f$ preserves $\Eq$ but not $N$, and a standard iterative argument shows that $f$ generates a function whose range is contained in a single equivalence class. Hence, we may assume that $N$ is preserved between any two sets in $S$. Since $n$ is finite, this is only possible if $\Eq$ is preserved on $C_0^-\cup C_0^+$ and on $C_1^-\cup C_1^+$. By composing $f$ with an automorphism of $(C_n^\omega,E)$, we may thus assume that $f[C_i^-\cup C_i^+]\subseteq C_i$ for $i\in\{0,1\}$ and that $f$ preserves the classes $C_i$ for $i>1$. Either $f(u)\nin C_0$ or $f(v)\nin C_1$. Assume without loss of generality that $f(u)\in C_i$ where $i>0$. Let $e$ be a self-embedding of $(C_n^\omega,E)$ with range $C_n^\omega\setminus\{v\}$. Then $f\circ e$ preserves all equivalence classes except for the element $u$, which it moves from $C_0$ to $C_i$. Iterating applications of $f\circ e$ and automorphisms, and using topological closure, we obtain a function which joins $C_0$ and $C_i$. By further iteration, we obtain a function which joins all classes.

\smallskip

\textbf{Case 3.} Assume that $s<\omega$ and $n=\omega$.

Suppose that $f$ violates $N$ and preserves $\Eq$; then, by topological closure, $f$ generates a mapping onto a clique. If it preserves both $\Eq$ and $N$, then as above, $f$ generates a function whose range is an independent set.

Therefore, we may assume that $f$ violates $\Eq$. Fix $u,v\in C_\omega^s$ with $E(u,v)$ such that $N(f(u),f(v))$. By Proposition~\ref{prop:canfct-C-high-n-low-s}, we may assume that $f$ is canonical as a function from $(C_\omega^s, E, \prec, u,v)$ to $(C_\omega^s, E, \prec)$. If $f$ preserves $N$, then by topological closure $f$ generates a function whose range induces an independent set. Otherwise, there exist $a,b\in C_\omega^s$ with $N(a,b)$ and $\Eq(f(a),f(b))$. Without loss of generality, $a$ is not contained in the class of $u$ and $v$. Then $\{a'\in C_\omega^s\;|\; \tp(a',b)=\tp(a,b) \text{ in } (C_\omega^s, E, \prec, u, v)\}$ contains an infinite independent set $S$. By canonicity, we have $\Eq(f(a'),f(b))$ for all $a'\in S$, so that $S$ is mapped into a single class. Since this class is finite, there exist $a',a''\in S$ with $f(a')=f(a'')$, and so by topological closure, we can generate a function from $f$ whose range is contained in a single equivalence class.
\end{proof}

If the second case of
Proposition~\ref{prop:EndEq1} applies to a reduct $\Gamma$ of $(C_n^s,E)$, then $\Gamma$ is homomorphically equivalent to a reduct of equality,
and its CSP is understood. In the following sections, we investigate essential polymorphisms of reducts $\Gamma$ of $(C_n^s,E)$ satisfying $\End(\Gamma)=\overline{\Aut(\Gamma)}=\overline{\Aut(C_n^s,E)}$. In particular, such reducts are model-complete cores. The following proposition implies that in the situation where $2<s$ the equivalence relation $\Eq$ is invariant under  $\Pol(\Gamma)$.

\begin{prop}\label{prop:eqpreserved}
Let $\Gamma$ be a reduct of $(C_n^s, E)$, where $1\leq n\leq \omega$ and $2<s \leq \omega$. 
If $\End(\Gamma)=\overline{\Aut(C_n^s,E)}$, then $E$, $N$, and $\Eq$ are preserved by $\Pol(\Gamma)$. 
\end{prop}
\begin{proof}
By Lemma~\ref{lem:arity-reduction}, the condition $\End(\Gamma)=\overline{\Aut(C_n^s,E)}$ implies that all polymorphisms of $\Gamma$ preserve $E$ and $N$, and hence also $\Eq$ since $\Eq(x,y)$ has the primitive positive definition $\exists z\; (E(x,z)\wedge E(z,y))$. Note that we need that the classes contain at least three elements for this definition to work.
\end{proof}

If $s=1$, then $\Eq$ is pp-definable as equality, but if $s=2$ then $\Eq$ is not in general pp-definable; this will account for an additional non-trivial (tractable) case in our analysis.

Since in the situation of Proposition~\ref{prop:eqpreserved}, $\Eq$ is an equivalence relation which is invariant under $\Pol(\Gamma)$, it follows that $\Pol(\Gamma)$ acts naturally on the equivalence classes of $\Eq$: for $f(x_1,\ldots,x_n)\in\Pol(\Gamma)$ and classes $C_{i_1},\ldots,C_{i_n}$ of $\Eq$, the class $f(C_{i_1},\ldots,C_{i_n})$ is then defined as the equivalence class of $f(c_{i_1},\ldots,c_{i_n})$, where $c_{i_1}\in C_{i_1},\ldots,c_{i_n}\in C_{i_n}$ are arbitrary.

Moreover, if we fix any class $C$ of $\Eq$ and expand the structure $\Gamma$ by the predicate $C$ to a structure $(\Gamma,C)$, then $\Pol(\Gamma, C)$ acts naturally on $C$ via restriction of its functions. Since $\Aut(C_n^s,E)$ can flip any two equivalence classes, all such actions are isomorphic, i.e., for any two classes $C,C'$ there exists a bijection $i\colon C\To C'$ such that 
\begin{align*}
\Pol(\Gamma,C') & =\{(x_1,\dots,x_n) \mapsto i(f(i^{-1}(x_1),\ldots,i^{-1}(x_n))) \mid f\in\Pol(\Gamma,C)\} \\
\Pol(\Gamma,C) & =\{(x_1,\dots,x_n) \mapsto i^{-1}(f(i(x_1),\ldots,i(x_n))) \mid f\in\Pol(\Gamma,C')\}
\end{align*}
(in fact any bijection $i$ works, since any permutation on $C$ extends to an automorphism of $(C_n^s,E)$ which fixes the elements of $C'$ pointwise). It is for this reason that in the following, it will not matter if we make statements about all such actions, or a single action.


In the following sections, we analyze these two types of actions.

\subsection{The case $2<n<\omega$ or $2<s<\omega$}\label{sect:eq:other} It turns out that in these cases, one of the two types of actions always yields hardness of the CSP. We are going to use the following fact about function clones on a finite domain.

\begin{prop}[from~\cite{HaddadRosenberg}]\label{prop:finite3}
Every function clone on a finite domain of at least three elements which contains all permutations as well as an essential function contains a unary constant function.
\end{prop}

We can immediately apply this fact to the action of $\Pol(\Gamma)$ on the equivalence classes, when there are more than two, but finitely many classes.

\begin{prop}\label{prop:n>2}
Let $\Gamma$ be a reduct of $(C_n^\omega, E)$, where $2<n<\omega$, such that $\End(\Gamma)=\overline{\Aut(C_n^\omega,E)}$. Then the action of $\Pol(\Gamma)$ on the equivalence classes of $\Eq$ has no essential and no constant operation.
\end{prop}
\begin{proof}
The action has no constant operation because $N$ is preserved. Therefore, it cannot have an essential operation either, by Proposition~\ref{prop:finite3}.
\end{proof}

Similarly, we can apply the same fact to the action of $\Pol(\Gamma,C)$ on any equivalence class $C$ on $C_\omega^s$ if this class is finite and has more than two elements.

\begin{prop}\label{prop:s>2}
Let $\Gamma$ be a reduct of $(C_\omega^s, E)$, where $2<s<\omega$, such that $\End(\Gamma)=\overline{\Aut(C_\omega^s,E)}$. Then for any equivalence class $C$ of $\Eq$, the action of $\Pol(\Gamma,C)$ on $C$ has no essential and no constant operation.
\end{prop}
\begin{proof}
The action has no constant operation because $E$ is preserved. Therefore, it cannot have an essential operation either, by Proposition~\ref{prop:finite3}.
\end{proof}\smallskip

\subsection{The case of two infinite classes: $n=2$ and $s=\omega$}~\label{sect:eq:2infinity} The following proposition states that either one of the two sources of hardness applies, or $\Pol(\Gamma)$ contains a ternary canonical function with a certain behaviour.

\begin{prop}\label{prop:2omega}
Let $\Gamma$ be a reduct of $(C_2^\omega, E)$ such that $\End(\Gamma)=\overline{\Aut(C_2^\omega,E)}$.
Then one of the following holds:
\begin{itemize} 
\item the action of $\Pol(\Gamma)$ on the equivalence classes of $\Eq$ has no essential function;
\item the action of $\Pol(\Gamma,C)$ on some (or any) class $C$ has no essential function;
\item $\Pol(\Gamma)$ contains a canonical ternary injection of behaviour minority which is hyperplanely of behaviour  balanced xnor.
\end{itemize}
\end{prop}
To prove the proposition, we need to recall a special case of Post's classical result about function clones acting on a two-element set. Comparing this statement with Proposition~\ref{prop:finite3} sheds light on why the case of this section is more involved than the cases of the preceding section. 

\begin{proposition}[Post~\cite{Post}]\label{prop:Post}
Every function clone with domain $\{0,1\}$ containing both permutations of $\{0,1\}$ as well as an essential function contains a unary constant operation or the ternary addition modulo 2.
\end{proposition}

We moreover require the following result on polymorphism clones on a countable set.

\begin{proposition}[from~\cite{ecsps}]\label{prop:BK}
Every polymorphism clone on a countably infinite set which contains all permutations as well as an essential operation contains a binary injection.
\end{proposition}

We now combine these two results to a proof of Proposition~\ref{prop:2omega}.

\begin{proof}[of Proposition~\ref{prop:2omega}] 
Recall that the equivalence classes of $\Eq$ are denoted by $C_0$ and $C_1$, and that $E$, $N$, and $\Eq$ are preserved by the functions of $\Pol(\Gamma)$, by Proposition~\ref{prop:eqpreserved}.
Suppose that the first statement of the proposition does not hold. Then by Proposition~\ref{prop:Post}, the action of $\Pol(\Gamma)$  on $\{C_0,C_1\}$ contains a unary constant operation, or a function which behaves like  ternary addition modulo 2. The first case  is impossible since the unary functions in $\Pol(\Gamma)$ preserve $N$, so the latter case holds and $\Pol(\Gamma)$ contains a ternary function $g$ which acts like $x+y+z$ modulo $2$ on the classes. 

Suppose now in addition that the second statement of the proposition does not hold either, and fix some equivalence class $C$. Since the action of $\Pol(\Gamma,C)$ on $C$ contains all permutations of $C$, by Proposition~\ref{prop:BK} it also contains a binary injection. Therefore $\Pol(\Gamma)$ contains for each $i\in\{0,1\}$ a binary function $f_i$ whose restriction to $C_i$ is an injection on this set.

We claim that there is a single function $f\in\Pol(\Gamma)$ which has this property for both $C_0$ and $C_1$. Note that since $N$ is preserved by $f_0$, it maps $C_1$ into itself. If $f_0$ is essential on $C_1$, then Proposition~\ref{prop:BK} implies that  together with all permutations which fix the classes, it generates a function which is injective on $C_1$; this function is then injective on both classes $C_0, C_1$. So assume that $f_0$ is not essential on $C_1$, say without loss of generality that it depends only on the first coordinate (and injectively so, since it preserves $E$). Then $f_0(f_1(x,y),f_0(x,y))$ preserves both classes and is injective on each of them.

By Proposition~\ref{prop:canfct}, we may assume that $f$ is canonical as a function from $(C_2^\omega,E,\prec)\times (C_2^\omega,E,\prec)$ to $(C_2^\omega,E,\prec)$. We claim  that $f$ is also canonical as a function from $(C_2^\omega,E)\times(C_2^\omega,E)$ to $(C_2^\omega,E)$. To prove this, it suffices to show that if $u,v,u',v'\in C_2^\omega\times C_2^\omega$ are so that $(u,v)$ and $(u',v')$ have the same type in $(C_2^\omega,E)\times (C_2^\omega,E)$, then $(f(u),f(v))$ and $(f({u'}),f(v'))$ have the same type in $(C_2^\omega,E)$. 
There exist $u'',v''\in C_2^\omega\times C_2^\omega$ such that $(u',v')$ and $(u'',v'')$ have the same type in $(C_2^\omega,E,\prec)\times (C_2^\omega,E,\prec)$ and such that ${\Eq}{\Eq}(u,u'')$ and ${\Eq}{\Eq}(v,v'')$; by the canonicity of $f$ as a function from $(C_2^\omega,E,\prec)\times (C_2^\omega,E,\prec)$ to $(C_2^\omega,E,\prec)$, it suffices to show that $(f(u),f(v))$ and $(f({u''}),f(v''))$ have the same type in $(C_2^\omega,E)$. Since $\Eq$ is preserved, we have  $\Eq(f(u),f(u''))$ and $\Eq(f(v),f(v''))$, and so $\Eq(f(u),f(v))$ implies $\Eq(f(u''),f(v'')))$ and vice-versa, by the transitivity of $\Eq$. Failure of canonicity can therefore only happen if $\Eq(f(u),f(v))$ and $\Eq(f(u''),f(v'')))$, and precisely one of $f(u)=f(v)$ and $f(u'')=f(v'')$ holds, say without loss of generality the former. But then picking any $v'''\in C_2^\omega\times C_2^\omega$ distinct from $v$ such that ${\Eq}{\Eq}(v,v''')$ and such that the type of $(u,v)$ equals the type of $(u,v''')$ in $(C_2^\omega,E,\prec)\times (C_2^\omega,E,\prec)$ shows that $f(v)=f(u)=f(v''')$ by canonicity, contradicting the fact that $f$ is injective on each equivalence class.

We analyze the behaviour of the canonical function $f\colon (C_2^\omega,E)\times (C_2^\omega,E) \To(C_2^\omega,E)$. Because $E$ and $N$ are preserved, we have $f(E,E)=E$ and $f(N,N)=N$. Moreover, because $f$ is injective on the classes, and because $\Eq$ is preserved, we have $f(=,E)=f(E,=)=E$.

We next claim that either $f(\cdot ,N)=N$ or $f(N,\cdot)=N$. Otherwise, there exist $Q,P\in\{E,=\}$ such that $f(Q,N)\neq N$ and $f(N,P)\neq N$. Pick $u,v,w\in (C_2^\omega)^2$ such that $\QN(u,v), \NPe(v,w)$, and $\NN(u,w)$. Then $\Eq(f(u),f(w))$ and $N(f(u),f(w))$, a contradiction.

Assume henceforth without loss of generality that $f(N,\cdot)=N$. 
Then $f(P,N)\neq N$ for $P\neq N$, because there are only two equivalence classes. 
Moreover, $f(E,N)= {=}$ or $f(=,N)= {=}$ would imply that $f$ is not injective on the classes, so we have $f(E,N) = f(=,N) = E$.

Summarizing, $f$ is a binary injection of behaviour  $p_1$, balanced in the first argument, and $E$-dominated in the second argument.

Let $q\in\Pol(\Gamma)$ be any ternary injection (for example, $(x,y,z)\mapsto f(x,f(y,z))$), and set $h(x,y,z):=f(g(x,y,z), q(x,y,z))$. 
We now show that $h$ is canonical by establishing all type conditions satisfied by it. To this end, we use the behaviour of $f$ and the fact that $g$ acts like $x+y+z$ modulo $2$ on the classes. The latter fact implies that $g$ satisfies certain type conditions as well, as is easily verified: $g(\Eq,\Eq,N)=g(\Eq,N,\Eq)=g(N,\Eq,\Eq)=N$, $g(\Eq,\Eq,\Eq)=\Eq$, and moreover $g(\Eq,N,N)=\Eq$, $g(N,\Eq,N)=\Eq$, and $g(N,N,\Eq)=\Eq$. In the following table, $u,v,w\in (C_2^\omega)^2$ are three pairs for which $\eqeqeq(u,v,w)$ does not hold, and according to the type of $(u,v,w)$ in $(C_2^\omega,E)\times (C_2^\omega,E)$ the type of $h(u,v,w)$ in $(C_2^\omega,E)$  is computed. By the symmetry of the type conditions of $g$ listed above, and since of $q$ we only use injectivity so that ${\neq}(q(u,v,w))$ holds, the value of a triple of types does not change if its components are permuted. Therefore, we only list all possibilities of types for $(u,v,w)$ up to permutations.
\[
\begin{array}{ccc}
\mbox{$\tp(u,v,w)$} & \mbox{$\tp(g(u,v,w), q(u,v,w))$} &  \mbox{$\tp(h(u,v,w))$} \\
\EEE & (E,\neq) & E\\
\NNN & (N,\neq) & N\\
\EEN & (N,\neq)& N\\
\ENN & (\Eq,\neq) & E\\
\eqEE & (\Eq,\neq) & E \\
\eqNN & (\Eq,\neq) & E \\
\eqEN & (N,\neq) & N \\
\eqeqE & (\Eq,\neq) & E \\
\eqeqN & (N,\neq) & N \\
\end{array}
\]
So $h$ acts like a minority which is hyperplanely of behaviour balanced xnor.
\end{proof}\smallskip

\subsection{The case of infinitely many classes of size two: $n=\omega$ and $s=2$}~\label{sect:eq:infinity2} Recall that in this situation, Proposition~\ref{prop:eqpreserved} does not apply, and $\Eq$ might not be pp-definable in a reduct $\Gamma$ of  $(C^2_\omega,E)$, even if $\Gamma$ is a model-complete core. We first show that if this happens, then $\Pol(\Gamma)$ contains a certain binary canonical function (Proposition~\ref{prop:al-jabr}). We then show, in Proposition~\ref{prop:omega2}, that if $\Eq$ does have a primitive positive definition in $\Gamma$, then either one of the two sources of hardness applies, or $\Pol(\Gamma)$ contains a ternary function of a certain behaviour.

\begin{proposition}
Let $\Gamma$ be a reduct of $(C_\omega^2, E)$ such that $\End(\Gamma)=\overline{\Aut(C_\omega^2,E)}$, and such that $\Eq$ is not pp-definable. Then $\Gamma$ enjoys a binary canonical polymorphism of behaviour $\mini$ which is $N$-dominated.
\label{prop:al-jabr}
\end{proposition}
\begin{proof}
By Theorem~\ref{conf:thm:inv-pol}, $\Gamma$ has a polymorphism $f$ which violates $\Eq$. By the assumption, all endomorphisms preserve $E$ and $N$, and hence, by Lemma~\ref{lem:arity-reduction}, so does $f$. By the same lemma, because $\Eq$ consists of two orbits with respect to the action of the automorphism group of $(C_\omega^2,E)$ on pairs, we may assume that $f$ is binary.


We refer to sets of the form $C \times D$, where $C, D$ are equivalence classes of $\Eq$, as \emph{squares}.  Note that each square is the disjoint union of precisely two edges in the product graph $({C_\omega^2},E)^2$, and that each of these edges is mapped by $f$ to an edge in $({C_\omega^2},E)$, since $f$ preserves $E$. We say that $f$ \emph{splits} a square if it does not map this square into a single class; in this case, it necessarily maps it into two classes, by the previous observation. 

By composing $f$ with automorphisms from the inside, we may assume that $f$ violates $\Eq$ on a square of the form $C\times C$. Writing $C=\{u,v\}$, we may invoke Proposition~\ref{prop:canfct-C-high-n-low-s} and assume that $f$ is canonical when viewed as a function  from $({C_\omega^2},E,\prec,u,v)\times ({C_\omega^2},E,\prec,u,v)$ to $({C_\omega^2},E,\prec)$. We set  $S:= C_\omega^2\setminus C$ and $S':=\{x\in S\;|\; x\prec u\wedge x\prec v\}$.

We now distinguish two cases to show the following.\smallskip

{\bf Claim.} $f$ generates a binary function $f'$ which still splits $C\times C$ and satisfies either  $f'(N,\cdot)=N$ or $f'(\cdot, N)=N$.\smallskip 

{\bf Case 1:}
We first assume that $f$ splits a square within $(S')^2$. Then, by canonicity, it splits all squares within $(S')^2$. In that case, the function $f(e(x),e(y))$, where $e$ is a self-embedding from $({C_\omega^2},E,\prec)$  onto the structure induced therein by $S'$, is canonical whilst splitting all squares. Replacing $f$ by this function, we henceforth assume $f$ to split all squares. The constants $u,v$ which were introduced to witness the occurrence of a splitting will not be of importance to us anymore in the further discussion of this case.


%

The function $g$ on $({C_\omega^2})^2$ sending every pair $(x,y)$ to the pair $(f(x,y),f(y,x))$ is canonical when viewed as a function 
$$
({C_\omega^2},E,\prec)\times ({C_\omega^2},E,\prec)\To(({C_\omega^2})^2,\EE,\EN,\NE,\NN,\Eeq,\eqE,\Neq,\eqN, {\prec\prec})\; ,
$$
by the canonicity of $f$. In the following, we analyse the behaviour of $g$.

We start by observing that every square consists of an \emph{upward edge} and a \emph{downward edge} in $(C_\omega^2,E,\prec)^2$, the orientation being induced by the order $\prec$: by the upward edge $(p,q)\in \EE$ we refer to the one on which the order $\prec$ agrees in both coordinates between $p$ and $q$, and by the downward one we refer to the other edge in the square (on which $\prec$ disagrees between the coordinates). Let $U$ be the set of points contained in an  upward edge, and $V$ the set of points contained in a downward edge, so that $({C_\omega^2})^2$ is the disjoint union of $U$ and $V$. We are going to verify the following properties of $g$:
\begin{itemize}
\item[(i)] $g[U]\subseteq U$ and $g[V]\subseteq V$.
\item[(ii)] $\Eeq(p,q)$, $\eqE(p,q)$, and $\NN(p,q)$ all imply $\NN(g(p),g(q))$, for all $p,q\in ({C_\omega^2})^2$.
\item[(iii)] $\NN(g(p),g(q))$ for all $p\in U$ and all $q\in V$.
\item[(iv)] On $U$ as well as on $V$, either $f(N,\cdot)=N$ or $f(\cdot,N)=N$ holds.
\end{itemize}
Of Property~(i), we give the argument that $g[U]\subseteq U$; proving $g[V]\subseteq V$ is similar. Let $p=(p_1,p_2), q=(q_1,q_2)\in ({C_\omega^2})^2$ be so that $(p,q)$ forms an upward edge, and say that $p_1\prec q_1$ and $p_2\prec q_2$. 
If $f(p_1,p_2)\prec f(q_1,q_2)$, then by canonicity also $f(p_2,p_1)\prec f(q_2,q_1)$, and so $(g(p),g(q))$ is related by $\prec$ in both coordinates. Since $f$ preserves $E$, $(g(p),g(q))$ is also related by $E$ in both coordinates, and hence $(g(p),g(q))$ forms an upward edge. If $f(q_1,q_2)\prec f(p_1,p_2)$, then a similar argument shows that $(g(p),g(q))$ forms an upward edge.

Property~(ii) follows since $f$ preserves $N$ and because $f$ splits all squares.

For~(iii), suppose that $\NN(g(p),g(q))$ does not hold for some $p\in U$ and $q\in V$. We cannot have $\EE(p,q)$ since $p$ is contained in an upward and $q$ in a downward edge, so by~(ii), $p$ and $q$ must be related by $N$ in one coordinate. Say we have $\Neq(p,q)$; the other situations are handled similarly. Pick $q'\in V$ distinct from $q$ such that the types of $(p,q)$ and $(p,q')$ in $({C_\omega^2},E,\prec)\times ({C_\omega^2},E,\prec)$ coincide. Then, by canonicity, we have that $g(p),g(q)$ are equivalent with respect to $\Eq$ in the same coordinate as $g(p),g(q')$; hence, so are $g(q),g(q')$, by the transitivity of $\Eq$. By canonicity, we then know that for the unique $q''\in V$ with $\Eeq(p,q'')$, we have that $g(q)$ and $g(q'')$ are equivalent in that very same coordinate, since the types of $(q,q')$ and either $(q,q'')$ or $(q'',q)$ agree. Again by transitivity, $g(p),g(q'')$ are then equivalent in that coordinate, contradicting~(ii).

To see Property~(iv), suppose that both $f(N,\cdot)=N$ and $f(\cdot,N)=N$ do not hold on $U$. Then there exist $p,q,p',q' \in U$ such that $p,q$ are related by $N$ in the first coordinate, $p',q'$ are related by $N$ in the second coordinate, and $\Eq(f(p),f(q))$ and $\Eq(f(p'),f(q'))$ hold. But then we could pick $q''\in U$ such that $\tp(p,q'')=\tp(p',q')$ in $({C_\omega^2},E,\prec)\times ({C_\omega^2},E,\prec)$; any such $q''$ necessarily satisfies $\NN(q,q'')$. By canonicity we would have $\Eq(f(p),f(q''))$, and hence, by transitivity, this would imply $\Eq(f(q),f(q''))$, a contradiction since $f$ preserves $N$.

Now suppose that $f(N,\cdot)=N$ on both $U$ and $V$. Then the function $f'(x,y):=f(g(x,y))=f(f(x,y),f(y,x))$ has the same property  by~(iii), and moreover it splits all squares, so we are done. If  $f(\cdot,N)=N$ on both $U$ and $V$, then by symmetry $f'(x,y):=f(f(y,x),f(x,y))$ has the same property everywhere and splits all squares. It remains to consider the case where, say, $f(N,\cdot)=N$ on $U$ and $f(\cdot,N)=N$ on $V$. 
The function $f'(x,y):=f\circ g$ satisfies $f'(N,\cdot)=N$. To see this, let $p,q\in ({C_\omega^2})^2$ be related by $N$ in the first coordinate. If $p,q\in U$, then $g(p), g(q)$ are related by $N$ in the first coordinate, and because $g[U]\subseteq U$, we have $N(f(g(p)), f(g(q)))$. When $p\in U$ and $q\in V$, then $\NN(g(p),g(q))$ by~(iii), and so $N(f'(p), f'(q))$ since $f$ preserves $N$. Finally, if $p,q \in V$, then $g(p),g(q)$ are related by $N$ in the second coordinate, and using $g[V]\subseteq V$, we see that $N(f'(p), f'(q))$. Since $f'$ moreover splits all squares by~(ii), we are done.\smallskip

{\bf Case 2:} 
Assume now that $f$ does not split any square within $(S')^2$. We claim that $f(N,\cdot)=N$ or $f(\cdot,N)=N$ on $(S')^2$: otherwise, there would exist $p,q,p',q' \in (S')^2$ such that $p,q$ are related by $N$ in the first coordinate, $p',q'$ are related by $N$ in the second coordinate, and $\Eq(f(p),f(q))$ and $\Eq(f(p'),f(q'))$ hold. But then we could pick $q''\in (S')^2$ such that $\tp(p,q'')=\tp(p',q')$ in $({C_\omega^2},E,\prec,u,v)\times ({C_\omega^2},E,\prec,u,v)$; any such $q''$ necessarily satisfies $\NN(q,q'')$. By canonicity we would have $\Eq(f(p),f(q''))$, and hence, by transitivity, this would imply $\Eq(f(q),f(q''))$, a contradiction since $f$ preserves $N$. We assume without loss of generality that $f(N,\cdot)=N$ on $(S')^2$.

The function $f(e(x),e(y))$, where $e$ is a self-embedding from $({C_\omega^2},E,u,v)$  onto the structure induced therein by $S'\cup\{u,v\}$, still splits $C\times C$, splits no square within $S\times S$, and satisfies $f(N,\cdot)=N$ on $S^2$.
Invoking Proposition~\ref{prop:canfct-C-high-n-low-s} again, we may assume that that function is moreover canonical as a function from $({C_\omega^2},E,\prec,u,v)\times ({C_\omega^2},E,\prec,u,v)$ to $({C_\omega^2},E,\prec)$. Replacing $f$ by this function, we may therefore henceforth assume that $f$ itself enjoys the listed properties.

Note that the function $f(e(x),e(y))$ as in the preceding paragraph does not distinguish between elements of $S'$ and those of $S\setminus S'$, since $e$ sends the entire set $S$ into $S'$ before $f$ is applied. In particular, it has the property that for any $p\in C\times C$ and any $q\in (C_\omega^2)^2$, the type of $(f(p),f(q))$ in $({C_\omega^2},E)$ only depends on the type of $(p,q)$ in $({C_\omega^2},E,u,v)\times ({C_\omega^2},E,u,v)$, and not on the more precise type of $(p,q)$ in  $({C_\omega^2},E,u,v,\prec)\times ({C_\omega^2},E,u,v,\prec)$ (which does distinguish between $S'$ and $S\setminus S'$).

We now distinguish two subcases to show that $f$ generates a binary function $f'$ which splits $C\times C$ and such that $f'(N,\cdot)=N$ everywhere, thus proving the claim.

{Case 2.1:} If $f(N,\cdot)=N$ on $S \times C$, then
by canonicity and the remark above, one easily concludes $N(f(p),f(q))$ for all $p\in C\times C$ and all $q\in S \times C$, so that altogether $f(N,\cdot)=N$ everywhere. Hence, setting $f':=f$ we have achieved our goal.

{Case 2.2:} If $f(N,\cdot)=N$ does not hold on $S \times C$, then  there exists $c\in S\times C$ such that $N(f(c),f(q))$ for any $q\in 
S^2$. To see this, we can 
pick any $c\in S\times C$ so that there exists $q'\in S\times C$ related to $c$ by $N$ in the first coordinate and such that $\Eq(f(c),f(q'))$. Then, if there existed $q\in S^2$ with $\Eq(f(c),f(q))$, we would have $\Eq(f(q),f(q'))$; replacing $q'$ by $q''\in S\times C$ such that $\tp(c,q')=\tp(c,q'')$ in $({C_\omega^2},E,\prec,u,v)$ and such that $q',q''$ are related by $N$ in both coordinates, this would yield a contradiction to the preservation of $N$.

We are going to check the following properties of the function $g$ on $({C_\omega^2})^2$ defined by $(x,y)\mapsto (x,f(x,y))$.
\begin{itemize}
\item[(i)] Whenever $p,q\in ({C_\omega^2})^2$ are related by $N$ in the first coordinate, then so are $g(p), g(q)$.
\item[(ii)] If $p\in ({C_\omega^2})^2$, and $q\in S^2$ is related to $p$ by $N$ in the first coordinate, then $\NN(g(p),g(q))$.
\item[(iii)] Writing $a:=(u,u)$ and $b:=(v,u)$, we have $\Eeq(a,b)$ and $\EN(g(a),g(b))$.
\end{itemize}
Property~(i) is obvious from the definition of $g$. Property~(ii) is clear if $p\in {C_\omega^2}\times C$, since in that case $\NN(p,q)$ and since $f$ preserves $N$. If $p\in S^2$, then it follows from the fact that $f(N,\cdot)=N$ on $S$. Finally, consider the case where $p\in C \times S$. If we had $\Eq(f(p),f(q))$, then picking $q'\in S^2$ such that $\Neq(q,q')$ and such that $\tp(p,q)=\tp(p,q')$ in $({C_\omega^2},E,\prec)\times ({C_\omega^2},E,\prec)$, we would get $\Eq(f(p),f(q'))$ by canonicity, and so $\Eq(f(q),f(q'))$, contradicting that $f(N,\cdot)=N$ on $S$. Property~(iii) just restates that $f$ splits $C\times C$.

Let $e_1,e_2$ be self-embeddings of $({C_\omega^2},E)$ such that the range of $(e_1,e_2)\circ g$
is contained in $S\times {C_\omega^2}$ and such that $(e_1,e_2)\circ g(a)=c$. Using that assumption, $g':=g\circ (e_1,e_2)\circ g$ clearly also satisfies~(i) and~(ii). Moreover, since $(e_1,e_2)\circ g(a)=c$, and since $\EN((e_1,e_2)\circ g(a),(e_1,e_2)\circ g(b))$, we have $(e_1,e_2)\circ g(b)\in S^2$; this implies $\EN(g'(a),g'(b))$, since $N(f(c),f(q))$ for all $q\in S^2$. Hence, $g'$ still satisfies~(iii).

We then pick a pair $(e_1',e_2')$ of self-embeddings of $({C_\omega^2},E)$ with $(e_1',e_2')\circ g'(b)=c$, and consequently $(e_1',e_2')\circ g'(a)\in S^2$. Then $g'':=g\circ (e_1',e_2')\circ g'=g\circ (e_1',e_2')\circ g\circ (e_1,e_2)\circ g$ has the property that whenever $p,q\in ({C_\omega^2})^2$ are related by $N$ in the first coordinate, then $\NN(g''(p),g''(q))$; this is because every point went
through $S^2$ in one of the applications of $g$, and because of~(ii). Moreover, we have $\EN(g''(a),g''(b))$.

Setting $f'$ to be the projection of $g''$ onto the second coordinate then completes the proof.


{\bf Wrap-up.} Replacing $f$ by $f'$ from the claim, we thus henceforth assume that $f(N,\cdot)=N$. For the function $h$ on $({C_\omega^2})^2$ defined by $(x,y)\mapsto (f(x,y),f(y,x))$, we are going to prove the following properties.
\begin{itemize}
\item[(i)] If $p,q\in ({C_\omega^2})^2$ are related by $N$ in some coordinate, then $h(p), h(q)$
are related by $N$ in the same coordinate.
\item[(ii)] There are $p',q'\in ({C_\omega^2})^2$ with $\Eeq(p',q')$ such that $h(p'), h(q')$ are related
by $N$ in the first coordinate. 
\item[(iii)] There are $p'',q''\in ({C_\omega^2})^2$ with $\EN(p'',q'')$ such that $\NN(h(p''), h(q''))$.
\item[(iv)] There are $p''',q''' \in ({C_\omega^2})^2$ with $\eqN(p'',q'')$ such that $\NN(h(p'''), h(q'''))$.
\end{itemize}
Property~(i) is obvious because $f(N,\cdot)=N$, and~(ii) follows because $f$ splits a square. To see~(iii),  we first observe that there exist $p,q\in ({C_\omega^2})^2$ with equal first coordinate and such that $h(p),h(q)$ are related by $N$ in the first coordinate: simply pick $p,p'$ with $\eqE(p,p')$ within the square that is split; then $\NN(h(p),h(p'))$, and so for any $q\in ({C_\omega^2})^2$ with $\eqN(p,q)$ and $\eqN(p',q)$ we have that $h(q)$ must be related by $N$ in the first coordinate to either $h(p)$ or $h(p')$, showing the observation. Now fix $p,q$ with this property, and pick $v\in ({C_\omega^2})^2$ with
$\EN(p,v)$ and $\EN(q,v)$. Then $h(v)$ is related to $h(p)$ and $h(q)$ by $N$ in the
second coordinate by~(i), but also  necessarily to one of them in the first coordinate, showing~(iii). The proof of~(iv) is similar.


Using these properties, we first construct, by composition and topological closure, a function $h'$ on $({C_\omega^2})^2$ which yields $\NN(h'(p),h'(q))$ for all $p,q\in ({C_\omega^2})^2$ which are related by $N$ in at least one coordinate. To do this, let $\{(p_i,q_i)\;|\; i>0\}$ be an enumeration of all pairs in $({C_\omega^2})^2$ which are related by $N$ in at least one coordinate. We proceed inductively, constructing functions $h_0,h_1,\ldots$ with the property that: $\NN(h_n(p_j),h_n(q_j))$, for all $0<j\leq n$, and $h_n(p_j)$ and $h_n(q_j)$ are related by $N$ in at least one coordinate, for all $j>n$. For the base case, we set $h_0:=h$ (note that the first conjunct of the inductive hypothesis acts here on an empty set of pairs). 
Suppose we have already constructed $h_n$. Then $h_n(p_{n+1})$ and $h_n(q_{n+1})$ are related by $N$ in at least one coordinate. If $\NN(h_n(p_{n+1}),h_n(q_{n+1}))$, then we set $h_{n+1}:=h_n$. If $\EN(h_n(p_{n+1}),h_n(q_{n+1}))$, then let $(\alpha,\beta)$ be a pair of automorphisms of $({C_\omega^2},E)$ such that $(\alpha,\beta)(h_n(p_{n+1}))=p''$ (from~(iii)), and $(\alpha,\beta)(h_n(q_{n+1}))=q''$. Setting $h_{n+1}:=h\circ (\alpha,\beta)\circ h_n$ then yields the desired property for $(p_{n+1},q_{n+1})$. If $\NE(h_n(p_{n+1}),h_n(q_{n+1}))$, then $\EN(\pi\circ h_n(p_{n+1}),\pi\circ h_n(q_{n+1}))$, where $\pi\colon ({C_\omega^2})^2\to ({C_\omega^2})^2$ is defined by $(x,y)\mapsto (y,x)$ (i.e., $\pi$ is a pair of projections); we can thus proceed as before. The cases $\eqN(h_n(p_{n+1}),h_n(q_{n+1}))$ and $\Neq(h_n(p_{n+1}),h_n(q_{n+1}))$ are treated similarly, using~(iv) instead of~(iii). By topological closure, we obtain the function $h'$. 

Setting  $h'':=h'\circ h$, we retain the defining property of $h'$ by~(i), but moreover have $\NN(h''(p'),h''(q'))$, for the pair $(p',q')$ from~(ii).

The function $g_0:=f\circ h''$ then satisfies $g_0(N,\cdot)=g_0(\cdot,N)=N$, and moreover satisfies $N(g(p'),g(q'))$, since $\NN(h''(p'),h''(q'))$ and since $f$ preserves $N$.

Let $\{(p_i,q_i)\;|\; i\geq 0\}$ be an enumeration of all pairs in $({C_\omega^2})^2$ related by $\Eeq$, where $(p_0,q_0)=(p',q')$. As above, we obtain, by composition and topological closure, for every $i\geq 0$ a function $g_i$ which satisfies $g_i(N,\cdot)=g_i(\cdot,N)=N$ and such that $N(g_i(p_i),g_i(q_i))$. Setting $t_0:=g_0$, and $t_{n+1}:=f(t_n(x,y),g_{n+1}(x,y))$ for all $n\geq 0$, we obtain binary functions $t_0,t_1,\ldots$ satisfying $t_i(N,\cdot)=t_i(\cdot,N)=N$ and with the property that $N(t_i(p_j),t_i(q_j))$ for all $j\leq i$. By topological closure, we obtain a binary function $t$ satisfying $t(N,\cdot)=t(\cdot,N)=N$ and  $N(t(p),t(q))$ for all $p,q\in ({C_\omega^2})^2$ with $\Eeq(p,q)$. This function clearly has behaviour $\mini$ and is $N$-dominated in the first argument; since it preserves $E$, these properties also imply that it is $N$-dominated in the second argument.
\end{proof}

We now turn to the case where $\Eq$ is pp-definable in a reduct $\Gamma$, so that $\Pol(\Gamma)$ acts on its equivalence classes.

\begin{prop}\label{prop:omega2}
Let $\Gamma$ be a reduct of $(C_\omega^2, E)$ such that $\End(\Gamma)=\overline{\Aut(C_\omega^2,E)}$ and such that $\Eq$ is pp-definable. 
Then one of the following holds:
\begin{itemize} 
\item the action of $\Pol(\Gamma)$ on the equivalence classes of $\Eq$ has no essential function;
\item the action of $\Pol(\Gamma,C)$ on some (or any) equivalence class of $C$ has no essential function;
\item $\Pol(\Gamma)$ contains a ternary canonical function $h$ such that $h(N,\cdot,\cdot)=h(\cdot,N,\cdot)=h(\cdot,\cdot,N)=N$ which behaves like a minority on $\{E,=\}$ (so $h(E,=,=)=E$ etc.).
\end{itemize}
\end{prop}

To prove the proposition, we are again going to make use of Propositions~\ref{prop:Post} and~\ref{prop:BK}, and the following lemma. We are going to say that a ternary function $f$ on $C_\omega^2$ behaves like $x+y+z$ modulo 2 on an equivalence class $C=\{0,1\}$ of $\Eq$ if the restriction of $f$ to $C$ is of the form $\alpha\circ g_C$, where $\alpha\in\Aut(C_\omega^2,E)$ and $g_C$ is the ternary function on $C$ defined by $g_C(x,y,z)=x+y+z$ modulo 2. Note that this property can be expressed in terms of type conditions satisfied on $C$: namely, $f$ behaves like $x+y+z$ modulo 2 on $C$ if and only if it satisfies $f(E,E,E)=E$, $f(E,E,=)=f(E,=,E)=f(=,E,E)={=}$, and $f(E,=,=)=f(=,=,E)=f(=,E,=)={E}$ on $C$. In other words, $f$ behaves like a minority on the types $\{E,=\}$.

\begin{lem}\label{lem:propagating+}
Let $\Gamma$ be a reduct of $(C_\omega^2, E)$ such that $\End(\Gamma)=\overline{\Aut(C_\omega^2,E)}$, $\Eq$ is pp-definable, and $\Pol(\Gamma)$ contains a ternary function which behaves like $x+y+z$ modulo 2 on some equivalence class. Then $\Pol(\Gamma)$ contains a ternary function which behaves like $x+y+z$ modulo 2 on all equivalence classes.
\end{lem}
\begin{proof}
Let $C_0, C_1,\ldots$ be the equivalence classes of $\Eq$. We show, by induction over $n$, that for all $n\in\omega$, $\Pol(\Gamma)$ contains a function $g_n$ which equals $x+y+z$ modulo 2 on each class $C_0,\ldots,C_n$. The lemma then follows by a standard compactness argument: by $\omega$-categoricity, there exist $\alpha_n\in\Aut(C_\omega^2,E)$, for $n\in\omega$, such that $(\alpha_n\circ g_n)_{n\in\omega}$ converges to a function $g\in\Pol(\Gamma)$ (cf.~for example the proof of Proposition~\ref{prop:redendo}). That function then has the desired property: for every $i\in\omega$, there exists $n>i$ such that $g$ agrees with $\alpha_n \circ g_n$ on $C_i$, and hence it behaves like  $x+y+z$ modulo 2 on $C_i$.

For the base case $n=0$, the statement follows from the assumption of the lemma. Now suppose it holds for $n$. By the assumption that $\End(\Gamma)=\overline{\Aut(D,E)}$, we may assume that $g_n(x,x,x)=x$ for all $x\in C_0\cup \cdots \cup C_{n+1}$, and in particular $g_n$ preserves each of the classes $C_0, \ldots, C_{n+1}$. In particular, the restriction of $g_n$ to any $C_i$ with $0\leq i\leq n$ actually equals the function $x+y+z$ modulo 2 on that class.

Assume first that $g_n$ is not essential on $C_{n+1}$; by composing it with an automorphism of $(C_\omega^2,E)$, we may assume it is a projection, without loss of generality to the first coordinate, on $C_{n+1}$. Let $g_n'\in\Pol(\Gamma)$ be a ternary function which has the properties of $g_n$, but with the roles of $C_n$ and $C_{n+1}$ switched. 
Such a function $g_n'$ can be obtained by composing $g_n$ in all arguments with the same automorphism that switches $C_n$ and $C_{n+1}$.
Then $$g_{n+1}(x,y,z):=g_n(g_n'(x,y,z),g_n'(y,z,x),g_n'(z,x,y))$$ has the desired property.

Next assume that $g_n$ is essential on $C_{n+1}$, and write $g_n'$ for its restriction to $C_{n+1}$. Let $\alpha\in \Aut(C_\omega^2,E)$ flip the two elements of $C_{n+1}$, and fix all other elements of $C_\omega^2$; then the restriction $\alpha'$ of $\alpha$ to $C_{n+1}$ is the only non-trivial permutation of  $C_{n+1}$. By Proposition~\ref{prop:Post}, there exists a term $h'(x,y,z)$ over $\{g_n',\alpha'\}$ which induces either a constant function or the function $x+y+z$ modulo 2 on $C_{n+1}$. The term $h(x,y,z)$ obtained from $h'$ by replacing all occurrences of $\alpha'$ by $\alpha$, and all occurrences of $g_n'$ by $g_n$ induces a ternary function on $C_\omega^2$ whose restriction to $C_{n+1}$ equals $h'$. Since $h$ preserves $E$, it cannot be constant on $C_{n+1}$, and hence it is equal to $x+y+z$ modulo 2 on $C_{n+1}$. For each $0\leq i\leq n$, since $g_n$ equals $x+y+z$ modulo 2 on $C_i$, and since $\alpha$ is the identity on $C_i$, it is easy to see that the term function $h$, restricted to $C_i$, is of the form $\beta'\circ g$, where $\beta'$ is a permutation on $C_i$ and $g$ either equals $x+y+z$ modulo 2 or a projection on $C_i$. Hence, iterating the preceding case we obtain the desired function.
\end{proof}

\begin{proof}[of Proposition~\ref{prop:omega2}]  Suppose that neither of the first two items hold. 
Then by Proposition~\ref{prop:BK}, $\Pol(\Gamma)$ contains a binary function $f$ acting injectively on the classes of $\Eq$; moreover, using Proposition~\ref{prop:Post} and since $E$ is preserved, we see that $\Pol(\Gamma)$ contains a ternary function which equals $x+y+z$ modulo 2 on some equivalence class. Hence, by Lemma~\ref{lem:propagating+} it contains a ternary function $g$ which behaves like $x+y+z$ modulo 2 on all equivalence classes.


Observe first that since $f$ acts injectively on the classes of $\Eq$, we have that whenever $p,q\in (C_\omega^2)^2$ are not equivalent with respect to $\Eq$ in at least one coordinate, then $\Eq(f(p),f(q))$ cannot hold. In other words, we have the type conditions $f(N,\Eq)=f(\Eq,N)=f(N,N)=N$.

We next argue that on each class $C$ the operation $f$ is essentially unary. 
Write $C=\{0,1\}$. Since $E$ is preserved, we have $E(f(0,0),f(1,1))$; similarly, we know that $E(f(0,1),f(1,0))$. Since $f$ moreover preserves $\Eq$, the four values are contained in a single class. Hence either $f(0,1)=f(0,0)$ and $f(1,0)=f(1,1)$, or $f(1,0)=f(0,0)$ and $f(0,1)=f(1,1)$. In the first case, the restriction of $f$ to $C$ only depends on its first argument, and in the second case on its second argument. Assume without loss of generality that the former, i.e., $f(E,=)=E$ and $f(=,E)={=}$, holds on infinitely many equivalence classes $C$. By precomposing $f$ with self-embeddings of $(C_\omega^2,E)$ we may assume that $f$ satisfies these type conditions everywhere. In particular, we then have that $f$ is also canonical as a function from $(C_\omega^2,E)^2$ to $(C_\omega^2,E)$.

The function $q(x,y,z):=f(x,f(y,z))$ satisfies $q(N,\cdot,\cdot)=q(\cdot,N,\cdot)=q(\cdot,\cdot,N)=N$, and $q(P,Q,R)=P$ if $P,Q,R\in\{E,=\}$.

Consider the function $t$ on $(C_\omega^2)^3$ which sends every triple $(x,y,z)$ to the triple $(q(x,y,z),q(y,z,x),q(z,x,y))$. Then, 
whenever $P,Q,R\in\{E,=\}$ and $p,q\in (C_\omega^2)^3$ satisfy ${P}{Q}{R}(p,q)$, then also ${P}{Q}{R}(t(p),t(q))$, by the properties of $q$. Moreover, whenever $p,q\in (C_\omega^2)^3$ are related by $N$ in at least one coordinate, then $\NNN(t(p),t(q))$. By the latter property of $t$, there exist $\alpha,\beta,\gamma\in\overline{\Aut(C_\omega^2,E)}$ such that the function 
$$(\alpha,\beta,\gamma)\circ t(x,y,z):= (\alpha(q(x,y,z)),\beta(q(y,z,x)),\gamma(q(z,x,y)))$$ sends any product $C_i\times C_j\times C_k$ of three equivalence classes into the cube $C^3$ of a single equivalence class; moreover, this function still has the properties of $t$ mentioned above. Set $h(x,y,z):=g\circ (\alpha,\beta,\gamma)\circ t(x,y,z)=
g(\alpha (q(x,y,z)),\beta(q(y,z,x)),\gamma(q(z,x,y)))$. 
Then $h(N,\cdot,\cdot)=h(\cdot,N,\cdot)=h(\cdot,\cdot,N)=g(N,N,N)=N$. 
We claim that $h$ behaves like a minority on $\{E,=\}$. If $P,Q,R\in\{E,=\}$ then $h(P,Q,R)=g(P,Q,R)$. Since $(\alpha,\beta,\gamma)\circ t(x,y,z)$ maps the product of three equivalence classes into the cube of a single equivalence class, and since $g$ behaves like $x+y+z$ modulo 2 on each equivalence class,
the claim follows.
\end{proof}

\section{Polynomial-time tractable CSPs over homogeneous equivalence relations}\label{sect:CSP_equivalence}
We provide two polynomial-time algorithms: the first one is designed for the $\CSP$s of reducts of $(C_2^\omega,E)$ with a ternary injective canonical polymorphism of behaviour minority which is hyperplanely of behaviour balanced xnor (Section~\ref{thm:C-low-2-high-omega-P}), and the second one for reducts of $(C_\omega^2,E)$ with a ternary canonical polymorphism $h$ such that $$h(N,\cdot,\cdot)=h(\cdot,N,\cdot)=h(\cdot,\cdot,N)=N$$ and which behaves like a minority on $\{=,E\}$ (Section~\ref{thm:C-low-omega-high-2-P}).

\subsection{Two infinite classes}
\label{thm:C-low-2-high-omega-P}
We consider the case where  
$\Gamma$ is a reduct of $(C_2^\omega,E)$ which is preserved by a canonical
injection $h$ of behaviour minority which is
hyperplanely of behaviour balanced xnor (cf.~Proposition~\ref{prop:2omega}). Our algorithm for 
CSP$(\Gamma)$ is an adaptation 
of an algorithm for reducts of the random graph~\cite{BodPin-Schaefer}. 

We first reduce CSP$(\Gamma)$ to the CSP of a structure that we call the \emph{injectivization} of $\Gamma$, which  can then be reduced to a tractable CSP over a Boolean domain. 

\begin{definition}\label{def:injective}
	A tuple is called \emph{injective} if all its entries are pairwise distinct.
    A relation is called \emph{injective} if all its tuples are injective. 
    A structure is called \emph{injective} if all its relations are injective. 
\end{definition}

\begin{definition}\label{def:inj}
    We define \emph{injectivizations} for relations, atomic formulas, and structures.
    \begin{itemize}
        \item Let $R$ be any relation. Then the \emph{injectivization of $R$}, denoted by $\inj(R)$, is the (injective) relation consisting of all injective tuples of $R$.
        \item Let $\phi(x_1,\ldots,x_n)$ be an atomic formula in the language of $\Gamma$, where $x_1,\ldots,x_n$ is a list of the variables that appear in $\phi$. Then
        the \emph{injectivization of $\phi(x_1,\dots,x_n)$} is the formula $R^{\inj}_\phi(x_1,\ldots,x_n)$, where $R^{\inj}_\phi$ is a relation symbol which stands for the injectivization of the relation defined by $\phi$.
        \item The \emph{injectivization} of a relational structure $\Gamma$, denoted by $\inj(\Gamma)$, is the relational structure with the same domain as $\Gamma$ whose relations are the injectivizations of the atomic formulas over $\Gamma$, i.e., the relations $R^{\inj}_\phi$.
    \end{itemize}
\end{definition}

To state the reduction to the CSP of an injectivization, we also need the following operations on instances of $\Csp(\Gamma)$.
Here, it will be convenient to view instances of $\Csp(\Gamma)$ as primitive positive $\tau$-sentences.

\begin{definition}
    Let $\Phi$ be an instance of $\Csp(\Gamma)$. Then
    the \emph{injectivization of $\Phi$}, denoted by $\inj(\Phi)$, is the instance $\psi$
    of $\Csp(\inj(\Gamma))$ obtained from $\phi$ by replacing each conjunct
    $\phi(x_1,\dots,x_n)$ of $\Phi$ 
    by $R^{\inj}_\phi(x_1,\ldots,x_n)$.
\end{definition}

We say that a constraint in an instance of $\Csp(\Gamma)$ is \emph{false} if it defines an empty relation in $\Gamma$.
Note that a constraint
 $R(x_1,\dots,x_k)$ might be false
even if the relation
 $R$ is non-empty (simply because some of the variables from $x_1,\dots,x_k$
 might be equal).
The proof of the following statement
is identical to the proof for the
random graph instead of $(C^\omega_2,\Eq)$
in~\cite{BodPin-Schaefer}. 

\begin{proposition}[Lemma 71 in \cite{BodPin-Schaefer}]
\label{prop:inj}
    Let $\Gamma$ be preserved by 
    a binary injection $f$ of behaviour $E$-dominated projection. Then $\Csp(\Gamma)$ can be 
    reduced to $\Csp(\inj(\Gamma))$
     in polynomial time. 
\end{proposition}

We are now in a position to give our reduction.
\begin{proposition}
Let $\Gamma$ be a reduct of $(C_\omega^2, E)$ such that $\End(\Gamma)=\overline{\Aut(C_\omega^2,E)}$ and $\Gamma$ has a ternary injection $f$ which behaves like minority. Further, let $\Delta$ be $(\{0,1\};0,1,\{(x,y,z):z+y+z=1 \bmod 2\})$. There is a polynomial time reduction from $\Csp(\inj(\Gamma))$ to $\Csp(\Delta)$.
\label{prop:bool}
\end{proposition}
\begin{proof}
Note that $f$ preserves $\inj(\Gamma)$ since $f$ is injective. From $f$ one can derive a polymorphism $f'$ on the two-element structure obtained from $\Gamma$ by factoring by the equivalence classes, which behaves like the ternary minimum function on domain $\{0,1\}$.

Take an instance $\phi$ for $\Csp(\inj(\Gamma))$ and build an instance $\phi'$ for $\Csp(\Delta)$ in the following manner. The variable set remains the same and every constraint $(b_1,\ldots,b_k) \in R$ from $\phi$ becomes $(a_1,\ldots, a_k)\in R'$ in $\phi'$ where $b_i \in C_{a_i}$. From Proposition~\ref{prop:Post}, through the presence of $f'$ and the lack of a polymorphism of $\Gamma$ identifying one equivalence class alone, we can assume that the relations of $\phi'$ are preserved by $x+ y + z \bmod 2$, and can thus be taken to be pp-definable in the relation $(x + y + z =1 \bmod 2)$ (see \mbox{e.g.} \cite{Creignou}).

Suppose $\phi$ is a yes-instance of $\Csp(\inj(\Gamma))$, then $\phi'$ is a yes-instance of $\Csp(\Delta)$, by application of the polymorphism $f'$.

Suppose $\phi'$ is a yes-instance of $\Csp(\Delta)$, with solution $f \colon V\rightarrow \{0,1\}$. Then we can build a satisfying assignment for $\phi$ by choosing any injective function from $V$ to $(C_\omega^2, E)$ sending $x \rightarrow C_{f(x)}$.
\end{proof}

\begin{corollary}
   Let $\Gamma$ be a reduct of $(C_2^\omega,E)$ which is preserved by a
    ternary injection $h$ of behaviour minority which is hyperplanely of behaviour balanced xnor. Then 
    $\Csp(\Gamma)$ can be 
    solved in polynomial time.
\end{corollary}
\begin{proof}
Note that the binary function $h(x,y,y)$ is of type $p_1$ and $E$-dominated in the second argument. So the statement
is a consequence of Proposition~\ref{prop:inj} and~\ref{prop:bool}. 
\end{proof}

\subsection{Infinitely many classes of size two}\label{thm:C-low-omega-high-2-P}

We now prove tractability of $\Csp(\Gamma)$ for reducts $\Gamma$ of $(C^2_\omega,\Eq)$ in a finite language 
such that $\Pol(\Gamma)$ contains a ternary canonical function $h$ such that $$h(N,\cdot,\cdot)=h(\cdot,N,\cdot)=h(\cdot,\cdot,N)=N$$ which behaves like a minority on $\{=,E\}$. 

\begin{prop}\label{prop:syntax}
A relation $R$ with a first-order definition
in $(C^2_\omega,\Eq)$ is preserved by $h$
if and only if it 
can be defined by a conjunction of formulas of the form 
\begin{align}\label{eq:one}
 N(x_1,y_1) \vee \cdots \vee N(x_k,y_k) \vee \Eq(z_1,z_2)
 \end{align}
for $k \geq 0$, or of the form 
 \begin{align}
N(x_1,y_1) \vee \cdots \vee N(x_k,y_k) \, \vee & \, (|\{i \in S : x_i \neq y_i\}| \equiv_2 p)  \label{eq:two}
 \end{align}
 where $p \in \{0,1\}$ and $S \subseteq \{1,\dots,k\}$.
\end{prop}

The proof is inspired from a proof for tractable phylogeny constraints~\cite{Phylo-Complexity}. 




\begin{proof}
For the backwards implication, it
suffices to verify that formulas 
of the form in the statement are preserved
by $h$. Let $o,p,q \in R$, and
let $r := h(o,p,q)$. Assume that
$R$ has a definition by a formula
$\phi$ of the form as described in the statement. Suppose for contradiction that
$r$ does not satisfy $\phi$.
For any conjunct of $\phi$ violated by $r$, of the form
$N(x_1,y_1) \vee \dots \vee N(x_k,y_k) \vee \theta$, the tuple
$r$ must therefore satisfy
 $\Eq(x_1,y_1) \wedge \cdots \wedge \Eq(x_k,y_k)$. 
Since $h$ has the property
that $h(N,\cdot,\cdot) = h(\cdot,N,\cdot) = h(\cdot,\cdot,N) = N$, this means
that each of $o$, $p$, and $q$ 
also satisfies this formula. This in 
turn implies that $o$, $p$, and $q$
must satisfy the formula $\theta$.
It suffices to prove that $r$ satisfies $\theta$, too, since
this contradicts the assumption that
$r$ does not satisfy $\phi$. 
Suppose first that $\theta$
is of the form $\Eq(z_1,z_2)$. 
In this case, $r$ must also satisfy $\Eq(z_1,z_2)$ since $h$ preserves $\Eq$. So assume that $\theta$ is of the form 
$|\{i \in S: x_i \neq y_i\}| \equiv_2 p$ for $S \subseteq \{1,\dots,k\}$ and $p \in \{0,1\}$. Since each of $o$, $p$, $q$ 
satisfies this formula and $h$ behaves like a minority on $\{E,=\}$, we have 
that $r$ satisfies this formula, too. 
 
For the forwards implication, 
let $R$ be an $n$-ary relation with a first-order definition in $(C^2_\omega,\Eq)$ that is preserved by $h$. Define $\sim$ to be the equivalence relation on $(C^2_\omega)^n$ where $a \sim b$ iff $\Eq(a_i,a_j) \Leftrightarrow \Eq(b_i,b_j)$ for all $i,j \leq n$. Note that $h$ preserves $\sim$. 
For $a \in (C^2_\omega)^n$, let $R_a$ be the relation 
that contains all $t \in R$ with $t \sim a$. Let $\psi_a$ be the formula
$$\bigwedge_{i < j \leq n, \Eq(a_i,a_j)} \Eq(x_i,x_j)$$ 
and $\psi_a'$ be the formula
$$\bigwedge_{i < j \leq n, N(a_i,a_j)} N(x_i,x_j) \, .$$ 
Note that $t \in (C^2_\omega)^n$ satisfies 
$\psi_a \wedge \psi_a'$ if and only if
$t \sim a$, and hence
a tuple from
$R$ 
is in $R_a$ if and only it satisfies $\psi_a \wedge \psi_a'$. 

Pick representatives $a_1,\dots,a_m$ for all orbits of $n$-tuples in $R$. 

\medskip \noindent
{\bf Claim 1.} $\bigvee_{i \leq m} (\psi_{a_i} \wedge \psi'_{a_i})$
is equivalent to a conjunction of formulas of the form $(\ref{eq:one})$ from the statement. 

Rewrite the formula into an equivalent formula $\psi_0$
in conjunctive normal form of minimal size where every literal is either of the form $\Eq(x,y)$ or of the form $N(x,y)$. Suppose that $\psi_0$ contains a conjunct with literals $\Eq(a,b)$ and $\Eq(c,d)$.
Since $\psi_0$ is of minimal size there
exists $r \in (C^2_\omega)^n$ that satisfies
$\Eq(a,b)$ and none of the other literals in the conjunct, and similarly there exists $s \in (C^2_\omega)^n$ that
satisfies $\Eq(c,d)$ and none of the other literals. By assumption, $r \sim r' \in R$ and $s \sim s' \in R$. 
Since $R$ is preserved by $h$,
we have $t' := h(r',s',s') \in R$.
Then $t \sim t'$ since $h$ preserves $\sim$, and hence $t$ satisfies $\psi_0$. 
But $t$ satisfies none of the literals in the conjunct, a contradiction. 
Hence, all conjuncts of $\psi_0$ have
form $(\ref{eq:one})$ from the statement.

Let $t \in (C^2_\omega)^n$, set $l := {n \choose 2}$, and let $i_1j_1,\dots,i_lj_l$ be an enumeration of ${\{1,\dots,n\} \choose 2}$. The tuple
$b \in \{0,1\}^{n \choose 2}$
with $b_s = 1$ if $t_{i_s} \neq t_{j_s}$
and $b_s = 0$ otherwise is called
the \emph{split vector} of $t$. 
We associate to $R_a$ the Boolean relation $B_a$ consisting of all split
vectors of tuples in $R_a$. 
Since $R$ and $R_a$ are preserved by $h$, the relation $B_a$ is preserved by
the Boolean minority operation, and
hence has a definition by a Boolean system of equations. Therefore,
there exists a conjunction $\theta_a$
of equations of the form
$|\{s \in S : x_{i_s} = y_{j_s}\}| \equiv_2 p$, $p \in \{0,1\}$ such that $\theta_a \wedge \psi_a \wedge \psi_a'$ defines $R_a$. 

\medskip
\noindent
{\bf Claim 2.} The following formula 
$\phi$ defines $R$:
$$\phi := \psi_0 \wedge \bigwedge_{a \in \{a_1,\dots,a_m\}} (\neg \psi_a \vee \theta_a)$$
It is straightforward to see that this
formula can be rewritten into a formula
of the form as required in the statement. 

To prove the claim, we first show
 that every $t \in R$ satisfies $\phi$. 
 Clearly, $t$ satisfies $\psi_0$. 
 Let $a \in \{a_1,\dots,a_m\}$ be arbitrary; we have to verify that $t$ satisfies $\neg \psi_a \vee \theta_a$. If there are indices $i,j \in \{1,\dots,n\}$ such that $N(t_i,t_j)$
and $\Eq(a_i,a_j)$, then $t$ satisfies $\neg \psi_a$. We are left with the case that for all $i,j \in \{1,\dots,n\}$ if $\Eq(a_i,a_j)$ then $\Eq(t_i,t_j)$. 
In order to show that $t$ satisfies
$\theta_a$, it suffices to show that
there exists a $t' \in R_a$ such 
that
for all $i,j \leq n$ with $\Eq(a_i,a_j)$
we have $t_i = t_j$ iff $t'_i = t'_j$. 
Note that $t' := h(a,a,t) \sim a$ 
since $h(N,\cdot,\cdot) = h(\cdot,N,\cdot) = h(\cdot,\cdot,N) = h$.
Moreover, $t' \in R$ and thus $t' \in R_a$. Finally, for all $i,j \leq n$ with $\Eq(a_i,a_j)$
we have $t_i = t_j$ iff $t'_i = t'_j$
because $h$ behaves as a minority on $\{E,=\}$. Hence, $t$ satisfies $\phi$.
 
 \medskip
Next, we show that every tuple
$t$ that satisfies $\phi$ is in $R$. 
Since $t$ satisfies $\psi_0$ we have that $t \sim a$
for some $a \in \{a_1,\dots,a_m\}$.
Thus, 
$t \models \psi_a \wedge \psi_a'$.
By assumption, $t$ satisfies 
$\neg \psi_a \vee \theta_a$
 and hence
$t \models \theta_a$. Therefore,
$t \in R_a$ and  in particular $t \in R$. 
\end{proof}

\begin{prop}\label{prop:algorithm}
There is a polynomial-time algorithm
that decides whether a given set 
 $\Phi$ 
of formulas as in the statement of
Proposition~\ref{prop:syntax}
is satisfiable.
\end{prop}
\begin{proof}
Let $X$ be the set of variables that appear in $\Phi$. 
Create 
a graph $G$ with vertex set $X$ 
that contains an edge
between $z_1$ and $z_2$ 
if 
$\Phi$ contains a formula
of the form $\Eq(z_1,z_2)$. 
Eliminate all literals of the form
$N(x_i,y_i)$ in formulas from $\Phi$ when $x_i$ and $y_i$ lie in the same connected component of $G$. 
Repeat this procedure until no
more literals get removed. 

We then create a Boolean system of equations 
$\Psi$ 
with variable set ${X \choose 2}$ as follows;
if $x,y \in X$ are distinct, for better readability we write $xy$ for the respective Boolean variable instead of $\{x,y\}$.
For each formula
$|\{i \in S \mid x_i \neq y_i \}| \equiv_2 p$
we add the Boolean equation
$\sum_{i \in S} x_iy_i = p$.
We additionally add for all
 $xy,yz,xz \in {X \choose 2}$ the equation $xy+yz = xz$. 
If the resulting system of equations $\Psi$ does not have
a solution 
over $\{0,1\}$, reject the instance. 
Otherwise accept. 

\medskip
To see that this algorithm is correct,
observe that the literals that have been removed in the first part of the algorithm are false in all solutions, 
so removing them from the disjunctions does not change the set of solutions. 

\medskip
If the algorithm rejects,
then there is indeed no solution to 
$\Phi$. 
To see this, suppose that $s \colon C^2_\omega \to C^2_\omega$ is a solution to $\Phi$. 
Define
$b \colon {X \choose 2} \to \{0,1\}$ as follows. 
Note that for every variable $x_iy_i$
that appears in some Boolean equation in $\Psi$, a literal $N(x_i,y_i)$ has been deleted in the first phase of the algorithm (recall the syntactic form in~(\ref{eq:two}); we only add Boolean equations to $\Psi$ if all the literals involving $N$ have been deleted), and hence we have $\Eq(s(x_i),s(y_i))$. Define
 $s'(x_iy_i) := 1$ if $s(x_i) \neq s(y_i)$
 and $s'(x_iy_i) := 0$ otherwise. 
Then $s'$ is a satisfying assignment for $\Psi$.

\medskip
We still have to show that there exists
a solution to $\Phi$ if the algorithm
accepts. Let $s' \colon {X \choose 2} \to \{0,1\}$ be a solution to $\Psi$. 
For each connected component
$C$ in the graph $G$ at the final stage of the algorithm we pick two values
$a_C,b_C \in C^2_\omega$ such that $\Eq(a_C,b_C)$, and such that
$N(a_C,d)$ and $N(b_C,d)$ for all
previously picked values $d \in C^2_\omega$.
Moreover, for each connected component $C$ of $G$ we pick a representative $r_C$. 
Define $s(r_C) := a_C$, 
and for $x \in C$ define $s(x) := a_C$
if $s'(xr_C) = 0$,
and $s(x) := b_C$ otherwise.  

Then $s$ satisfies all formulas in $\Psi$
that still contain disjuncts of the form
$N(x_i,y_i)$, since these disjuncts are satisfied by $s$. 
Formulas of the form
$|\{i \in S : x_i \neq y_i\}| \equiv_2 p$
are satisfied, too, since $x_i$ and
$y_i$ lie in the same connected component $C$, 
and hence $s(x_i) \neq s(y_i)$ iff
$s'(xr_C) \neq s'(y_ir_C)$,
which is the case iff
$s'(xr_C) + s'(y_ir_C) = s'(x_iy_i) = 1$ because of the additional equations we have added to $\Psi$. Therefore,
$|\{i \in S : x_i = y_i\}| \equiv_2 p$
iff $\sum_{i \in S} s'(x_iy_i) = p$. 
\end{proof}

\begin{corollary}\label{cor:tractability}
Let $\Gamma$ be a reduct of $(C^2_\omega,\Eq)$ with finite signature and such that
$\Pol(\Gamma)$ contains a ternary canonical injection $h$ as described in the beginning of Section~\ref{thm:C-low-omega-high-2-P}. Then $\Csp(\Gamma)$ is in P. 
\end{corollary}
\begin{proof} 
Direct consequence of Proposition~\ref{prop:syntax} and Proposition~\ref{prop:algorithm}. 
\end{proof} 

\section{Summary for the homogeneous equivalence relations}\label{sect:summary_equivalence} \

\begin{thm}
\label{thm:equivalence-above-2}
Let $\Gamma$ be a finite signature reduct of $(C_n^s,E)$, where either $2< n<\omega$ or  $2< s<\omega$, and either $n$ or $s$ equals $\omega$. Then one of the following holds.
\begin{itemize}
\item[(1)] $\Gamma$ is homomorphically equivalent to a reduct of $(C_n^s,=)$, and $\CSP(\Gamma)$ is in P or NP-complete by~\cite{ecsps}.
\item[(2)] $\End(\Gamma)=\overline{\Aut(C_n^s,E)}$, $\Pol(\Gamma)$ has a uniformly continuous h1 clone homomorphism, and $\CSP(\Gamma)$ is NP-complete.
\end{itemize}
\end{thm} 
\begin{proof}
If $\Gamma$ has an endomorphism whose image is a clique or an independent set, then $\Gamma$ is homomorphically equivalent to a reduct of $(C_n^s,=)$ and the complexity classification is known from~\cite{ecsps}. Otherwise, courtesy of Propositions~\ref{prop:EndEq1} and~\ref{prop:eqpreserved}, we may assume that $\End(\Gamma)=\overline{\Aut(C_n^s,E)}$, and that there is a pp-definition of $E$, $N$, and $\Eq$ in $\Gamma$.

In the first case, that $\Eq$ has a finite number $n\geq 3$ of classes, we use Proposition~\ref{prop:n>2} to see that the action of $\Pol(\Gamma)$ on the classes of $\Eq$ has no essential and no constant operation. It follows that this action has a uniformly continuous projective clone homomorphism as in Definition~\ref{defn:clonehomo}. The mapping which sends every function in $\Pol(\Gamma)$ to the function it becomes in the action on the classes of $\Eq$ is a uniformly continuous clone homomorphism~\cite{Topo-Birk}, and hence the original action of $\Pol(\Gamma)$ has a uniformly continuous projective clone homomorphism as well. This implies NP-completeness of $\CSP(\Gamma)$ (Theorem~\ref{thm:wonderland}).

In the second case, that $\Eq$ has classes of finite size $s \geq 3$, we use Proposition~\ref{prop:s>2} to see that the action of $\Pol(\Gamma,C)$ on some equivalence class $C$ has no essential and no constant operation, and hence has a uniformly continuous projective clone homomorphism. Picking any $c\in C$, we have that $\Pol(\Gamma,c)\subseteq \Pol(\Gamma,C)$ since $C$ is pp-definable from $c$ and $\Eq$. Consequently, $\Pol(\Gamma,c)$ has a uniformly continuous projective clone homomorphism as well. Because $\Gamma$ is a model-complete core, this implies that $\Pol(\Gamma)$ has a uniformly continuous projective h1 clone homomorphism~\cite{wonderland}, and hence $\CSP(\Gamma)$ is NP-complete by Theorem~\ref{thm:wonderland}.
\end{proof}

\begin{thm}
\label{thm:C-low-2-high-omega-complexity}
Suppose $\Gamma$ is a finite signature reduct of $(C_2^\omega,E)$. Then one of the following holds.
\begin{itemize}
\item[(1)] $\Gamma$ is homomorphically equivalent to a reduct of $(C_2^\omega,=)$, and $\CSP(\Gamma)$ is in P or NP-complete by~\cite{ecsps}.
\item[(2)] $\End(\Gamma)=\overline{\Aut(C_2^\omega,E)}$, $\Pol(\Gamma)$ contains a canonical ternary injection of behaviour minority which is hyperplanely of behaviour  balanced xnor, and $\CSP(\Gamma)$ is in P.
\item[(3)] $\End(\Gamma)=\overline{\Aut(C_2^\omega,E)}$, $\Pol(\Gamma)$ has a uniformly continuous h1 clone homomorphism, and $\CSP(\Gamma)$ is NP-complete.
\end{itemize}
\end{thm}
\begin{proof}
As in the proof of Theorem~\ref{thm:equivalence-above-2} we may assume that $\End(\Gamma)=\overline{\Aut(C_2^\omega,E)}$, and that $E$, $N$ and $\Eq$ are pp-definable. We apply Proposition~\ref{prop:2omega}. The first two cases from that proposition imply a uniformly continuous projective h1 clone homomorphism, and hence NP-completeness of the CSP, as in the proof of Theorem~\ref{thm:equivalence-above-2}. The third case in Proposition~\ref{prop:2omega} yields case~(2) here, and tractability as detailed in Section~\ref{thm:C-low-2-high-omega-P}.
\end{proof}

\begin{thm}
\label{thm:C-low-omega-high-2-complexity}
Suppose $\Gamma$ is a finite signature reduct of $(C_\omega^2,E)$. Then one of the following holds.
\begin{itemize}
\item[(1)] $\Gamma$ is homomorphically equivalent to a reduct of $(C_\omega^2,=)$, and $\CSP(\Gamma)$ is in P or NP-complete by~\cite{ecsps}.
\item[(2)] $\End(\Gamma)=\overline{\Aut(C_\omega^2,E)}$, $\Eq$ is not pp-definable, $\Pol(\Gamma)$ contains a canonical binary injective polymorphism of behaviour $\mini$ that is $N$-dominated, and $\CSP(\Gamma)$ is in P.
\item[(3)] $\End(\Gamma)=\overline{\Aut(C_\omega^2,E)}$, $\Eq$ is pp-definable, $\Pol(\Gamma)$ contains a  ternary canonical function $h$ with $h(N,\cdot,\cdot)=h(\cdot,N,\cdot)=h(\cdot,\cdot,N)=N$ and which behaves like a minority on $\{E,=\}$, and $\CSP(\Gamma)$ is in P.
\item[(4)] $\End(\Gamma)=\overline{\Aut(C_\omega^2,E)}$, $\Pol(\Gamma)$ has a uniformly continuous h1 clone homomorphism, and $\CSP(\Gamma)$ is NP-complete.
\end{itemize}
\end{thm}
\begin{proof}
As in Theorem~\ref{thm:equivalence-above-2}  we may assume that $\End(\Gamma)=\overline{\Aut(C_\omega^2,E)}$, and that therefore $E$ and $N$ are pp-definable. If $\Eq$ is not pp-definable, then by Proposition~\ref{prop:al-jabr}, we have a binary injective polymorphism of behaviour $\mini$ that is $N$-dominated, and we have a polynomial algorithm from Theorem~\ref{thm:maximal}, similarly as in Proposition~\ref{prop:mintractable} for reducts of $(H_n,E)$. Suppose now that $\Eq$ is pp-definable.
We apply Proposition~\ref{prop:omega2}. As before, the first two cases imply NP-completeness of $\CSP(\Gamma)$. The third case from Proposition~\ref{prop:omega2} yields tractability as detailed in Section~\ref{thm:C-low-omega-high-2-P}.
\end{proof}

Summarizing, we obtain a proof of Theorem~\ref{thm:equiv}.

\begin{proof}[of Theorem~\ref{thm:equiv}]
The statement follows from the preceding three theorems, together with~\cite{equiv-csps} (for $C_\omega^\omega$) and~\cite{ecsps} (for $C_\omega^1$ and $C_1^\omega$). 
\end{proof}

\ignore{
We close the section with a more detailed variant of Theorem~\ref{thm:equiv}. 

\begin{theorem}\label{thm:main2equiv}
Let  $(C_n^s,E)$ be an infinite graph whose reflexive closure $\Eq$ is an equivalence relation with $n$ classes of size $s$, where $1\leq n, s \leq \omega$. 
Let $\Gamma$ be a reduct of $(C_n^s,E)$. 
Then one of the following holds.
\begin{itemize}
\item[(1)] $\Gamma$ has an endomorphism whose image induces a clique or an independent set, and is homomorphically equivalent to a reduct of $(C_n^s,=)$.
\item[(2)] $\Gamma$ is a model complete core, $\End(\Gamma)=\overline{\Aut(C_n^s,E)}$, and $\Pol(\Gamma)$ has a uniformly continuous projective h1 clone homomorphism.
\item[(3)] $n=2, s=\omega$, $\Gamma$ is a model complete core, and $\Pol(\Gamma)$ contains a canonical ternary injection of behaviour minority which is hyperplanely of behaviour  E-dominated projection. 
\item[(4)] $n=\omega, s=2$, $\Gamma$ is a model complete core, and $\Pol(\Gamma)$ contains a ternary canonical function $h$ with $h(N,\cdot,\cdot)=h(\cdot,N,\cdot)=h(\cdot,\cdot,N)=N$ and which behaves like a minority on $\{E,=\}$. 
\end{itemize}
Neither items~(2) and~(3), nor items~(2) and~(4) can simultaneously hold, and when $\Gamma$ has a finite relational signature, then $(2)$ implies NP-completeness and both (3) and (4) imply tractability of its CSP.
\end{theorem}
}

\section{Outlook}\label{sect:final}



We have classified the computational
complexity of CSPs 
for reducts of the infinite homogeneous graphs. 
Our proof shows that the scope 
of the classification method from~\cite{BodPin-Schaefer} is much larger
than one might expect at first sight. 
The general research goal here is
to identify larger and larger classes
of infinite-domain CSPs where systematic
complexity classification is possible; two dichotomy conjectures are given for CSPs of reducts of finitely bounded homogeneous structures in \cite{BPP-projective-homomorphisms} and~\cite{wonderland}, where these have now been proved equivalent in~\cite{TwoDichotomyConjectures}. We have given additional evidence for these conjectures by proving that they hold for all reducts of homogeneous graphs. The next step in this direction
might be to 
show a general complexity dichotomy
for reducts of homogeneous structures whose
age is finitely bounded 
and has the \emph{free amalgamation property} (the Henson graphs provide natural examples for such structures).
The present paper
indicates that this problem might be within reach.

\section*{Acknowledgements}

We are most grateful to our reviewers, whose patience and diligence has uncovered errors as well as encouraged us to make the paper significantly more readable.


\newcommand{\etalchar}[1]{$^{#1}$}
\def\cprime{$'$} \def\cprime{$'$}
\begin{thebibliography}{BCKvO09}

\bibitem[BCKvO09]{Maximal}
Manuel Bodirsky, Hubie Chen, Jan K\'ara, and Timo von Oertzen.
\newblock Maximal infinite-valued constraint languages.
\newblock {\em Theoretical Computer Science (TCS)}, 410:1684--1693, 2009.
\newblock A preliminary version appeared at ICALP'07.

\bibitem[BJP17]{Phylo-Complexity}
Manuel Bodirsky, Peter Jonsson, and Trung~Van Pham.
\newblock {The Complexity of Phylogeny Constraint Satisfaction Problems}.
\newblock {\em ACM Transactions on Computational Logic (TOCL)}, 18(3), 2017.
\newblock An extended abstract appeared in the conference STACS 2016.

\bibitem[BK08]{ecsps}
Manuel Bodirsky and Jan K\'ara.
\newblock The complexity of equality constraint languages.
\newblock {\em Theory of Computing Systems}, 3(2):136--158, 2008.
\newblock A conference version appeared in the proceedings of Computer Science
  Russia {(CSR'06)}.

\bibitem[BK09]{tcsps-journal}
Manuel Bodirsky and Jan K\'ara.
\newblock The complexity of temporal constraint satisfaction problems.
\newblock {\em Journal of the ACM}, 57(2):1--41, 2009.
\newblock An extended abstract appeared in the Proceedings of the Symposium on
  Theory of Computing (STOC).

\bibitem[BKJ05]{JBK}
Andrei~A. Bulatov, Andrei~A. Krokhin, and Peter~G. Jeavons.
\newblock Classifying the complexity of constraints using finite algebras.
\newblock {\em SIAM Journal on Computing}, 34:720--742, 2005.

\bibitem[BKN09]{BartoKozikNiven}
Libor Barto, Marcin Kozik, and Todd Niven.
\newblock The {CSP} dichotomy holds for digraphs with no sources and no sinks
  (a positive answer to a conjecture of {B}ang-{J}ensen and {H}ell).
\newblock {\em SIAM Journal on Computing}, 38(5), 2009.

\bibitem[BKO{\etalchar{+}}17]{TwoDichotomyConjectures}
Libor Barto, Michael Kompatscher, Miroslav Ol\v{s}\'{a}k, Michael Pinsker, and
  Trung~Van Pham.
\newblock The equivalence of two dichotomy conjectures for infinite domain
  constraint satisfaction problems.
\newblock In {\em Proceedings of the 32nd Annual {ACM/IEEE} Symposium on Logic
  in Computer Science -- LICS'17}, 2017.
\newblock Preprint arXiv:1612.07551.

\bibitem[BM16]{Bodirsky-Mottet}
Manuel Bodirsky and Antoine Mottet.
\newblock Reducts of finitely bounded homogeneous structures, and lifting
  tractability from finite-domain constraint satisfaction.
\newblock In {\em Proceedings of the 31th Annual IEEE Symposium on Logic in
  Computer Science -- LICS'16}, pages 623--632, 2016.
\newblock Preprint available at ArXiv:1601.04520.

\bibitem[BMM18]{dCSPs2}
Manuel Bodirsky, Barnaby Martin, and Antoine Mottet.
\newblock Discrete temporal constraint satisfaction problems.
\newblock {\em J. {ACM}}, 65(2):9:1--9:41, 2018.

\bibitem[BN06]{BodirskyNesetrilJLC}
Manuel Bodirsky and Jaroslav Ne\v{s}et\v{r}il.
\newblock Constraint satisfaction with countable homogeneous templates.
\newblock {\em Journal of Logic and Computation}, 16(3):359--373, 2006.

\bibitem[Bod07]{Cores-journal}
Manuel Bodirsky.
\newblock Cores of countably categorical structures.
\newblock {\em Logical Methods in Computer Science}, 3(1):1--16, 2007.

\bibitem[Bod12]{Bodirsky-HDR}
Manuel Bodirsky.
\newblock Complexity classification in infinite-domain constraint satisfaction.
\newblock M\'emoire d'habilitation \`a diriger des recherches, Universit\'{e}
  Diderot -- Paris 7. Available at arXiv:1201.0856, 2012.

\bibitem[BP11]{BP-reductsRamsey}
Manuel Bodirsky and Michael Pinsker.
\newblock Reducts of {R}amsey structures.
\newblock {\em AMS Contemporary Mathematics, vol. 558 (Model Theoretic Methods
  in Finite Combinatorics)}, pages 489--519, 2011.

\bibitem[BP14]{RandomMinOps}
Manuel Bodirsky and Michael Pinsker.
\newblock Minimal functions on the random graph.
\newblock {\em Israel Journal of Mathematics}, 200(1):251--296, 2014.

\bibitem[BP15a]{BodPin-Schaefer}
Manuel Bodirsky and Michael Pinsker.
\newblock Schaefer's theorem for graphs.
\newblock {\em Journal of the ACM}, 62(3):52 pages (article number 19), 2015.
\newblock A conference version appeared in the Proceedings of STOC 2011, pages
  655--664.

\bibitem[BP15b]{Topo-Birk}
Manuel Bodirsky and Michael Pinsker.
\newblock Topological {B}irkhoff.
\newblock {\em Transactions of the American Mathematical Society},
  367:2527--2549, 2015.

\bibitem[BP16a]{wonderland}
Libor Barto and Michael Pinsker.
\newblock The algebraic dichotomy conjecture for infinite domain constraint
  satisfaction problems.
\newblock In {\em Proceedings of the 31th Annual IEEE Symposium on Logic in
  Computer Science -- LICS'16}, pages 615--622, 2016.
\newblock Preprint arXiv:1602.04353.

\bibitem[BP16b]{BodPin-CanonicalFunctions}
Manuel Bodirsky and Michael Pinsker.
\newblock {C}anonical functions: a proof via topological dynamics.
\newblock Preprint arXiv:1610.09660, 2016.

\bibitem[BP18]{Topo}
Libor Barto and Michael Pinsker.
\newblock Topology is irrelevant.
\newblock Preprint available from the authors' websites, 2018.

\bibitem[BPP14]{BPP-projective-homomorphisms}
Manuel Bodirsky, Michael Pinsker, and Andr\'{a}s Pongr\'acz.
\newblock Projective clone homomorphisms.
\newblock Accepted for publication in the Journal of Symbolic Logic, Preprint
  arXiv:1409.4601, 2014.

\bibitem[BPT13]{BPT-decidability-of-definability}
Manuel Bodirsky, Michael Pinsker, and Todor Tsankov.
\newblock Decidability of definability.
\newblock {\em Journal of Symbolic Logic}, 78(4):1036--1054, 2013.
\newblock A conference version appeared in the Proceedings of LICS 2011.

\bibitem[Bul06]{Bulatov}
Andrei~A. Bulatov.
\newblock A dichotomy theorem for constraint satisfaction problems on a
  3-element set.
\newblock {\em Journal of the ACM}, 53(1):66--120, 2006.

\bibitem[Bul17]{BulatovFVConjecture}
Andrei~A. Bulatov.
\newblock A dichotomy theorem for nonuniform {CSP}s.
\newblock In {\em 58th {IEEE} Annual Symposium on Foundations of Computer
  Science, {FOCS} 2017, Berkeley, CA, USA, October 15-17, 2017}, pages
  319--330, 2017.
\newblock Avaliable at {arXiv:1703.03021}.

\bibitem[BW12]{equiv-csps}
Manuel Bodirsky and Micha\l\ Wrona.
\newblock Equivalence constraint satisfaction problems.
\newblock In {\em Proceedings of Computer Science Logic}, volume~16 of {\em
  LIPICS}, pages 122--136. Dagstuhl Publishing, September 2012.

\bibitem[CKS01]{Creignou}
Nadia Creignou, Sanjeev Khanna, and Madhu Sudan.
\newblock {\em Complexity Classifications of Boolean Constraint Satisfaction
  Problems}.
\newblock SIAM Monographs on Discrete Mathematics and Applications 7, 2001.

\bibitem[FV99]{FederVardi}
Tom\'as Feder and Moshe~Y. Vardi.
\newblock The computational structure of monotone monadic {SNP} and constraint
  satisfaction: {a} study through {D}atalog and group theory.
\newblock {\em {SIAM} Journal on Computing}, 28:57--104, 1999.

\bibitem[GP08]{GoldsternPinsker}
Martin Goldstern and Michael Pinsker.
\newblock A survey of clones on infinite sets.
\newblock {\em Algebra Universalis}, 59:365--403, 2008.

\bibitem[HN90]{HellNesetril}
Pavol Hell and Jaroslav Ne\v{s}et\v{r}il.
\newblock On the complexity of {H}-coloring.
\newblock {\em Journal of Combinatorial Theory, Series B}, 48:92--110, 1990.

\bibitem[HR94]{HaddadRosenberg}
Lucien Haddad and Ivo~G. Rosenberg.
\newblock Finite clones containing all permutations.
\newblock {\em Canadian Journal of Mathematics}, 46(5):951--970, 1994.

\bibitem[JCG97]{JeavonsClosure}
Peter Jeavons, David Cohen, and Marc Gyssens.
\newblock Closure properties of constraints.
\newblock {\em Journal of the ACM}, 44(4):527--548, 1997.

\bibitem[KPT05]{Topo-Dynamics}
Alexander Kechris, Vladimir Pestov, and Stevo Todorcevic.
\newblock Fraiss\'e limits, {R}amsey theory, and topological dynamics of
  automorphism groups.
\newblock {\em Geometric and Functional Analysis}, 15(1):106--189, 2005.

\bibitem[LW80]{LachlanWoodrow}
Alistair~H. Lachlan and Robert~E. Woodrow.
\newblock Countable ultrahomogeneous undirected graphs.
\newblock {\em Transactions of the AMS}, 262(1):51--94, 1980.

\bibitem[NR89]{NesetrilRoedlPartite}
Jaroslav Ne\v{s}et\v{r}il and Vojt\v{e}ch R\"odl.
\newblock The partite construction and {R}amsey set systems.
\newblock {\em Discrete Mathematics}, 75(1-3):327--334, 1989.

\bibitem[Pon17]{Pon11}
Andr\'{a}s Pongr\'{a}cz.
\newblock Reducts of the {H}enson graphs with a constant.
\newblock {\em Annals of Pure and Applied Logic}, 168(7):1472--1489, 2017.

\bibitem[Pos41]{Post}
Emil~L. Post.
\newblock The two-valued iterative systems of mathematical logic.
\newblock {\em Annals of Mathematics Studies}, 5, 1941.

\bibitem[Sch78]{Schaefer}
Thomas~J. Schaefer.
\newblock The complexity of satisfiability problems.
\newblock In {\em Proceedings of the Symposium on Theory of Computing (STOC)},
  pages 216--226, 1978.

\bibitem[Sze86]{Szendrei}
\'Agnes Szendrei.
\newblock {\em Clones in universal algebra}.
\newblock S\'eminaire de Math\'ematiques Sup\'erieures. Les Presses de
  l'Universit\'e de {M}ontr\'eal, 1986.

\bibitem[Tho91]{RandomReducts}
Simon Thomas.
\newblock Reducts of the random graph.
\newblock {\em Journal of Symbolic Logic}, 56(1):176--181, 1991.

\bibitem[TZ12]{Tent-Ziegler}
Katrin Tent and Martin Ziegler.
\newblock {\em A course in model theory}.
\newblock Lecture Notes in Logic. Cambridge University Press, 2012.

\bibitem[Zhu17]{ZhukFVConjecture}
Dmitriy Zhuk.
\newblock A proof of {CSP} dichotomy conjecture.
\newblock In {\em 58th {IEEE} Annual Symposium on Foundations of Computer
  Science, {FOCS} 2017, Berkeley, CA, USA, October 15-17, 2017}, pages
  331--342, 2017.
\newblock Avaliable at {arXiv:1704.01914}.

\end{thebibliography}

\newcommand{\etalchar}[1]{$^{#1}$}
\def\cprime{$'$} \def\cprime{$'$}

\end{document}